%% file: Consistent.tex
\documentclass{article}

\input{preamble}

\title{Load Balancing with Dynamic Set of Balls and Bins}
\author{Anders Aamand\footnote{Basic Algorithms Research Copenhagen (BARC), University of Copenhagen.} \and Jakob B\ae k Tejs Knudsen$^*$ \and Mikkel Thorup$^*$ }
\date{\today}

\begin{document}
\maketitle

\begin{abstract}
  In dynamic load balancing, we wish to distribute balls into
  bins in an environment where both balls and bins can be added and removed.
  We want to minimize the maximum load of any bin but we also want to minimize
  the number of balls and bins that are affected when
  adding or removing a ball or a bin. We want a hashing-style solution
  where we given the ID of a ball can find its bin efficiently.

  We are given a user-specified balancing parameter $c=1+\eps$, where
  $\eps\in (0,1)$. Let $\balls$ and $\bins$ be the current number of
  balls and bins. Then we want no bin with load above
  $C=\ceil{c\balls/\bins}$, referred to as the {\em capacity} of the
  bins.

  We present a scheme where we can locate a ball checking
  $1+O(\log 1/\eps)$ bins in expectation.
  When inserting or deleting a ball, we expect to move
  $O(1/\eps)$ balls, and when inserting or deleting a bin,
  we expect to move $O(C/\eps)$ balls. Previous bounds were off by a factor $1/\eps$. 

  The above bounds are best possible when $C=O(1)$ but for larger $C$, we can do much better: Let
  \[f=\left\{\begin{array}{ll}
  \eps C&\textnormal{ if } C\leq \log 1/\eps \\
  \eps\sqrt{C}\cdot \sqrt{\log(1/(\eps\sqrt{C}))}&\textnormal{ if } \log 1/\eps\leq C<\tfrac{1}{2\eps^2}\\ 1& \textnormal{ if } C\geq \tfrac{1}{2\eps^2}\end{array}\right.\]
  We show that we expect to move $O(1/f)$ balls when
  inserting or deleting a ball, and $O(C/f)$ balls when
  inserting or deleting a bin. Moreover, when $C\geq \log 1/\eps$, we can search a ball checking only $O(1)$ bins in expectation.  
  
  For the bounds with larger $C$,  we first have to resolve a much simpler
  probabilistic problem. Place $\balls$ balls in $\bins$ bins of
  capacity $C$, one ball at the time. Each ball picks a uniformly
  random non-full bin. We show that in expectation and with high probability, the fraction of non-full bins is $\Theta(f)$. Then the expected number of bins that a new ball would have to visit to find one that is not full is $\Theta(1/f)$. As it turns out, this is also the complexity of an insertion in our more complicated scheme where both balls and bins can be added and removed.  
\end{abstract}
\setcounter{page}0
\newpage
  \section{Introduction}\label{sec:intro}
Load balancing in dynamic environments is a central problem in
designing several networking systems and web
services~\cite{chord,chord-theory}. We wish to allocate {\em
  clients\/} (also referred to as {\em balls}) to {\em servers\/}
(also referred to as {\em bins}) in such a way that none of the
servers gets overloaded. Here, the {\em load\/} of a server is the
number of clients allocated to it. We want a hashing-style solution
where we given the ID of a client can efficiently find its server.
Both clients and servers
may be added or removed in any order, and with such changes, we do not
want to move too many clients. Thus, while the dynamic allocation
algorithm has to always ensure a proper load balancing, it should aim
to minimize the number of clients moved after each change to the
system.  
For every update in the system, we need to change the allocation of clients to servers. 
For simplicity, we assume that the updates (ball and bin insertions and removals) do not happen simultaneously and will be operated one at  a time, so that
we have time to finish changing the allocation before
we get another update. Such allocation problems become even more
challenging when we face hard constraints in the capacity of each
server, that is, each server has a {\em capacity\/} and the load may
not exceed this capacity. Typically, we want capacities close to the
average loads.  

There is a vast literature on solutions in the much
simpler case where the set of servers is fixed and only the client set
is updated. For now, we focus on solutions that are known to work
in our fully-dynamic case where both clients and servers can be added
and removed in an arbitrary order.  This rules out solutions where
only the last added server may be removed\footnote{In particular, this rules out the
  external memory techniques \cite{Lar88} where blocks (playing the
  role of fixed capacity servers) can only be added to and removed
  from the top of the current memory.}. The above problem formulation
is very general, and does not assume anything about the ratio between
the number of clients $\balls$, and the number of servers $\bins$.
Processors are cheap, so one could for instance imagine systems with a large
number of servers. However, it is also conceivable having a system with many clients or a balanced system with $\balls\approx\bins$. 

The classic solution to the scenario where both clients and servers
can be added and removed is Consistent
Hashing~\cite{chord,chord-theory} where the current clients are
assigned in a random way to the current servers.  While consistent
hashing schemes minimize the expected number of movements, they may
result in hugely overloaded servers, and they do not allow for
explicit capacity constraints on the servers. The basic point is that
the load balancing of consistent hashing~\cite{chord-theory,chord} is
no better than a random assignment of clients to servers. The
same issue holds for Highest Random Weight
Hashing (popularly known as Rendezvous Hashing)~\cite{TR98:rendezvous}.
Hence, with
$\balls$ clients and $\bins$ servers, we expect good load
balancing if $\balls/\bins=\omega(\log \bins)$, but the balance is
lost with smaller loads, e.g., with $\balls\approx \bins$, we expect many
servers to be overloaded with $\Theta(\log \bins/ \log\log \bins)$
clients.

More recently, Mirrokni et al. \cite{MTZ18:consistent} presented an
algorithm that works with arbitrary capacity constraints on the
servers. For the purpose of load balancing, the system designer can
specify a balancing parameter $c=1+\eps$, guaranteeing that the
maximum load is at most $\ceil{c\balls/\bins}$.  While maintaining
this hard balancing constraint, they limit the expected number of
clients to be moved when clients or servers are inserted or
removed. From a more practical perspective, we think of the load
balancing parameter $c=1+\eps$ as a simple knob which captures the
tradeoff between load balancing and stability upon changes in the
system. This gives a more direct control to the system designer in
meeting explicit balancing constraints.

Even without capacity constraints, the obvious general lower bounds
for moves are as follows. When a client is added or removed, at least
we have to move that client.  When a server is added or removed, at
least we have to move the clients belonging to it. On the average, we
therefore have to move least ${\balls \over \bins}$ clients when a server is
added or removed.

With the algorithm from \cite{MTZ18:consistent}, while guaranteeing a
balancing parameter $c=1+\eps\leq 2$, when a client is added or
removed, the expected number of clients moved is $O({1\over
  \eps^2})$. When a server is added or removed, the expected number of
clients moved is $O({\balls \over \eps^2\bins })$. These numbers are only a factor
$O({1\over \eps^2})$ worse than the general lower bounds without
capacity constrains.  For balancing parameter $c\geq 2$, the expected
number of moves is increased by a factor $1+O(\frac{\log c}c)$ over
the lower bounds. This implies that for superconstant
$c$, we only expect to pay a negligible cost in extra moves.

Focusing on the challenging case where $c=1+\eps \leq 2$, we present an algorithm which reduces the number
of moves by a factor $1/\eps$.  When inserting or deleting a ball, we
expect to move $O(1/\eps)$ balls, and when inserting or deleting a
bin, we expect to move $O(C/\eps)$ balls. To search a ball we only need
to consider $O(\log (1+1/\eps))$ ``consecutive'' bins.

With $C:=cn/m$, these bounds are essentially best possible when $C=O(1)$ is a constant. However, for larger $C$, we can do even better. In order to explain, this we first have to consider the following much simpler probabilistic problem: Consider placing $\balls$ balls in $\bins$ bins, each of capacity $C=(1+\eps)n/m$, one ball at the time, where each ball picks a uniformly random non-full bin. We are interested in the number of non-full bins both in expectation and with concentration bounds. To our surprise, this relatively simple problem does not seem to have been analyzed before, and so, we believe our bounds to be of independent interest. To state our bounds, we define  
\begin{align}\label{eq:formula-for-f}
  f=\left\{\begin{array}{ll}
  \eps C&\textnormal{ if } C\leq \log 1/\eps \\
  \eps\sqrt{C}\cdot \sqrt{\log(1/(\eps\sqrt{C}))}&\textnormal{ if } \log 1/\eps\leq C<\tfrac{1}{2\eps^2}\\ 1& \textnormal{ if } C\geq \tfrac{1}{2\eps^2}\end{array}\right.,
  \end{align}
whenever $0<\eps\leq 1$ and $C\geq1$ is integral.
We are going to prove the following result
  \begin{theorem}\label{thm:main0}
Let $n,m\in \N$ and $0<\eps<1$ be such that $C=(1+\eps)n/m$ is integral. Moreover assume that that $1/\eps=m^{o(1)}$. Suppose we distribute $n$ balls sequentially into $m$ bins each of capacity $C$, for each ball choosing a uniformly random non-full bin. The expected fraction of non-full bins is $\Theta(f)$.
    \end{theorem}
How does this result relate to our dynamic load allocation problem? We can think of the distribution scheme in the theorem as the algorithmically \emph{weakest} way to assign the balls to the capacitated bins. Here, by algorithmically weak, we mean that it  cannot be implemented in the dynamic setting where balls and bins can come and go. However, it is still helpful to think of it as the mathematically \emph{ideal} way of solving dynamic load allocation with bounded loads in the following sense. Imagine that an insertion of a ball is carried out by repeatedly choosing a random bin until we find a non-full one where we place the ball. Then we avoid all the unpleasant dependencies between the loads of the bins visited during the insertion that arise in algorithmically stronger schemes. For example, one can compare to a scheme like linear probing where the cascading effect of balls causes heavy dependencies between the loads of bins visited during a search or an insertion. It follows from~\Cref{thm:main0} that in the simple scheme above,  the expected number of bins visited when making an insertion is $O(1/f)$. The main contribution of this paper is to present a much stronger scheme which supports general insertions and deletions of both balls and bins, and which, nonetheless, achieves complexity bounds that are analogous to those in the mathematically ideal scheme above. To be precise, with our scheme, we expect to move $O(1/f)$ balls when inserting or deleting a ball, and $O(C/f)$ balls when inserting or deleting a bin and this is tight. Similar bounds holds on the number of bins visited when performing any of these updates. Our main technical challenge is handling all the intricate dependencies that arise in the much more complicated probabilistic setting in our scheme.

\paragraph{Applications.}
Consistent hashing has found numerous applications~\cite{OV11,GF12}
and early work in this area \cite{chord-theory,chordSIGComm,chord} has
been cited more than ten thousand times. To highlight the wide
variety of areas in which similar allocation problems might arise, we mention a few more important references to applications: content-addressable
networks \cite{CAN}, peer-to-peer systems and their associated
multicast applications \cite{pastry,scribe}.  Our algorithm and that
from \cite{MTZ18:consistent} are very
similar to consistent hashing, and should work for most of the same
applications, bounding the loads whenever this is desired.
In fact, the algorithm from \cite{MTZ18:consistent} already found two
quite different industrial applications; namely Google's cloud system
\cite{MZ17:google} and Vimeo's video streaming \cite{Rod16}. Both systems
had to handle the lightly loaded case. Also, in both cases, load
balancing was not an objective to maximize, but rather a hard
constraint, e.g., in the Vimeo blog post \cite{Rod16}, Rodland
describes how no server is allowed to be overloaded, and how he found
a load balancing parameter $c=1.25$ to be satisfactory for Vimeo's
video steaming. We shall return to this later. With our algorithm,
we get the same load balancing but with much fewer reallocations.

\subsection{Background: Consistent Hashing}
The standard solution to our fully-dynamic allocation problem is 
consistent hashing \cite{chord,chord-theory}. We shall use it as a starting point for own own solution, so we review it below.
\paragraph{Simple Consistent Hashing.}
In the simplest version of consistent hashing, we hash the active
balls and bins onto a unit circle, that is, we hash to the unit
interval, using the hash values to
create a circular order of balls and bins. Assuming no collisions, a ball is
placed in the bin succeeding it in the clockwise order around the
circle. One of the nice features of consistent hashing is that it is
history-independent, that is, we only need to know the IDs of the
balls and the bins and the hash functions, to compute the distribution
of balls in bins. If a bin is closed, we just move its
balls to the succeeding bin. Similarly, when we
open a new bin, we only have to consider the balls from the succeeding
bin to see which ones belong in the new bin. 

With $\balls$ balls, $\bins$ bins, and a fully random hash function
$h$, each bin is expected to have $\balls/\bins$ balls. This is also
the number of balls we expect to move when a bin is opened or closed.

One problem with simple consistent hashing as described above is that
the maximum load is likely to be $\Theta(\log \bins)$ times bigger
than the average. This has to do with a big variation in the coverage
of the bins. We say that bin $b$ {\em covers\/} the interval of the
cycle from the preceding bin $b'$ to $b$ because all balls hashing to this
interval land in $b$. When $\bins$ bins are placed randomly on the unit
cycle, on the average, each bin covers an interval of size $1/\bins$, but we expect some bins to cover intervals of size $\Theta({\log \bins
  \over \bins})$, and such bins are expected to get
$\Theta({\balls \log \bins \over \bins})$ balls. The maximum load is thus
expected to be a factor $\Theta(\log \bins)$ above the average.

A related issue is that the expected number of balls landing in the
same bin as any given ball is almost twice the average.  More
precisely, consider a particular ball $x$. Its expected distance to
the neighboring bin on either side is exactly $1/(\bins+1)$, so the
expected size of the interval between these two neighbors is
$2/(\bins+1)$. All balls landing in this interval will end in the same
bin as $x$; namely the bin $b$ succeeding $x$. Therefore we expect
$2(\balls-1)/(\bins+1)\approx 2\balls/\bins$ other balls to land with
$x$ in $b$. Thus each ball is expected to land in a bin with load
almost twice the average. If the load determines how efficiently a
server can serve a client, the expected performance is then only half
what it should be.

In \cite{chord-theory} they addressed the above issue using so called virtual bins. We will also employ these virtual bins in our solution and describe them below.

\paragraph{Consistent Hashing with Virtual Bins.}
To get a more uniform bin cover, \cite{chord-theory} suggests the use
of {\em virtual bins}. The virtual bin trick is that the ball contents
of $k=O(\log \bins)$ virtual bins is united in a single super bin. The
super bins are the $\bins$ bins seen by the user of the
system. Internally it is the $k\bins$ virtual bins we place on the cycle
together with the $\balls$ balls. Each virtual bin has a pointer to its super
bin. To place a ball, we go along the cycle to the first virtual bin, and
then we follow the pointer to its super bin.

A super bin covers the union of the intervals covered by its $k$ virtual
bins.  The point is that for any constant $\eps>0$, if we pick a large
enough $k=O(\log \bins)$, then with high probability, each super bin
covers a fraction $(1\pm\eps)/\bins$ of the unit cycle. 

We note that many other methods have been proposed to maintain such a
uniform bin cover as bins are added and removed (see, e.g.,
\cite{BSS00,GH05,Man04,KM05,KR06,TR98:rendezvous}), and in our algorithms, we shall also employ such virtual bins.

With a uniform bin cover, balls distribute uniformly between bins. On
the positive side, in the heavily loaded case when $\balls/\bins$ is
large, e.g., $\balls/\bins=\omega(\log \bins)$, all loads are
$(1\pm o(1))\balls/\bins$, w.h.p. However, with $\balls=\bins$, we still
expect many bins with $\Theta((\log \bins)/(\log\log \bins))$ balls
even though the average is $1$. In this paper, we aim for good
load balancing for all possible load levels.

\subsection{Simple Consistent Hashing with Bounded Loads.}\label{sec:simple-bound}
As we mentioned earlier, Mirrokni et al.~\cite{MTZ18:consistent} presented an
algorithm that works with arbitrary capacity constraints on the
bins. For the purpose of load balancing, the system designer can
specify a balancing parameter $c=1+\eps$, guaranteeing that the
maximum load is at most $C=\ceil{c\balls/\bins}$.

Their idea is very simple. As in simple consistent hashing, we place
balls and bins randomly on a cycle, but instead of placing balls in
the first bin along the cycle, we place them in the first non-full
bin.  Thus we can think of the distribution as first placing all the
bins on the cycle, and then placing the balls one-by-one, putting each
in the first non-full bin found by going in clockwise around the cycle. If we have hash functions for placing
arbitrary balls and bins along the cycle, and if we have a priority
order on all balls, telling us the order in which we insert balls,
then this completely determines the placement of any set of the
balls in any set of capacitated bins. This means that the distribution
is history independent as in \cite{BG07}. It also means that we know exactly
which balls to move if balls or bins are added or removed.

As terminology, we say a ball {\em hash to\/} the first bin following it in 
the clockwise order. However, the ball may be
{\em placed\/} in a later bin if the bin it hashed to was full.

Note that the priority order makes the insertion of a new ball a bit
more complicated since it may have higher priority than balls already
in the system. To place it, we first place it in the bin it hashes to
directly (that is, the one just after its hash location on the
cycle). If the bin becomes overfull, we pop the lowest priority ball
and place it in the next bin, and repeat. It is, however, important
to notice that the bins we end up considering are exactly the bins
from the one the ball hashes to, and to the first non-full bin.

The details of all the different system updates are described in
Mirrokni et al.~\cite{MTZ18:consistent}. This also includes
rolling adjustment of the capacities relative to average load $\balls/\bins$.
Instead of giving all bins the maximal capacity $C=\ceil{c\balls/\bins}$,
they always have $\ceil{c\balls}-\bins\floor{c\balls/\bins}$ bins with
capacity $\floor{c\balls/\bins}$. The only exception is that
we never drop any capacity below 1. A hash function choose which bins have
which capacities, and this ensures that only few capacities have
to be changed with each system update. In Mirrokni et al.~\cite{MTZ18:consistent} they show that their results hold, both when capacities 
are adjusted to $\eps$, and when a joint capacity $C$ is
given, defining $\eps=C\bins/\balls-1$. In this paper, for simplicity,
we will focus on the latter model with fixed capacities.

Mirrokni et al.~\cite{MTZ18:consistent} also provided an analysis of their system. With
$\eps\leq 1$, they showed that starting from the hash location of any
ball, the expected number of full bins passed on the way to the first
non-full bin is $O(1/\eps^2)$. From this they get that the expected
number of balls that has to be moved when a ball is inserted or
deleted is $O(1/\eps^2)$. Likewise, the expected number of balls that
has to be moved when a bin is inserted or deleted is
$O(C/\eps^2)$. These bounds are all tight for simple consistent
hashing with bounded loads.

Finally, Mirrokni et al.~\cite{MTZ18:consistent} also discussed many potentially relevant
techniques that could possibly be made to work for fully-dynamic load balancing where both balls and bins can be added and removed, and with strict requirements on the maximal load for each bin.
In these comparisons, their scheme was the one with the best proven
bounds on the number of moves needed in connection with the updates.

\subsubsection{Faster Searches}\label{sec:intro-fast-search}
Mirrokni et al.~\cite{MTZ18:consistent} states that to search a
ball, they have to consider $O(1/\eps^2)$ bins, but using an old trick
\cite{AK74:ordered-hash-table,knuth-vol3}, this is easily improved to
$O(1/\eps)$. The idea is that when we search for a ball, we can stop
as soon as we reach a bin that is not filled with balls of higher
priority.  This helps the searches if the priorities are random. We shall
use the idea later, so let's elaborate. The bins considered in the
search are exactly the bins from the bin hashed to and till the first
non-full bin if only the balls of higher priority was inserted. Let
$r(q,m,C)$ be expected number of bins considered if there are
$q$ balls of higher priority, and $m$ bins of capacity $C$.
Then with $n$ balls in total, the expected cost with random
priorities is $\sum_{q=0}^n r(q,m,C)/(n+1)$. The analysis in
\cite{MTZ18:consistent} implies $r(q,m,C)=O(1/\eps_q^2)$ where
$\eps_q=C/\frac qm-1$, implying an expected cost of $O(1/\eps)$ with
random priorities.

We note that random priorities do not help with updates, for if we,
say, want to insert a ball, and meet a bin that is full including
balls of lower priority, then we have to place the lowest priority
ball in a later bin. However, finding the established
server of a client if any, is often the most frequent operation in the
system, so a faster search is very important in practice. As stated, a similar analysis gives that for our system, we have to consider fewer bins when searching than when inserting a ball. In particular, we only need to consider $O(1)$ bins in expectation when $C\geq \log 1/\eps$.

\subsection{Our Scheme: Consistent Hashing with Virtual Bins and Bounded Loads}\label{sec:virt-bound}
Our algorithm basically just combines the bounded loads with
virtual bins. When a ball is placed in a virtual bin, it is also placed
in its super bin which has a limited capacity. In the following, we describe two different versions of our scheme. The first one, described in~\Cref{sec:implementation1}, is conceptually the simplest to understand and easier to analyze mathematically. It is this version that we will analyze in the main body of the paper. The second one, described in~\Cref{sec:implementation2}, is the version most suitable to be implemented in practice for several reasons to be described. Our results hold for both implementations, and in~\Cref{sec:practical-imp}, we sketch how to derive the results for the second more practical version. Common to both versions is that we fix some natural number $k$, which is the number of virtual bins for each super bin.

\subsubsection{Mathematically Clean Version: Many Independent Cycles}\label{sec:implementation1}
For this version, we hash each super bin to $k$ different cycles or levels using independent hash functions\footnote{For simplicity, we advice the reader to think of all our hash functions
as fully random. However, our results hold even when the hashing is implemented with the practical mixed tabulation from~\cite{DKRT15:k-part}. We will later sketch how our proofs can be modified to show this.}. The $k$ hash values on the $k$ cycles will be the associated virtual bins of the given super bin. We also hash the balls to the cycles, but contrary to the bins, each ball gets just a single random hash value on a single random cycle. 

The static placement of the balls can be described as follows: We start by placing all balls which hash to the first cycle using standard consistent hashing with bounded loads as described in~\Cref{sec:simple-bound}. We assume that we have priorities on the balls and we will simulate that they are inserted in priority order. After the first level, the balls hashing to this level have thus been distributed into the virtual bins and we put them in the corresponding super bins. Initially, each super bin had capacity $C$. If the virtual bin of such a super bin received $a$ balls  at the first level, its new capacity is then reduced accordingly to $C-a$. We continue this process on level $i=2,\dots,k$. At level $i$, each super bin has a certain remaining capacity and we use standard consistent hashing with bounded loads (with these capacities) to place the balls at level $i$ into the virtual bins and thus, into the corresponding super bins. If a super bin had capacity $C_0$ before the hashing to level $i$, and it received $a$ balls at level $i$, its remaining capacity for the next levels is $C_0-a$ . Traversing the levels one at a time like described, corresponds to enforcing that regardless of the initial priorities of the balls, if two balls hash to different levels, the ball hashing to the lower level will have the highest priority of the two. With these modified priorities, the static image at a given point can be obtained by simply  inserting the balls one by one in priority order, placing each ball in the first virtual bin whose super bin is not full. This completely describes the placement of balls in bins if we know the hash functions and the priority order, so the system is history-independent as described in \cite{BG07}.

Searching for a ball $x$ is almost the same as for normal consistent hashing. We calculate the hash value of $x$ and visit the virtual bins starting from that hash value in cyclic order until we either find $x$ in a corresponding super bin or we meet a ball of lower priority hashing to the same level. 

Insertions are a bit more complicated. For inserting a ball $x$ we calculate $h(x)$ which in particular indicates the level, $i$, that $x$ hashes to. We traverse level $i$ starting at $h(x)$ until we meet a bin, $b$, which either (a) is not full or (b) contains a ball of lower priority than $x$ (all balls hashing to levels $j>i$ have lower priority than $x$ by convention). We insert $x$ in $b$. In case (a), the insertion is complete, but in case (b) we pop $y$ from $b$ and recurse the insertion starting with $y$ (which happens at some level $j\geq i$). 

Ball deletions are symmetric to ball insertions in the sense that the hash functions tells us exactly the placement of all balls in bins, both before and after the ball which we are to insert or delete is inserted or deleted. Deleting a bin is the same as re-inserting all balls in it, and inserting a bin is symmetric to deleting a bin. Therefore we get that the number of balls to be moved is essentially determined by the number that has to be moved in connection with an insertion (we shall discuss this in more detail later).

For most of our results, we will assume that the hashing of balls to the different levels is uniform, but in~\Cref{sec:simple-improvement} we will see an applications where the probability of hashing to level $i$ is $1/2^i$ for $1\leq i \leq k-1$ and $2^{-k+1}$ for $i=k$. In this setting we already obtain a big improvement over standard consistent hashing using just $\log 1/\eps$ levels.

\subsubsection{Practical Version: A Single Linear Order}\label{sec:implementation2} 
We next describe the more practical implementation of our algorithm and here we will also give more details on the concrete ranges of the hash functions. As will be seen, it is very similar to the  the version above having some minor alterations. For this implementation all balls and all virtual bins
are hashed to a single range, which we think of not as a cyclic order but rather as a linear order. In order to describe the static image at given point, we would again consider the balls one by one in priority order, placing each ball in the first virtual bin whose super bin is not full. Again, this ensures that the system is history-independent.

We now provide some more details on the hash functions and the priority
order. Generally the hash values are in some universe
$[u]=\{0,\ldots,u-1\}$. We imagine $u$ to be so large that we expect
no collisions between hash values (if there are ties, we can break
them in favour of the ID's of the balls, but we will ignore this detail). We also think
of both balls and bins having ID's in $[u]$.

We have a single hash $h:[u]\to[u]$ describing the hash
location of the balls.  We also use $h$ to give the
random priority order of the balls, inserting those with smallest hash
values first.

For the super bins, and for some parameter $k$, each bin has $k+1$
associated virtual bins. Their hash locations are described via $k+1$
hash functions $h_i:[u]\to[u]$, $i\in [k]=\{0,\ldots,k\}$. We assume
that $k$ divides $u$, e.g., that both are powers of two, and
we restrict $h_i$ to map uniformly into $[iu/k,(i+1)u/k)$.
  This
way each super bin gets exactly one virtual bin in each of the $k+1$ intervals
$[iu/k,(i+1)u/k)$. Having this spread is important because of the priority
order of the balls, which implies that virtual bins with larger hash values
are more likely to be full.

The last interval $[u,u+u/k)$ is outside the normal hash range
  $[u]$. These last virtual bins will pick up any key that did not end
  in a bin in the normal range $[u]$. Since every super bin is
  represented in $[u,u+u/k)$, all balls are picked up unless there are
    more balls than the total capacity. As a result, we do no longer
    think of balls and bins as hashing to a cycle, but just to a linearly
    ordered universe with an extra set of representative virtual bins by the end making sure that all
    balls get placed. 
    
    We briefly explain why this system is preferable in practice. The first reason is that when using the hash values of the balls as their priorities we obtain a very simple description of the static distribution of balls in the bins: We may simply insert the balls in order from lowest to highest hash value, always placing the ball in the first non-full bins. A way of picturing this is to imagine that the balls of lower hash values are ``pushing'' balls of higher hash values ahead of them. On a line, it is very easy to implement this comparison as a standard comparison between hash values. In fact, it is possible to obtain a similar image for cycles, but for this one needs to impose a \emph{cyclic} priority order of the balls hashing to a given level, and performing comparisons for such a cyclic order is a bit more technical to implement\footnote{For example, for just two balls, the notion of one hashing before the other is not well defined.}. If on the other hand, we decided to stick with the linear priority order on each cycle, thus giving up on the nice image from above, we still encounter some technical issues with the implementation. With searches and insertions, everything works fine, but the issues come up when deleting balls and inserting bins. For instance, when deleting a ball which is placed in the ``last'' bin on the cycle, we may have to pull back balls that have been forwarded from this bin to the ``first'' bins in the cycle, and for deciding if such balls are to be pulled back, we have to use a different comparison of hash values. Thus, even with linear priorities the cyclic probing still muddies the implementation and makes it less efficient.

Again, we shall play a bit with the ranges of the hash functions for the virtual bins. However, they will always partition $[u]$ consecutively with the range of $h_i$ following the range of $h_{i-1}$. With the {\em exponentially decreasing\/} hash ranges described by the end of~\Cref{sec:implementation1}, $h_i$, maps uniformly to  $[u-u/2^i,u-u/2^{i+1})$ for $i\in [k-1]$ and $h_{k-1}$ maps uniformly to $[u-u/2^{k-1},u)$. As above $h_k$ is special, mapping to $[u,u+u/k)$.

Searches and insertions have similar descriptions to the ones given in~\Cref{sec:implementation1}. Moreover, the history independence again implies that deletions are symmetric to insertions. Finally, deleting a bin corresponds to inserting the ball in the bin, and inserting a bin is symmetric to the deletion of the bin.

\subsection{Main Results on Consistent Hashing}
We now present our main results on consistent hashing with bounded loads and virtual bins. 

\subsubsection{$O(1/\eps)$ Reallocated Balls, with $\log 1/\eps$ Levels}
Our first result, to be proved in~\Cref{sec:simple-improvement}, uses a logarithmic number of virtual bins to achieve that the number of bins visited during an insertion (and thus the number of reallocated balls) is $O(1/\eps)$.  It uses a non-uniform distribution of the balls to the different levels, with the probability of a ball hashing to level $i$ being $2^{-i}$ for $1\leq i \leq k-1$ and $2^{-k+1}$ for  $i=k$.
\begin{theorem}\label{thm:main1}
Let $0<\eps <1$ and suppose that we distribute $n$ balls into $m$ bins each of capacity $C=(1+\eps)n/m$ using consistent hashing with bounded loads and $k=\lceil \log (1/\eps)\rceil$ levels, where the probability, $p_i$, that a ball hashes to level $i$ is 
$$
p_i=\begin{cases}
2^{-i}, & 1\leq i \leq k-1 \\
2^{-k+1},& i=k.
\end{cases}
$$
Assume that $1/\eps=n^{o(1)}$. When inserting or deleting a ball, we expect to visit (and hence move) $O(1/\eps)$ balls, and when inserting or deleting a bin, we expect to move $O(C/\eps)$ balls. Finally, when searching a ball, we expect to visit $O(\log 1/\eps)$ bins.
\end{theorem}
In the previous system of simple consistent hashing with bounded loads, but no virtual bins, Mirrokni et al.~\cite{MTZ18:consistent} proved that ball insertions and deletions are expected to move $O(1/\eps^2)$ balls while bin insertions and deletions are expected to move $O(C/\eps^2)$ balls.  Those bounds are a factor $1/\eps$ worse than ours. Mirrokni et al.~\cite{MTZ18:consistent} would also perform searches considering $O(1/\eps^2)$ bins in expectation, but using the trick of assigning random priorities to the balls, one can get down to $O(1/\eps)$ bins in expectation, still without the use of virtual bins. Combining our scheme using virtual bins, with the trick of random priorities the expected number of bins visited during a search drops exponentially to $O(\log 1/\eps)$, as stated in the theorem. 

When proving~\Cref{thm:main1}, the main technical challenge is bounding the expected number of bins visited during an insertion. In fact, the remaining parts of the theorem follow once we have this bound. In~\Cref{sec:focus-on-insertions}, we will argue why the results on ball deletions and bin insertions and deletions follow. Finally, in~\Cref{sec:faster-search}, we will use the trick described in~\Cref{sec:intro-fast-search} to prove the result on ball searches.

\subsubsection{Better Bounds when the Capacities are Large}
In classic consistent hashing without virtual bins, we obtain no advantage when the number of balls $\balls$ are much larger than the number of bins $\bins$, or in other words, when the capacity of a bin, $C$, is large. The basic issue is that most of the uncertainty in the system without virtual bins stems from the uncertainty in the distance between a bin and its predecessor, which determines the expected number of balls hashing directly to the bin.

However, the use of virtual bins improves the concentration of the number of balls hashing directly to a super bin, and we do obtain an advantage of this improved concentration. This was in fact the whole point of introducing virtual bins in classic consistent hashing without load bounds~\cite{chord}. To be precise, fix $k=A(\log n)/\eps^2$ 
for some appropriately large constant $A$. Then standard Chernoff bounds show that each bin cover a fraction $(1\pm\lambda\eps)/m$ of the combined hash range, where $\lambda$ can be made arbitrarily small (by increasing $A$). If further the average load $\bins/\balls$ is above $k$, then with high probability, no bin gets load above $C=(1+\eps)\bins/\balls$ by balls hashing directly to them. In particular, all load bounds are satisfied without the having to forward a single ball. The result below (which is the main result of our paper) asymptotically settles the expected insertion time for general $C$, in particular for any $C\leq (\log n)/\eps^2$. Before stating the theorem, we encourage the reader to recall the definition of $f$ in~\cref{eq:formula-for-f}

\begin{theorem}\label{thm:main}
Let $0<\eps <1$ and suppose that we distribute $n$ balls into $m$ bins each of capacity $C=(1+\eps)n/m$ using consistent hashing with bounded loads and $k=c/\eps^2$ uniform levels for a sufficiently large constant $c$. Assume that $1/\eps=n^{o(1)}$.  In expectation we move $O(1/f)$ balls when inserting or deleting a ball, and $O(C/f)$ balls when inserting or deleting a bin. Finally, when searching a ball, we expect to visit $O(1)$ bins when $C\geq \log 1/\eps$ and $O(\frac{\log 1/\eps}{C})$ bins when $C< \log 1/\eps$.
\end{theorem}

Our bounds in Theorem \ref{thm:main} show that we do get an advantage
from bigger capacities even when $C$ is smaller than $k=\Theta((\log n)/\eps^2)$. In fact, already for $C=1/\eps^2$, the expected insertion time drops to $O(1)$.

Again, the hardest part of proving~\Cref{thm:main}, is bounding the expected number of bins visited during an insertion by $O(1/f)$. As for~\Cref{thm:main1}, we argue that the remaining parts of the theorem follows in~\Cref{sec:focus-on-insertions,sec:faster-search}

\paragraph{High Probability Bounds} ~\Cref{thm:main1,thm:main} only bound the expected number of balls moved during the insertions and deletions of balls and bins. However, it is also possible to obtain high probability bounds. We will provide such high probability bounds in a later full version of the paper.  

\subsubsection{Distributing Balls Randomly into Capacitated Bins}
To understand the strength of our bounds, we consider a much
simpler problem where we place $\balls$ balls in $\bins$ bins, each of
capacity $C=(1+\eps)n/m$, one ball at the time. Each ball picks a uniformly random
non-full bin. Letting $X$ denote the fraction of non-full bins, we show in~\Cref{sec:capacitated-bins} that $\E[X]=\Theta(f)$ and $X=\Theta(f)$ with high probability. Surprisingly, this relatively simple question has not been studied before. 

What is the idea of considering this simpler distribution scheme? With a fraction of $X$ non-full bins, the expected number of random bins visited in order to find one of the non-full ones is $1/X$. This is reminiscent to searching for a non-full bin using (any variation of) consistent hashing with bounded loads, except that we get rid of the intricate dependencies which arise in the more complicated schemes that can handle both insertions and deletions. In this way, the scheme above can be thought of as the simplest way of achieving the desired load balancing, but of course it has no chance of working in a fully dynamic setting. We thus obtain, the same complexity bounds as the weakest system imaginable, at the same time being able to handle both insertions and deletions of balls and bins.

\subsubsection{The Practical Implementation with Mixed Tabulation}
When proving~\Cref{thm:main1,thm:main}, we will assume that our scheme is implemented as described in~\Cref{sec:implementation1} and, moreover, using fully random hash functions. In~\Cref{sec:practical-imp} we will sketch why our results hold even with the more practical implementation from~\Cref{sec:implementation2}. We will also sketch how one can obtain the same results with the practical mixed tabulation scheme from~\cite{DKRT15:k-part}. In the implementation with mixed tabulation, we would use $k$ independent mixed tabulation hash functions for the hashing of virtual bins, and a single independent mixed tabulation hash function for the hashing of balls.

\subsection{The Model and its Applicability.} 
Consistent hashing with or without virtual bins is a simple versatile
scheme that has been implemented in many different systems with
different constraints and performance measures~\cite{OV11,GF12}.  The
most classic implementation of consistent hashing is the distributed
system Chord \cite{chordSIGComm,chord} which has more than ten
thousand citations. The Chord papers \cite{chordSIGComm,chord} give a
thorough description of the many issues affecting the design. On the
high level, they have a system of pointers so that given an arbitrary
hash location, they can find the next bin in the clockwise order
using $O(\log n)$ messages. This is how they find the (virtual) bin a ball
hashes to. In simple consistent hashing, this is where the ball is to
be found. With virtual bins, there are additional pointers between
virtual bins and their super bins that we can follow using $O(1)$
messages. In fact, Chord does maintain explicit successor pointers between
neighboring (virtual) bins, so we only have to pay $O(1)$ extra messages
to find a next bin along the cycle.

As described by Mirrokni et al.~\cite{MTZ18:consistent}, the successor
pointers give immediate support for forwarding in case of capacitated
bins. Mirrokni et al.~only used this forwarding for simple
consistent hashing without virtual bins, and this has been adopted
both by Google's Cloud Pub/Sub \cite{MZ17:google} and Vimeo
\cite{Rod16}. Both systems had to handle the lightly loaded
case. Also, in both cases, load balancing was not an objective to
maximize, but rather a hard constraint, e.g., in the Vimeo blog post
\cite{Rod16}, Rodland describes how no server is allowed to be
overloaded, and how he found a load balancing parameter
$c=1+\eps=1.25$ to be satisfactory for Vimeo's video steaming.

The successor pointers in Chord work equally well for moving between
virtual bins. In fact, Rodland from Vimeo has told (personal
communication) the last author, Thorup, that their system does allow a
combination of virtual bins and bounded loads, like what we suggest in
this paper, so a system similar to ours is already running.  Thorup had the
general idea from much earlier (around the time of the first versions of~\cite{MTZ18:consistent}), but deriving the mathematical
understanding, presented here in Theorem \ref{thm:main} took several
years.

Let us now consider the time to search a ball in a Chord-like
setting. By Theorem \ref{thm:main}, we expect to consider $O(\log(1/f))$
consecutive virtual bins with associated super bins. Finding the
virtual bin succeeding the hash location uses $O(\log n)$ messages
while each other bin is found with $O(1)$ messages. Then our
message bottleneck is actually to find the first virtual bin.

Now it could be the case that balls/clients themselves remembered if they
are in the system, and if so, what bin/server they belonged to. The
latter requires that they are notified if they get moved due to
other updates in the system, e.g., if their bin/server was removed.

Another way to circumvent the $O(\log n)$ messages for placing the
hash location would be if we for some $\hat m=\Theta(m)$, placed
the reference points $p_i=ui/\hat m$, $i\in[\hat m]$, in the
doubly-linked list of virtual bins. For a ball $x$ its
hash reference point is $p_{\floor{h(x)\hat m/u}}$. Regardless
of system updates, it could remember its reference point, and from
there follow in expectation $O(1)$ successor pointers to get the
current virtual bin succeeding its real hash location. The reference points
could be updated by background rebuilding to be ready every time $m$
is halved or doubled, thus maintaining an $\hat m$ approximating $m$
within a factor of 2.

In fact, our scheme is equally relevant for less distributed systems
than Chord. In Google's Cloud Pub/Sub \cite{MZ17:google},
the most important aspects of the system was (1) that it has good load
balance (2) that only few clients/balls have to be moved in connection
with update, that is, a ball or bin insertion or deletion, and (3)
history independence so that the placement of balls in bins can be
computed by anyone knowing the hash functions and the current set of
balls and bins. The fact that each system update only leads to few moves
implies that even if we have a few mistakes in the set of balls and bins, then
this only implies a few mistakes in the placement of balls in bins.

System updates, inserting or deleting a ball or a bins are hopefully
not too frequent. As mentioned in \cite{MZ17:google}, the
dominant concern is the actual reallocation of balls between bins; for in the real
world, this means moving clients between servers disrupting service
etc.~\Cref{thm:main1,thm:main} give us concrete bounds on how many balls
we expect to move.

The computation of which balls are to be moved in connection with
updates depends very much on the situation. As in \cite{MZ17:google},
thanks to history independence, we can compute the balls to be moved
from scratch. We know the update to the set of balls and bins, and the
hash functions tell us exactly which balls are placed in which bins
before and after update. The difference tells us exactly which balls
have to be moved. This solution if fine if the computation cost
is small compared with the cost of actually moving the clients.

Alternatively, we may want a more distributed local identification of
the moves as in in the Chord system. This is fairly straightforward for
insertions, and we already described it earlier. It does, however, get
a bit more complicated for the other updates, and we shall return to
such a distributed implementation in Section \ref{sec:local-moves}.

Stepping back, we offer a generic scheme for a load balanced
distribution of balls in bins when both can be added and removed.  We
are not claiming to have a theoretical model that captures all the
important aspects of performance since this depends very much on the
concrete implementation context. Our main contribution is a
theoretical analysis of combinatorial parameters described in~\Cref{thm:main1,thm:main}.

\subsection{Computing Moves Locally in a Distributed Environment}\label{sec:local-moves}
We will now discuss how we could compute which balls have to be moved
in connection with system updates in a distributed Chord-type system.
Recall that sometimes it may be fast enough to identify the moves
more centrally, simply by computing the placement of the balls in the
bins before and after the update, and just identify the
difference. However, in this subsection, we will discuss how to identify
the moves locally, not spending much more time than the number of moves
specified in~\Cref{thm:main1,thm:main}.

We already discussed how to insert
balls, but we want to do it in a way that also makes
it fast and easy to delete balls.
The basic idea to make deletions efficient is that we for every
virtual bin store the number of balls that have passed it. More
precisely, each bin has a pass count that starts at zero when there
are no balls. We now consider the process where balls are inserted in
priority order, each just placed in the first virtual bin with a
non-empty super bin. This increases the count on all the virtual bins
between the hash location and the virtual bin the ball ends in.
Each super bin will also store which of its virtual bins that have
a positive pass count.

The above pass counts are quite easy to maintain when balls arrive
to the real system, that is, not in priority order. To see this,
we review the insertion of a ball, adding when pass counts should
be incremented. To
insert a new ball, we first hash it to some location which also determines its priority.  Starting from the hash location, we visit
the virtual bins following, each time looking in the corresponding
super bin. If the super bin is not full, we simply place the ball in
it and terminate the insertion.  If the super bin is filled with balls
of higher priority, we increment the pass count of the virtual bin, and
continue to the next virtual bin. However, if the super bin
is filled and contains a ball of lower prioirty, we insert the new
ball and pop the ball of lowest priority.  The popped ball belongs to
some virtual bin, which could be the same, but could also be only much
later in the linear order than the virutal bin we just came from. The pass count is incremented from whichever virtual bin we pop the ball from, and
then we recursiviely rinsert the popped ball, continuing from the next
virtual bin. The $O(1/f)$ bound from Theorem
\ref{thm:main} actually bounds not only the number of moves, but also
the number of bins considered during the above insertion.

Next we consider the deletion of a ball. Essentially, we just want to
reverse the above process, systematically finding the balls the ball
to be deleted have displaced.  We think
of deletions as first removing a ball, and then recursively, filling a hole.
Finding the ball to be removed is easy, as described before, and when
we remove it, we will have to decrement the pass count on all the virtual
bins between its hash location and up to the virtual bin before the one
it landed in. Next we want to see if we can refill the whole. Assuming
that the bin we removed was in the level $i$ virtual bin of a super
bin. We now check corresponding super bin $b$ to see if any
ball has been displaced by the ball we deleted. This is
the case if and only if at least one of its virtual bins has
a positive pass count. Let $j$ be the lowest level of a virtual bin
with a positive pass count. It is not hard to see that we must
have $j\geq i$. We now consider the virtual bins following the level $j$
virtual bin until we find a ball with hash location before $h_j(b)$.
The virtual bins passed decrease their counts, and then we recursively
delete the ball. As described above, our total work is within a constant
factor of the symmetric insertion, that is, we consider $O(1/f)$ bins and spend $O(1/f)$ time in total.

We now consider the insertion of deletion of super bins. We think
of these super bin or server updates as more rare than the ball or
client updates.

Deleting a super bin $b$ is relatively easy. Essentially, we just reinsert
all the balls in it. A small detail is that if a ball $x$ was in the level $j$
virtual bin, then we insert it starting from $h_j(b)$ rather than from
$h(x)$. This can only save work over the regular insertion of $x$ and in
particular, this means that we do not increase the
pass count for virtual bins between $h(x)$ and $h_j(b)$.  By~\Cref{thm:main}, the expected number of balls that has to be moved when deleting a super bin is $O(C/f)$. However, on top of that, we do have to spend at least
$O(k)$ time on removing the $k$ virtual bins from the system.

Inserting a super bin $b$ is a bit more complicated. We would like
to just fill it as we filled the holes arising when deleting a ball, but
we have the issue that we do not know the pass counts for
the $k$ virtual bins representing the new super bin.
To handle this, for $i=1,\ldots,k$, we first find the hash location $h_i(b)$ of
its virtual bin $b_i$, which takes $O(\log n)$ messages, including inserting
it in the linked list of virtual bins. Next consider the virtual bin $u$
following $b_i$.  If bin $u$ has no ball and pass count zero, then
we can just set the pass count of $h_i(b)$ to zero. Otherwise, we
continue along the virtual bins, counting the balls in them, until we
find a ball that hash after $h_i(b)$. All but the last ball are the
balls that have passed the level $i$ virtual bin $b_i$, which now gets
a pass count.  Now that we have the pass count, we can move those
balls to $b_i$, as long as super bin $b$ has space for them, using the
same procedure as described under deletions of balls.

We now first analyze the number of bins considered to compute the pass
counts of the virtual bins $b_i$. We note that the bins considered
are exactly the same as if we searched for a ball that hashed to
$h_i(b)$. Now consider instead the case where we first generate a random
$i\in [k]$, and then generate $h_i(b)$. With $i$ random, $h_i(b)$ is uniformly
random in $[u]$, and then the expected number of bins considered
is exactly the same as those considered in the search of a ball with
hash value uniformly random in $[u]$. We conclude that the expected total number of bins
considered over all $i\in [k]$ is exactly $k$ times bigger. Thus, by~\Cref{thm:main}, we expect to consider at most $O(k\tfrac{\log(1/\eps)}{C})$ bins when $C\leq \log 1/\eps$, and only $O(k)$ bins when $C\geq \log 1/\eps$.
Now that the pass counts are fixed, inserting a bin is symmetric to deleting it and has
the same cost, yielding a bound of $O(C/f)$.

\subsection{Dynamic Load Capacities}\label{sec:dyn-load-cap}
We now also consider what happens when we use self-adjusting capacities like
Mirrokni et al.~\cite{MTZ18:consistent}. Below, the capacitated bins
correspond to our super bins. Rather than fixed capacities, the user
of the system specifies a balancing parameter $c=(1+\eps)$ and then
the maximal capacity is $C=\ceil{c\balls/\bins}$. We do not want all
bins to change capacity each time $c\balls/\bins$ passes an integer.

Instead, as in Mirrokni et al.~\cite{MTZ18:consistent}, 
assuming an arbitrary fixed ordering of the super bins, we let the lowest 
$q=\ceil{c\balls}-\bins\floor{c\balls/\bins}$ super bins have
capacity $C=\ceil{c\balls/\bins}$ while the remaining $r=\bins-q$ have capacity
$C-1$.  We refer to the former bins as {\em big
  bins} and the latter bins as {\em small bins}, though the difference
is only 1.  Moreover, as an exception to the above rule, we will never
let the capacity drop below $1$, that is, if $c\balls<\bins$, then all
bins have capacity $1$.

The basic point in the above system is that a ball update changes at most
$\ceil c=O(1)$ bin capacities while a bin update changes at most $O(C)$
capacities. Switching the capacity from large to small has the same
effect as inserting an extra high priority ball in the super bin while
leaving the capacity at $C$. In the other direction, switching the capacity from small to large corresponds to a deletion of an extra high priority ball.

From an analysis perspective, this means that we are essentially
studying a system with $\balls'=r+\balls$ balls in bins of capacity $C$
where $C\bins=\ceil{c\balls'/\bins}$.
In our analysis, this corresponds to having a 0th level which puts exactly
one ball in each of $r$ bins; 0 in the rest. Such a perfect level poses no issues for the analysis (see~\Cref{sec:dyn-cap-analysis} for more details on this). Thus the cost per capacity change is the same as that of
regular insertions/deletions, and therefore have no effect on our
overall bounds.

A small point, elaborated in Mirrokni~\cite{MTZ18:consistent}, is that
for all the bounds to hold, we may always do things in the order
that maximizes capacity in every step, so that we always have a total capacity of
$C\bins=\ceil{c(\balls+q)/\bins}$. For example, when inserting
a ball, we increase capacities before inserting, while deleting a ball,
we decrease capacities last. Likewise for a bin insertion, we insert it
before decreasing capacities, while when deleting a bin,
we start by increasing the capacities.

\subsection{Roadmap of the Paper}
We now present a brief roadmap of our paper as well as some of the theorems to be proven in the individual sections.

In Section~\ref{sec:simple-improvement}, we prove the part of~\Cref{thm:main1} concerning insertions of balls. That the statements about ball deletions and bin insertions and deletions follow, is covered in~\Cref{sec:focus-on-insertions}. Finally, in~\Cref{sec:faster-search} we prove the statement of the theorem concerning ball searches.

To prove the main result of the paper,~\Cref{thm:main}, we first have to solve the much simpler problem of showing that when $n$ balls are distributed into $m$ bins each of capacity $C=(1+\eps)n/m$, the expected fraction of non-full bins is $\Theta(f)$. This simpler problem is solved in Section~\ref{sec:capacitated-bins}. 

In Section~\ref{sec:helpful-lemmas}, we present a tail bound for sums of geometric random variables as well a technical lemma concerning consistent hashing with bounded loads and virtual bins. These results will be useful in the later sections towards the proof of~\Cref{thm:main}.  

In Section~\ref{sec:f-properties}, we show that when distributing $n$ balls into $m$ bins using consistent hashing with bounded loads and enough levels, it similarly holds that the expected fraction of non-full bins is $\Theta(f)$, and moreover, that the number of non-full bins is concentrated around its mean. The exhibition is divided into two parts: 
In Section~\ref{sec:f-concentration}, we prove the concentration result and in Section~\ref{sec:f-bound}, we determine the mean within a constant factor. The following theorem is a corollary of the results from~\Cref{sec:f-properties} and we will require it to prove our main result in Section~\ref{sec:insertion}.

\begin{theorem}\label{thm:jakobs}
Let $\balls,\bins\in \N$ and $0<\eps<1$. Suppose we insert $\balls$ balls into $\bins$ bins, each of capacity $C=(1+\eps)\balls/\bins$, using consistent hashing with bounded loads and virtual bins and $k$ levels. For $(i,j)\in [k]\times [C+1]$, we let $X_{i,j}$ denote the number of bins with at most $j$ balls after the hashing of balls to levels $0,\dots,i-1$ and $\mu_{i,j}=\E[X_{i,j}]$. For any $\gamma=O(1)$ and $(i,j)\in [k]\times [C+1]$, it holds that $|X_{i,j}-\mu_{i,j}|\leq m^{1/2+o(1)}$ with probability $1-n^{-\gamma}$. 

If moreover $k\geq c/\eps^2$ for a sufficiently large universal constant $c$, it holds that $\mu_{k-1,C-1}=\Theta(fm)$.
\end{theorem}

In Section~\ref{sec:insertion}, we show the part of~\Cref{thm:main} which concerns ball insertions. Again, ball deletions, bin insertions, and bin deletions are handled in~\Cref{sec:focus-on-insertions}, and searches are handled in~\Cref{sec:faster-search}.

In~\Cref{sec:practical-imp}, we sketch why our results hold, even if we use the practical implementation described in~\Cref{sec:implementation2}. We also sketch how to modify the proofs in the case where the hashing is implemented with the mixed tabulation scheme from~\cite{DKRT15:k-part}.

Finally, in~\Cref{sec:dyn-cap-analysis}, we sketch why our analysis continues to holds even with the dynamically changing capacities described in~\Cref{sec:dyn-load-cap}.

\section{Expected $O(1/\eps)$ Insertion Time with $\lceil \log(1/\eps)\rceil$ Levels}\label{sec:simple-improvement}
In this section we prove the part of~\Cref{thm:main1} concerning insertions, restated below. We will assume that we use the implementation described in~\Cref{sec:implementation1} but the result also holds with the other implementation in~\Cref{sec:implementation2} (see the~\Cref{sec:practical-imp}).
\begin{theorem}\label{thm:main_combinatorial}
Suppose that we distribute $n$ balls into $m$ bins each of capacity $C=(1+\eps)n/m$ using consistent hashing with bounded loads and\footnote{For simplicity, we have stated the theorem using $k=\lceil \log (1/\eps)\rceil+2$ levels as this makes the constants in the proof work out particularly nicely. However, a simple inspection of the proof of~\Cref{thm:main_combinatorial} will show that the bound holds for any positive integer $k= \log (1/\eps)-O(1)$. 
} $k=\lceil \log (1/\eps)\rceil+2$ levels, where the probability, $p_i$, that a ball hashes to level $i$ is 
$$
p_i=\begin{cases}
2^{-i}, & 1\leq i \leq k-1 \\
2^{-k+1},& i=k.
\end{cases}
$$
Assume that $1/\eps=n^{o(1)}$. The expected number of bins visited when inserting a ball is then $O(1/\eps)$.
\end{theorem}
We remark that one way to implement the above hashing is by using an auxiliary hash function $s:U \to [2^{k-1}]$. Letting $h_1,\dots,h_k$ denote the hash functions distributing balls at level $1,\dots,k$, the hash value of a key $x\in U$ is then given by $h_{i+1}(x)$, where $i$ is the number of leading $0$'s of $s(x)$.

\begin{proof}
Let $Z$ denote the number of virtual bins visited in total and $Z_i$ denote the number of virtual bins visited at level $i\in [k]$. Then $Z=\sum_{i=1}^kZ_i$. We will show that $\E[Z_i]=O(2^i)$ from which it follows that $\E[Z]=O(2^k)=O(1/\eps)$. 

First, it follows from a standard Chernoff bound that if $X_i$ is the number of balls hashing to level $i$ and $\mu_i=\E[X_i]=p_in$, then for $\delta\leq 1$,
$$
\Pr[|X_i-\mu_i|\geq \delta \mu_i]\leq \exp(-\delta^2\mu_i/3)
$$
Thus , it holds that $|X_i-\mu_i|=O(\sqrt{\mu_i\log n})$ with probability at least $1-n^{-3}$. Similarly, if $X_{<i}=\sum_{j<i}X_j$ and $\mu_{<i}=\sum_{j<i}\mu_j$, it holds that $|X_{<i}-\mu_{<i}|=O(\sqrt{\mu_{<i}\log n})$ with the same high probability.

For each $j\in [m]$, we define $C_j^{(i)}$ to be the remaining capacity of bin $j$ after the distribution of balls to levels $1,\dots,i-1$. Then $\sum_{j\in[m]}C_j^{(i)}=(1+\eps)n- X_{<i}$, so it follows from the above that with probability $1-O(n^{-2})$,
$$
\sum_{j\in[m]}C_j^{(i)}\geq (1+\eps) n-\mu_{<i}-O\left(\sqrt{\mu_{<i}\log n}\right)= (\eps+2^{-i+1})n-O(\sqrt{n\log n}).
$$
For $i<k$ we have that $\mu_i=2^{-i}n$, so it follows that, $\sum_{j\in[m]}C_j^{(i)}\geq 2X_i$ with probability $1-O(n^{-2})$, where we used the assumption that $1/\eps=n^{o(1)}$. In the case $i=k$, we instead have that 
$$
X_k\leq 2^{-k+1}n+O(\sqrt{n\log n})\leq \eps n/2+(\sqrt{n\log n}),
$$
with probability at least $1-O(n^{-2})$, so again $\sum_{j\in[m]}C_j^{(i)}=X_k+\eps n\geq 2X_k$, again using that $1/\eps=n^{o(1)}$.

Now fix $i\in[k]$, and write $C_j=C_j^{(i)}$ for simplicity. Let $\mathcal{E}$ denote the event that $p_in/2\leq X_i \leq 2p_in$ and that  $\sum_{j\in[m]}C_j^{(i)}\geq 2X_k$. Then $\Pr[\mathcal{E}^c]=O(n^{-2})$, so that 
$$
\E[Z_i]\leq \E[Z_i \mid \mathcal{E}]+\E[Z_i \mid \mathcal{E}^c]\Pr[\mathcal{E}^c]\leq \E[Z_i \mid \mathcal{E}]+O(mn^{-2})=\E[Z_i \mid \mathcal{E}]+O(1).
$$ 
Thus, it will suffice to show that $\E[Z_i \mid \mathcal{E}]=O(2^i)$. Let $b$ be the first bin visited at level $i$, i.e., during the insertion we at some level $j<i$ arrived at bin $b$ and $b$ is not full after the hashing of balls to level $1,\dots,i-1$. Let $I$ be a maximal interval at level $i$ containing $b$ and satisfying that all bins lying in $I$ are full at level $i$. Let $R$ denote the number of bins in $I$ excluding $b$. Then $Z_i\leq R+1$. We will show that $\E[R]=O(2^i)$ (for notational convenience the conditioning on $\mathcal{E}$ has been left out). Let $s\in \N$ be given and let $A_s$ denote the even that $s+1\leq R \leq 2s$. We are now going to provide an upper bound on $\Pr[A_s]$. Let $I_1^{-}$ and $I_1^+$ be the intervals respectively ending and starting at $b$ and of lengths $\frac{s}{3m}$. Similarly, let $I_2^{-}$ and $I_2^+$ be the intervals respectively ending and starting at $b$ and of lengths $\frac{3s}{m}$. Let $I_1=I^{-}_1\cup I^{+}_1$ and $I_2=I^{-}_2\cup I^{+}_2$. Finally, partition $I_2$ into $54$ intervals of equal lengths, $J_1,\dots,J_{54}$. Let $a$ be such that $(1-a)/(1+a)=5/6$ (or $a=1/11$) and $\overline{C}=\frac{1}{m}\sum_{j\in [m]}C_j$. We claim that if $A^s$ holds then either of the following events must be true
\begin{itemize}
\item[$B_1$:] $I_2^{-}$ or $I_2^{+}$ contains at most $2s$ virtual bins different from $b$.
\item[$B_2$:] $I_1^{-}$ or $I_1^{+}$ contains at least $s/2$ virtual bins different from $b$.
\item[$B_3$:] The total capacity of bins different than $b$ hashing to $J_j$ is at most $\frac{(1-a)s\overline{C}}{9}$ for some $1\leq \ell\leq 54$.
\item[$B_4$:] The total number of balls hashing to $J_\ell$ is at least $\frac{(1+a)s\overline{C}}{18}$ for some $1\leq \ell\leq 54$.
\end{itemize}
To see this, suppose that $A_s$ occurs but neither of $B_1,B_2,B_3$ occurs. We show that then $B_4$ must occur. As $R\leq 2s$ and $B_1$ did not occur, $I\subseteq I_2$. As $R\geq s+1$ and $B_2$ did not occur, either $I_1^{-}\subseteq I$ or $I_1^{+}\subseteq I$. Letting $\ell$ denote the number of $j$ such that $J_j\subseteq I$ it therefore follows that $\ell\geq 3$. Since $B_3$ did not occur, the total capacity of bins hashing to $I$ is at least $\frac{\ell\overline{C}(1-a)s}{9}$. Finally, since all balls which ends up in a bin in $I$ must have hashed to $I$ it follows that the total number of balls hashing to $I$ is at least $\frac{\ell\overline{C}(1-a)s}{9}$. In particular, for some $1\leq \ell\leq 54$, at least $\frac{\ell\overline{C}(1-a)s}{9(\ell+2)}$ balls must hash to $J_\ell$. But 
$$
\frac{\ell\overline{C}(1-a)s}{9(\ell+2)}\geq \frac{\overline{C}(1-a)s}{15}=\frac{\overline{C}(1+a)s}{18},
$$
so we conclude that $B_4$ holds.

Simple Chernoff bounds gives that inequality give that $\Pr[B_1]=\exp(-\Omega(s))$ and $\Pr[B_2]=\exp(-\Omega(s))$.
To bound $\Pr[B_3]$, let $\ell\in [54]$ be fixed and define $Y_j$ to be the indicator for bin $j$ hashing to $J_\ell$. Further, define $Y=\sum_{j\in [m]}Y_j$. Then $\E[Y]=\frac{\overline{C}s}{9}$. This time however, we only have that $|Y_j|\leq C$, so applying Chernoff we obtain that 
$$
\Pr[B_3]=\Pr[Y\leq (1-a)\E[Y]]=\exp\left(-\Omega\left(\frac{\E[Y]}{C}\right)\right)=\exp\left(-\Omega\left(\frac{s}{2^i}\right)\right)
$$
For $B_4$, note that since we conditioned on $\mathcal{E}$, the expected number of balls hashing to an interval $J_\ell$ is $ \frac{X_is}{9m}\leq \frac{\overline{C}s}{18}$. Thus, another Chernoff bound yields that $\Pr[B_4]=\exp(-\Omega(s\overline C))$.  Note that $\overline{C}\geq 1/2^i$, so that we in particular have that $\Pr[B_4]=\exp(-\Omega(s/2^i))$. Combining our bounds, it follows that for $s\geq 2^i$,
$$
\Pr[A_s]=\exp\left(-\Omega\left(\frac{s}{2^i}\right)\right).
$$
Now we can upper bound
$$
\E[R]\leq 2^i+\sum_{j=0}^\infty \Pr[A_{2^{i+j}}]2^{i+j+1}=2^i+2^{i+1}\sum_{j=1}^\infty \exp\left(-\Omega(2^{j})\right)2^j=O(2^i),
$$
as desired. This completes the proof.
\end{proof}

\section{Balls into Capacitated Bins}\label{sec:capacitated-bins}
In this section we prove~\Cref{thm:main0}. Let us start by recalling the setting of the theorem. We let $n,m\in \N$ and $\eps$ be given with $0<\eps <1$ and suppose that we sequentially distribute $n$ balls into $m$ bins, each of capacity $C=(1+\eps)n/m$. For simplicity, we assume that $n,m$ and $\eps$ are such that $C$ is a positive integer. Each ball is placed in a uniformly random non-full bin, where a bin is \emph{full} if it contains precisely $C$ balls. The theorem claims that if $1/\eps=m^{o(1)}$, then the expected fraction of non-full bins is $\Theta(f)$, where,
$$
f=\begin{cases}
\eps C, &C\leq \log(1/\eps) \\
\eps\sqrt{C\log\left(\frac{1}{\eps\sqrt C}\right)}, &\log(1/\eps)<C\leq  \frac{1}{2\eps^2} \\
1, &\frac{1}{2\eps^2}\leq C.
\end{cases}
$$

To prove the theorem, we will take an alternative viewpoint on the distribution process. Instead of picking a non-full bin for each ball, we disregard the capacities and instead pick a uniformly random bin (full or non-full). Then a bin may receive more than $C$ balls but if it does, we view it as having exactly $C$ balls. To be precise, for $j \in [m]$, and $i\in\Z_{\geq 0}$, we denote by $X_j^{(i)}$ the number of balls in bin $j$ after $i$ balls have been placed. We further define $Y_j^{(i)}=\min(X_j^{(i)},C)$. Let $T\in \N$ be minimal such that $\sum_{j\in [m]} Y_j^{(T)}=n$. Note that $T$ is a random variable with $T\geq n$ and that $\Pr[T<\infty]=1$. Further note that when the $n$ balls are distributed into the $m$ bins as in~\Cref{thm:main0}, the joint distribution of balls in bins has the same distribution as $(Y_j^{(T)})_{i\in [m]}$. We will first prove concentration bounds on $T$ and for this, we require Azuma's inequality.

\begin{theorem}[Azuma's inequality~\cite{Azuma67}]
Suppose that $(X_i)_{i=0}^k$ is a martingale satisfying that $|X_{i+1}-X_{i}|\leq s_i$ almost surely for all $i=0,\dots,k-1$. Let $s=\sum_{i=1}^k s_i^2$. Then for any $t >0$ it holds that
\begin{align}
\Pr(|X_k-X_0|\geq t) \leq 2\exp \left( \frac{-t^2}{2s}\right).
\end{align}
 \end{theorem}
 
The concentration bound on $T$ is as in the following lemma.
\begin{lemma}\label{lemma:ballsconc}
For any $N\geq 2Cm$ and any $t>0$ it holds that
$$
\Pr[|T-\E[T]|\geq t]\leq 2\exp\left(- \frac{t^2\eps^2}{8N}\right)+m\exp(-N/(8m)).
$$
\end{lemma}
\begin{proof}
For $i\in \Z_{\geq 0}$, we define $S_i\in [m]$ to be the randomly chosen bin for the $i$'th ball. We further define $\mathcal{F}_i=\sigma(S_1,\dots,S_i)$ to be the $\sigma$-algebra generated by the random choices of bins for the first $i$ balls. Finally, we put $X_i=\E[T| \mathcal{F}_i]$. Then $(X_i)_{i=0}^\infty$ is a martingale with $X_0=\E[T]$. Now the random variable $X_i$ is the expected value of $T$ conditioned on the placements of the first $i$ balls. We are going to prove that $|X_{i+1}-X_i|\leq \frac{1+\eps}{\eps}$ for each $i\geq 0$. 
To see this, fix $i$ and write $n'=\sum_{j\in [m]} Y_j^{(i)}$. If $n'\geq n$, then $X_i=X_{i+1}=T$, so we may assume that $n'<n$, i.e., after distributing the first $i$ balls we are still not done distributing the $n$ balls into the capacitated bins. 
In this case, it trivially holds that $X_{i+1}\leq X_i+1$ with equality holding if and only if the $(i+1)$'st ball is placed in a bin which is already full. 
On the other hand, we claim that $X_i\leq X_{i+1}+\frac{1+\eps}{\eps}$. To see this, let $T'$ be minimal such that $\sum_{j\in [m]} Y_j^{(T')}=n-1$ and let $T_i'=\max (T',i)$. From the assumption $n'<n$ it follows that $T_i'<T$ and we may write
$$
X_i=\E[T_i'| \mathcal{F}_i]+\E[T-T_i'| \mathcal{F}_i].
$$
Consider now any sequence of ball placements $s=(s_1,\dots, s_\ell)\in [m]^\ell$ with $\ell>i$ satisfying that if $(S_1,\dots, S_\ell)=s$, then $T=\ell$. Then, for any $s'\in [m]^\ell$ differing from $s$ in at most the $(i+1)$'st coordinate, it holds that if $(S_1,\dots, S_\ell)=s'$, then $T_i'\leq \ell$. From this it follows that $\E[T_i'| \mathcal{F}_i]\leq X_{i+1}$. 
We further claim that $\E[T-T_i'| \mathcal{F}_i]\leq \frac{1+\eps}{\eps}$. To see this, note that when placing $n$ balls into $m$ bins of capacity $C=(1+\eps)n/m$, at least $\frac{\eps}{1+\eps}$ bins will be non-full regardless of the positions of the balls. Now $T-T_i'$ counts the number of times we have to select a random bin until we find a non-full bin. 
Therefore, $T-T_i'$ will be geometrically distributed with parameter $p\geq \frac{\eps}{1+\eps}$, and it follows that $\E[T-T_i'\mid \mathcal{F}_i]\leq \frac{1+\eps}{\eps}$. Combining our bounds, we conclude that 
$$
|X_{i+1}-X_i|\leq \frac{1+\eps}{\eps}\leq \frac{2}{\eps}.
$$
Plugging into Azuma's inequality, we see that for any $N\geq 0$ and any $t>0$, it holds that 
$$
\Pr[|X_N-\E[T]|\geq t]=\Pr[|X_N-X_0|\geq t]\leq 2\exp\left(- \frac{t^2\eps^2}{8N}\right).
$$
Thus, for any $N\geq 0$,
$$
\Pr[|T-\E[T]|\geq t]\leq \Pr[|X_N-\E[T]|\geq t]+\Pr[X_N\neq T]\leq 2\exp\left(- \frac{t^2\eps^2}{8N}\right)+\Pr[N<T].
$$
Suppose $N\geq 2Cm$. By a standard Chernoff bound it follows if $N$ balls are distributed at random into $m$ bins, the probability that a given bin receives less than $C$ balls is upper bounded by $\exp(-N/(8m))$. Thus, we can trivially upper bound $\Pr[N<T]\leq m\exp(-N/(8m))$. Combining our bounds,
$$
\Pr[|T-\E[T]|\geq t]\leq 2\exp\left(- \frac{t^2\eps^2}{8N}\right)+m\exp(-N/(8m)),
$$
as desired.
\end{proof}
Curiously,~\Cref{lemma:ballsconc} does not tell us anything about the value of $\E[T]$ and in fact, we will not need it when proving~\Cref{thm:main0}. The bound in~\Cref{lemma:ballsconc} is a bit unwieldy, so below we state a corollary which is better suited for applications.
\begin{corollary}\label{corrolary:ballsconc}
Let $\gamma=O(1)$. If $C>\frac{3(1+\eps)(1+\gamma) \log n}{\eps^2}$, then $\Pr[T=n]=1-O(n^{-\gamma})$. Otherwise $|T-\E[T]|=O\left(\frac{\sqrt{m}\log n }{\eps^2} \right)$ with probability $1-O(n^{-\gamma})$, where the implicit constant in the $O$-notation depends on $\gamma$.
\end{corollary}
\begin{proof}
Suppose first that $C>\frac{3(1+\eps)(1+\gamma) \log n}{\eps^2}$. Consider throwing $n$ balls into $m$ bins uniformly at random. Let $X$ denote the number of balls landing in a given bin and $\mu=\E[X]=C/(1+\eps)$. Then a standard Chernoff bound shows that
$$
\Pr[X> C]=\Pr[X> (1+\eps)\mu]\leq \exp(-\eps^2\mu/3)\leq n^{-\gamma-1},
$$
so the probability that any bin receives more than $C$ ball is $O(n^{-\gamma})$ by a union bound. In particular $T=n$ with probability $1-O(n^{-\gamma})$.

Suppose on the other hand that $C\leq\frac{3(1+\eps)(1+\gamma) \log n}{\eps^2} \leq \frac{6(1+\gamma) \log n}{\eps^2}$. Applying~\Cref{lemma:ballsconc} with $N=\max(2Cm,8m(\gamma+1)\log n)$, we obtain that 
$$
\Pr[|T-\E[T]|\geq t]=2\exp\left(- \frac{t^2\eps^2}{8 N}\right)+n^{-\gamma}.
$$
In particular $|T-\E[T]|=O(\sqrt{N\log n}/\eps)$ with probability $1-O(n^{-\gamma})$. The desired bound follows by observing that $N=O(\frac{m\log n}{\eps^2})$.
\end{proof}
We need one further Lemma before proving~\Cref{thm:main0}.
\begin{lemma}\label{lemma:ballsconc2}
Let $k\geq 0$ be fixed and define $Z=\sum_{j\in [m]}Y_j^{(k)}$. Then for any $t>0$,
$$
\Pr[|Z-\E[Z]|\geq t]\leq 2\exp\left(-\frac{t^2}{2k}\right).
$$
\end{lemma}
\begin{proof}
Let $S_1,S_2,\dots$ and $\mathcal{F}_1,\mathcal{F}_2,\dots$ be defines as in the proof of~\Cref{lemma:ballsconc}. For $0\leq i\leq k$, we define $Z_i=\E[Z\mid \mathcal{F}_i]$ so that $Z_0=\E[Z]$ and $Z_k=Z$. Now it is easy to check that for $0\leq i<k$ it holds that $|Z_{i+1}-Z_i|\leq 1$. Thus the desired result follows from Azuma's inequality.
\end{proof}

We will next prove~\Cref{thm:main0}.
\begin{proof}[Proof of~\Cref{thm:main0}]
Note first, that if $\eps=\Omega(1)$, then $f=\Theta(1)$, regardless of the relationship between $\eps$ and $C$. When placing $n$ balls into $m$ bins, each of capacity $C=(1+\eps)n/m$, the fraction of non-full bins is at least $\eps/(1+\eps)$, regardless where the balls are placed. In the case $\eps=\Omega(1)$, this is $\Theta(1)$, so~\Cref{thm:main0} is trivial. In the following, we may therefore assume that $\eps$ smaller than a sufficiently small constant. 

We will again consider the alternative viewpoint where we throw an infinite sequence of balls uniformly at random into the bins. As before, we define $X_{j}^{(i)}$ to be the number of balls in bin $j$ after throwing $i$ balls, $Y_j^{(i)}=\min(X_j^{(i)},C)$ and $T=\min(i\in \N:\sum_{j\in [m]}Y_j^{(i)}=n)$. 

Let $\gamma>1$ be a constant to be fixed. We are going to split the argument into three cases.
\paragraph{Case 1: $C\leq \gamma \log(1/\eps)$.} We will show that in this case, the expected fraction of non-full bins is $\Theta(\eps C)$. To do this, we first show the following technical claim.
\begin{claim}
If $C\leq \gamma \log(1/\eps)$, then $\E[T]=(1+\Omega(1))n$.
\end{claim} 
\begin{proof}[Proof of Claim]
Fix a bin $j\in [m]$ and consider throwing $m\log(1/\eps)/2$ balls into $m$ bins. The probability that bin $j$ is empty is
$$
\left(1-\frac{1}{m} \right)^{m\log(1/\eps)/2}=\Omega(\sqrt{\eps}).
$$
As we will now argue, it follows that when throwing $N\geq m\log(1/\eps)/2$ balls into $m$ bins uniformly at random, a given bin receives at most $N/m-\log(1/\eps)/2$ balls with probability $\Omega(\sqrt{\eps})$. For this, we use the results of~\cite{GREENBERG201491}, stating that if $W\sim B(k,p)$ is binomially distributed with $p<1-1/k$, then $\Pr[X\leq \E[X]]>1/4$. Combining this result with the above, we obtain that the given bin receives none of the first $m\log(1/\eps)/2$ balls with probability $\Omega(\sqrt{\eps})$ and at most $(N-m\log(1/\eps)/m=N/m-\log(1/\eps)/2$ of the remaining balls with probability at least $1/4$. Moreover, these events are independent, happening simultaneously with probability $\Omega(\sqrt{\eps})$, which gives the desired. 

 Now let $N=n+\log(1/\eps)m/4$ and define $Z=\sum_{j\in [m]}Y_j^{(N)}$. From the above observation, it follows that 
$$
\E[Z]\leq Cm-\Omega(\sqrt{\eps}\log(1/\eps)m),
$$
and by applying Lemma~\ref{lemma:ballsconc2} it follows that it similarly hold with high probability that $Z\leq Cm-\Omega(\sqrt{\eps}\log(1/\eps)m)$, with a potentially larger implicit constant in the $\Omega$-notation. Assuming that $\eps$ is smaller than a sufficiently small constant we therefore have that  with high probability,
$$
Z\leq (C-\gamma \eps \log(1/\eps))m\leq C(1-\eps)m=(1+\eps)(1-\eps)n<n.
$$
Thus $T>N$ with high probability, but this also means that 
$$
\E[T]\geq N=n+\frac{\log(1/\eps)m}{4}\geq n+\frac{Cm}{4\gamma}=n\left(1+\frac{1+\eps}{4\gamma}\right)=n(1+\Omega(1)),
$$
as desired.
\end{proof}
Using the claim and~\Cref{corrolary:ballsconc} it follows that also $T=(1+\Omega(1))n$ with probability $1-n^{-\gamma}$ for any constant $\gamma$ and that $|T-\E[T]|=O\left(\frac{\sqrt{m}\log n }{\eps^2} \right)$ with the same high probability. 

 We now choose $N=\E[T]+O\left(\frac{\sqrt{m}\log n }{\eps^2} \right)$ so large that $\Pr[T\geq N]\leq n^{-2}$. Then $N=(1+\Omega(1))n$ as well. 
Consider a bin $j\in [m]$ and let $A_k=[X_j^{(N)}=k]$ for each $k\geq 0$. Then
$$
\Pr[A_k]=\binom{N}{k}\frac{1}{m^k}\left(1-\frac{1}{m} \right)^{N-k}.
$$
If $k=N^{1/2-\Omega(1)}$, then simple calculus yields that $\Pr[A_k]$ can be approximated with the Poisson distribution with mean $\mu=N/m$ as follows,
$$
\Pr[A_k]=(1+o(1)) \left( \frac{N}{m}\right)^k \frac{1}{k!}e^{-N/m}=(1+o(1)) \frac{\mu^k}{k!}e^{-\mu}.
$$
In particular, this holds when $k\leq C$. Thus, for any $k\leq C$ it holds that
$$
\frac{\Pr[A_k]}{\Pr[A_{k-1}]}=(1+o(1)) \frac{\mu}{k}=(1+\Omega(1))\frac{n}{km}\geq (1+\Omega(1))\frac{n}{Cm}=\frac{1+\Omega(1)}{1+\eps}=1+\Omega(1),
$$
where the last inequality requires that $\eps$ is smaller than a sufficiently small constant which we may assume. Let $\alpha=\Omega(1)$ be the implicit constant in the $\Omega$-notation above, such that for $k\leq C$ (and $n,m,1/\eps$ sufficiently large), we have that $\Pr[A_k]/\Pr[A_{k-1}]\geq 1+\alpha$. It follows that,
$$
\Pr[Y_j^{(N)}< C]=\sum_{k=1}^{C}\Pr[A_{C-k}]\leq \sum_{k=1}^{C}(1+\alpha)^{k-1} \Pr[A_{C-1}]=O(\Pr[A_{C-1}]),
$$ 
and
$$
\E[C-Y_j^{(N)}]=\sum_{k=1}^{C}k\Pr[A_{C-k} ]\leq \sum_{k=1}^{C}k(1+\alpha)^{k-1} \Pr[A_{C-1}]=O(\Pr[A_{C-1}]).
$$
It trivially holds that $\Pr[Y_j^{(N)}< C]\geq \Pr[A_{C-1}]$ and $\E[C-Y_j^{(N)}]\geq \Pr[A_{C-1}]$, so in fact we have proved that $\Pr[Y_j^{(N)}< C]=\Theta( \Pr[A_{C-1}])$ and $\E[C-Y_j^{(N)}]=\Theta(\Pr[A_{C-1}])$. By linearity of expectation,
\begin{align}\label{eq:quickdecay}
\E \left[\sum_{j\in [m]}C-Y_j^{(N)}\right]=\Theta(m\Pr[A_{C-1}])=\Theta(m\Pr[Y_j^{(N)}< C]).
\end{align}
Now with probability at least $1-n^{-2}$, it holds that $N-O\left(\frac{\sqrt{m}\log n }{\eps^2} \right)\leq T \leq N$. Since $T$ is chosen such that $\sum_{j\in [m]}C-Y_j^{(T)}=\eps n$, it follows that 
\begin{align}\label{eq:numberofholes}
\E \left[\sum_{j\in [m]}C-Y_j^{(N)}\right]=\Theta(\eps n).
\end{align}
Thus, combining~\eqref{eq:quickdecay} and~\eqref{eq:numberofholes}, we obtain that $\Pr[Y_j^{(N)}< C]=\Theta(\eps C)$. Finally,
$$
\Pr[Y_j^{(T)}< C]\geq \Pr[Y_j^{(N)}< C]-\Pr[N<T]=\Omega(\eps C)-n^{-2}=\Omega(\eps C).
$$
Using the exact same argument but instead choosing $N=\E[T]-O\left(\frac{\sqrt{m}\log n }{\eps^2} \right)$ so small that $\Pr[T\leq N]\leq n^{-2}$, we obtain that $\Pr[Y_j^{(T)}< C]=O(\eps C)$, so in fact $\Pr[Y_j^{(T)}< C]=\Theta(\eps C)$. But $\Pr[Y_j^{(T)}< C]$ is independent of $j$ and is exactly the expected fraction of non-full bins. Thus the proof is complete in the case $C\leq \gamma \log(1/\eps)$. 
\paragraph{Case 2: $\gamma \log(1/\eps)<C\leq \frac{1}{\gamma\eps^2}$.}
To make the argument work, we will assume that $\gamma=O(1)$ is sufficiently large. We can make this assumption since the argument from case 1 holds for any $\gamma=O(1)$.  In general, the argument from case 1 serves as a nice warm up but for the present case we have to be more careful in our estimates. Again, we choose $N=\E[T]+O\left(\frac{\sqrt{m}\log n }{\eps^2} \right)$ so large that $\Pr[T\geq N]\leq n^{-2}$ and put $\mu=N/m$. Let us state by proving some crude bounds on $N$ as stated in the following claim.
\begin{claim}
If $\gamma=O(1)$ is sufficiently large, then  $C+\sqrt{C}\leq N/m\leq 3C/2$.
\end{claim}
\begin{proof}
We first prove the lower bound. Suppose for contradiction that $N/m<C+\sqrt{C}$. Then $\E\left[\sum_{j\in [m]}C-Y_j^{(N)}\right]=\Omega(\sqrt{C}m)=\Omega(n/\sqrt{C})=\Omega(n\gamma \eps)$, so if $\gamma$ is sufficiently large, $\E\left[\sum_{j\in [m]}C-Y_j^{(N)}\right]\geq 2\eps n$ and this contradicts the fact that with high probability $T\leq N$. For the upper bound, note that if $N/m\geq 3C/2$, then for any $j\in [m]$,
$$
\Pr[X_j^{(N)}\leq C]\leq \exp\left( -\frac{N}{18m} \right)\leq \exp\left(-\frac{C}{12}\right)\leq \exp\left(-\frac{\gamma \log (1/\eps))}{12}\right)=\eps^{\gamma/12}\leq \eps^2
$$
by a Chernoff bound and assuming $\gamma\geq 24$. Thus, $\E\left[\sum_{j\in [m]}C-Y_j^{(N)}\right]\leq \eps^2 Cm\leq \eps n/2$, where the last inequality assumes that $\eps$ is sufficiently small. Again this contradicts the fact that with high probability $T\geq N-O\left(\frac{\sqrt{m}\log n }{\eps^2} \right)$
\end{proof}

As before, we consider a bin $j\in [m]$ and define $\Pr[A_k]=\Pr[X_j^{(N)}=k]$. Then for $k\leq C$,
$$
\frac{\Pr[A_k]}{\Pr[A_{k-1}]}=\frac{\binom{N}{k}}{\binom{N}{k-1}}\frac{1}{m-1}=\frac{N-k+1}{k}\frac{1}{m-1}=\frac{\mu}{k} \frac{m}{m-1}\frac{N-k+1}{N}=\frac{\mu}{k}\left(1\pm O(1/m) \right).
$$
 It follows from the claim that $\mu/k\geq 1+1/\sqrt{C}$ for $k\leq C$. By our assumptions $C\leq 1/\eps^2=m^{o(1)}$ and thus $\Pr[A_k]/\Pr[A_{k-1}]=(\mu/k)^{1\pm o(1)}$. Let $\alpha \in \N$ be minimal satisfying that $\Pr[A_{C-1}]/\Pr[A_{C-\alpha}]\geq 2$. Using the crude bounds in the claim and simple calculations we obtain that $\alpha=\Theta( 1/\log (\mu/C))$. Now,
\begin{align}\label{eq:interval1}
\Pr[Y_j^{(N)}< C]=\Theta(\alpha \Pr[A_{C-1}])=\Theta(\alpha \Pr[A_{C}]),
\end{align}
and
\begin{align}\label{eq:interval2}
\E[C-Y_j^{(N)}]=\Theta(\alpha^2 \Pr[A_{C-1}])=\Theta(\alpha^2 \Pr[A_{C}]).
\end{align}
As in case 1, $\Pr[Y_j^{(T)}< C]=\Theta(\Pr[Y_j^{(N)}< C])$ which is the the value we are looking for. Thus, if we can find the value of $\alpha$,~\cref{eq:interval1} will give us the result we are looking for. The problem is that $\alpha$ depends of $N$ and hence of $\E[T]$ which we as of now don't know the value of. However, we know that $\E[C-Y_j^{(N)}]$ is close to $\eps n$, so on a high level we can plug this into~\cref{eq:interval2} and solve for $\alpha$.

Let us make the above argument precise. First, we write $\mu=C+\beta$ noting that by the claim, $\sqrt{C}\leq \beta\leq C/2$. Note for later use that $\alpha=\Theta\left( \frac{1}{\log (\mu/C)}\right)=\Theta\left(\frac{C}{\beta}\right)$.
 Using the Poisson approximation,
$$
\Pr[A_C]=(1+o(1)) \frac{\mu^C}{C!}e^{-\mu}=\Theta\left( \left( \frac{\mu}{C} \right)^C \frac{1}{\sqrt{C}e^{\beta}} \right)=\Theta\left( \left(1+ \frac{\beta}{C} \right)^C \frac{1}{\sqrt{C}e^{\beta}} \right).
$$
Write $f(x)=\log \left(1+x \right)$, so that $\exp(f(\beta/C))=1+ \beta/C$. As $\beta/C\leq 1/2$, we can use a Taylor expansion to conclude that 
$$
f\left(\frac{\beta}{C}\right)=f(0)+f'(0)\frac{\beta}{C}-\Theta\left(f''(0) \left(\frac{\beta}{C}\right)^2\right)=\frac{\beta}{C}-\Theta\left(\left(\frac{\beta}{C}\right)^2\right).
$$
Write $\Delta=\beta-Cf(\beta/C)$, so that $\Delta=\Theta(\beta^2/C)=\Theta(C/\alpha^2)$. Then
$$
\Pr[A_C]=\Theta\left(\frac{1}{\sqrt{C}e^{\Delta}} \right).
$$
On the other hand, it follows from~\Cref{corrolary:ballsconc} that with high probability
$$
\sum_{j\in [m]}C-Y_j^{(N)}=\Theta\left(\sum_{j\in [m]}C-Y_j^{(T)}\right)=\Theta(\eps n),
$$
so that, $\E[C-Y_j^{(N)}]=\Theta(\eps C)$. Plugging all this into~\cref{eq:interval2}, we find that
$$
\frac{\alpha^2}{\sqrt{C}e^{\Delta}}=\Theta(\eps C).
$$
Using that $\alpha^2=\Theta(C/\Delta)$, this reduces to $\Delta e^\Delta =\Theta\left( \frac{1}{\eps \sqrt{C}} \right)$,
so that $\Delta =\Theta\left( \log\left( \frac{1}{\eps \sqrt{C}} \right)\right)$, and thus,
$$
\alpha=\Theta\left(\sqrt{C/\log\left( \frac{1}{\eps \sqrt{C}} \right)} \right).
$$
Combining~\cref{eq:interval1} and~\cref{eq:interval2}, we find that,
$$
\Pr[Y_j^{(N)}< C]=\Theta(\E[C-Y_j^{(N)}]/\alpha)=\Theta\left(\eps \sqrt{C} \sqrt{\left(\log \frac{1}{\eps \sqrt{C}}\right)}\right).
$$
A similar argument to that used in the first case shows that also $\Pr[Y_j^{(T)}< C]=\Theta(\Pr[Y_j^{(N)}< C])=\Theta(f)$ which completes the proof.

\paragraph{Case 3: $C>\frac{1}{\gamma\eps^2}$.} We can reduce this case to case $2$ as follows. Define the function, $f:\R\to\R$ by $f(x)=x(x-n/m)^2$. Then $f(n/m)=0$ and $f(C)=C\eps^2 (n/m)^2> \frac{(n/m)^2}{\gamma}$, so there exists $n/m<\hat C<C$ satisfying that $f(\hat C)=\frac{(n/m)^2}{\gamma}$. Let $\hat \eps$ be such that $\hat C=(1+\hat \eps)n/m$, so that $0<\hat \eps <\eps$. Then $f(\hat C)=\hat C \hat \eps^2 (n/m)^2$ which implies that $\hat C=\frac{1}{\gamma\hat \eps^2}$.
Now define $\hat{Y}_j^{(i)}=\min(X_j^{(i)},\hat C)$ and $\hat T=\min(i \in \N : \sum_{j\in [m]}\hat{Y}_j^{(i)}=n)$. As $\hat C\leq C$, it follows that $\hat T\geq T$. We can now apply the result from Case 2 to conclude that 
$$
\Pr[Y_j^{(T)}<C]\geq \Pr[Y_j^{(\hat T)}<C]=\Omega(1),
$$
which completes the proof.
\end{proof}

\section{Some Helpful Lemmas}\label{sec:helpful-lemmas}
In this section, we provide two helpful lemmas which will be useful in several of the later sections. The first is a tail bound for sums of geometric random variables, and the second can be seen as a high probability upper bound on the number of bins visited at a given level during an insertion with consistent hashing with bounded loads and virtual bins.

\subsection{A Tail Bound for Sums of Geometric Variables}\label{sec:geo-tail-bound}
 Recall that we say that $Y$ is geometrically distributed with parameter $p$ if for non-negative integers $k$ it holds that $\Pr[Y=k]=p^k(1-p)$. Then $\E[Y]=\frac{p}{1-p}$ and $\Var[Y]=\E[Y](1+\E[Y])$. Let $(X_i)_{i\in [n]}$ be independent random variables such that $X_i$ is geometrically distributed with parameter $p_i$. Let $X=\sum_{i\in [n]}X_i$. Define $\mu_i=\E[X_i]$, $\sigma_i^2=\Var[X_i]=\mu_i(1+\mu_i)$, $\mu=\sum_{i\in [n]}\mu_i$, and $\sigma^2=\sum_{i\in [n]}\sigma_i^2$. Finally let $W_0:[0,\infty)\to \R$ be the Lambert function defined by $W_0(x)e^{W_0(x)}=x$. We have the following theorem.
\begin{theorem}\label{thm:geotail}
For any $t\geq 0$ it holds that 
\begin{align}\label{eq:geotail}
    \Pr\left[X \ge \mu + 4 t \sigma^2\right]
        \le \begin{cases}
            e^{-2 \sigma^2 t W_0(t)},  & \text{if $t \le \left(1 + \tfrac{1}{2 \mu_{0}}\right) \log \left(1 + \tfrac{1}{2\mu_0} \right)$} \\
            \left(1 - \tfrac{1}{1 + 2 \mu_0} \right)^{2 \sigma^2 t}, & \text{if $t > \left(1 + \tfrac{1}{2 \mu_{0}}\right) \log \left(1 + \tfrac{1}{2\mu_0} \right)$}
        \end{cases}
        \; .
\end{align}
\end{theorem}
\begin{proof}
The idea of the proof is standard and uses the moment generating function of $X$.
    Let $0 \le \lambda \le \log\left(1 + \frac{1}{2 \mu_0}\right)$ be a parameter which we will fix later. Then
    \begin{align*}
        \ep{e^{\lambda (X_i - \mu_i)}}
            = \frac{e^{-\lambda \mu_i}}{1 - \mu_i(e^{\lambda} - 1)}
            = e^{-\lambda \mu_i - \log \left(1 - \mu_i(e^{\lambda} - 1) \right)}
        \; .
    \end{align*}
    Define $f(\lambda) = -\lambda \mu_i - \log \left(1 - \mu_i(e^{\lambda} - 1) \right)$. Using a Taylor expansion,
    \begin{align*}
        f(\lambda)
            \le f(0) + f'(0) \lambda + \frac{\max_{0 \le x \le \lambda} f''(x)}{2} \lambda^2
    \end{align*}
    It is easy to check that $f(0) = 0$, $f'(0) = 0$, and $f''(\lambda) = \sigma_i^2 \frac{e^{\lambda}}{(1 - \mu_i(e^{\lambda} - 1))^2}$.
    Now using that $\lambda \le \log\left(1 + \frac{1}{2 \mu_0}\right)$ we get that $f''(\lambda) \le 4 \sigma_i^2 e^{\lambda}$, and hence
    \begin{align*}
        f(\lambda)
            \le 2\sigma_i^2 e^{\lambda} \lambda^2.
    \end{align*}
    We now use Markov's inequality to conclude that 
    \begin{align*}
        \prb{X \ge \mu + 4 t \sigma^2}
            &= \prb{e^{\lambda\sum_{i \in [n]} (X_i - \mu_i)} \ge e^{\lambda 4 t \sigma^2}}
            \\&\le \frac{\prod_{i \in [n]} \ep{e^{\lambda(X_i - \mu_i)}}}{e^{\lambda 4 t \sigma^2}}
            \\&\le e^{2 \sigma^2 e^{\lambda} \lambda^2 - 4 \lambda t \sigma^2}
        \; .
    \end{align*}
    We will set $\lambda = \min\set{\log\left(1 + \tfrac{1}{2\mu_0}\right), W_0(t)}$.
    Now, if $t \le \left(1 + \tfrac{1}{2\mu_0} \right) \log\left( 1 + \tfrac{1}{2\mu_0} \right)$ then $\lambda = W_0(t)$.
    This implies that,
    \[
        \prb{\sum_{i \in [n]} X_i \ge \mu + 4 t \sigma^2}
            \le e^{2 \sigma^2 e^{\lambda} \lambda^2 - 4 \lambda t \sigma^2}
            = e^{-2 \sigma^2 t W_0(t)}
        \; .
    \]
    On the other hand, if $t > \left(1 + \tfrac{1}{2\mu_0} \right) \log\left( 1 + \tfrac{1}{2\mu_0} \right)$ then $\lambda = \log\left(1 + \tfrac{1}{2\mu_0}\right)$.
    This implies that,
    \begin{align*}
        \prb{\sum_{i \in [n]} X_i \ge \mu + 4 t \sigma^2}
            \le e^{2 \sigma^2 e^{\lambda} \lambda^2 - 4 \lambda t \sigma^2}
            \le e^{-2 \sigma^2 \log \left(1 + \tfrac{1}{2\mu_0} \right) t }
            = \left(1 - \frac{1}{1 + 2\mu_0} \right)^{2\sigma^2 t}
    \end{align*}
\end{proof}
Defining $\mathcal{C}:[0,\infty)\to \R$ by $\mathcal{C}(x)=(1+x)\log(1+x)-x$, it follows from standard calculus that $\mathcal{C}(x)=\Theta(xW_0(x))$. In particular, the first bound in~\eqref{eq:geotail} takes the form 
$$
\Pr\left[X \ge \mu + 4 t \sigma^2\right]=e^{-\Omega(\sigma^2\mathcal{C}(t))}.
$$
Up to the constant delay in the exponential decrease, this is the same as the standard variance-based Chernoff bound for the sum of independent variables in $[0,1]$. Intuitively, the second bound of~\eqref{eq:geotail} corresponds to the event that the heaviest of the geometric variables, $X_0$, satisfies $X_0=\mu_0+\Omega(\sigma^2 t)$.

\subsection{A High Probability Upper Bound on the Run Length at a Level}
We next prove the general lemma on consistent hashing with bounded loads and virtual bins. Consider a bin $b$ at level $i$ that may be chosen dependently on the hashing of balls and bins to levels $1,\dots,i-1$. We prove that if $I$ is a maximal interval of level $i$ containing $b$ satisfying that all bins in $I$ get full after the hashing of balls to levels $1,\dots,i$, then with probability $1-\delta$ the number of bins in $I$ is $O(\log(1/\delta)/\eps)$. This bound is quite crude but we require it for both the analyses of Section~\ref{sec:f-properties} and Section~\ref{sec:insertion} which proceed by step by step revealing the history of how a bin obtained its balls at a given level. The result entails that at a given point in the process, we have only revealed an insignificant part of the system. On a high level, this means that even conditioning on what we already know about the system, the probabilities of the various relevant events only change very slightly. The result is as follows.
\begin{lemma}\label{thm:worstcase1}
Let $\balls,\bins\in \N$ and $0<\eps<1$ with $1/\eps=n^{o(1)}$.  Suppose we distribute $\balls$ balls into $\bins$ bins, each of capacity $C=(1+\eps)\balls/\bins$, using consistent hashing with bounded loads and virtual bins and $k\geq 2/\eps$ levels.  Let $b$ be a bin at level $i$ which may be chosen dependently on the hashing of balls and bins to level $1,\dots,i-1$. Let $I$ be a maximal interval at level $i$ containing $b$ such that all bins lying in  $I$ are full after the hashing to level $1,\dots,i$. Let $1/\balls^{O(1)}<\delta \leq 1/2$. The number of bins in $I$ is $O(\log (1/\delta)/\eps)$ with probability at least $1-\delta$.
\end{lemma}
\begin{proof}
The proof is very similar to the proof of Theorem~\ref{thm:main_combinatorial}, so we just provide a sketch of the proof. We may clearly assume that $i=k$ as this can only decrease the remaining capacities of the bins. Let $C_1,\dots,C_m$ be the remaining capacities and $\overline{C}=\frac{1}{m}\sum_{i\in [m]}C_i$. Using a standard Chernoff bound and the assumptions that $k\geq 2/\eps$ and $1/\eps=n^{o(1)}$, we find that the number of balls hashing to level $k$ is at most $\eps n$ with probability $1-O(n^{-\gamma})$ for any $\gamma=O(1)$. Letting $X$ denote the number of such balls, it follows that $m\overline{C}\geq 2X$ with the same high probability. Condition on this event and let $R$ denote the number of bins in $I$.  For a given $s$, we find as in the proof of Theorem~\ref{thm:main_combinatorial}, that  there exists a constant number of intervals $I_1,\dots,I_\ell$, all of length $\Theta(s/m)$ Such that the following holds.  Let $X_j^{(1)}$, $X_j^{(2)}$, and $X_j^{(3)}$ denote respectively the number of bins, total capacity of bins, and number of balls hashing to $I_j$. Let further  $\mu_j^{(1)}=\E[X_j^{(1)}]$, $\mu_j^{(2)}=\E[X_j^{(2)}]$, and $\mu_j^{(3)}=\E[X_j^{(3)}]$.
If $s+1\leq R \leq 2s$, then there is a $j\in [\ell]$ such that either 
\begin{itemize}
\item[$B_1$:] $|X_j^{(1)}-\mu_j^{(1)}|=\Omega(\mu_j^{(1)})$,
\item[$B_2$:] $|X_j^{(2)}-\mu_j^{(2)}|=\Omega(\mu_j^{(2)})$,
\item[$B_3$:] $X_j^{(3)}-\mu_j^{(3)}=\Omega(\max(\mu_j^{(3)},\overline{C}s))$.
\end{itemize}
Note that $\mu_j^{(1)}=\Theta(s)$ and $\mu_j^{(2)}=\Theta(\overline{C}s)$. It therefore follows from standard Chernoff bounds that $\Pr[B_1]=\exp(-\Omega(s))$, $\Pr[B_2]=\exp(-\Omega(\overline {C}s/C))$, and $\Pr[B_3]=\exp(-\Omega(\overline {C}s))$. As $C=(1+\eps)n/m$, we always have that $\overline{C}\geq \eps n/m\geq \eps C/2$. Therefore, we obtain the combined bound
$$
\Pr[s+1\leq R \leq 2s]=\exp(-\Omega(\eps s)).
$$
With $t=O(\log (1/\delta)/ \eps)$ sufficiently large, it follows that
$$
\Pr[R'\geq t+1]\leq \sum_{i=0}^\infty \Pr[A_{t2^i}] \leq \sum_{i=0}^\infty \exp(-\Omega(t2^i\eps)) \leq \sum_{i=0}^\infty \exp(-\log (2/\delta )2^i)\leq \sum_{i=0}^\infty (\delta /2)^{2^i}\leq \delta.
$$
This completes the proof.
\end{proof}
\begin{remark}
We will use the bound of Lemma~\ref{thm:worstcase1} to obtain the results in Section~\ref{sec:f-properties} showing the concentration of the fraction of non-full bins around  its mean $\mu=\Theta(f)$. In fact, this allows us to prove a stronger version of Lemma~\ref{thm:worstcase1} in Section~\ref{sec:insertion} which bounds the number of bins in $I$ by $O(\log(1/\delta)/f)$, the only caveat being that here we have to use $k\geq 1/\eps^2$ levels.
\end{remark}

We finish the section with the following definition.
\begin{definition}\label{def:run}
For $I$ as in the lemma above, we will call $I$ the \emph{run} at level $i$ containing $b$. 
\end{definition}
Lemma~\ref{thm:worstcase1} shows that the number of bins in the run is $O(\frac{\log (1/\delta)}{\eps})$ with probability $1-\delta$.  It in particular follows that the number of bins visited at level $i$ during an insertion is $O(\frac{\log (1/\delta)}{\eps})$ with probability $1-\delta$. Indeed, if $b$ is a bin encountered during the insertion which is not full at level $i-1$, then all the bins encountered at level $i$ lie in the run at level $i$ containing $b$.

\section{Non-Full Bins: In Expectation and with Concentration}\label{sec:f-properties}
In this section we will show that with consistent hashing, for each level $d\in [k]$ and each and each $0\leq t\leq C$, the number of bins at level $d$ containing at most $t$ balls is sharply concentrated around its mean.
This goal is achieved in~\Cref{sec:f-concentration}.
Next, in~\Cref{sec:f-bound} we prove that with $k\geq c/\eps^2$ levels for a sufficiently large constant $c$, the expected number of non-full bins at the highest level $k-1$, is $\Theta(f)$ where $f$ is as defined in~\Cref{eq:formula-for-f}.

\subsection{High Probability Bounds on the Number of Non-Full Bins}\label{sec:f-concentration}
The goal of this section is to prove the first part of~\Cref{thm:jakobs}. For this, we first require some notation.
We define
\begin{itemize}
    \item[$X^{(j)}_d$] The remaining capacity of bin $j$ after distributing balls to all levels $i\leq d$.
    \item[$Y^{(j)}_d$] The number of balls landing in or forwarded by bin $j$ at level $d$.
    \item[$Z^{(j)}_{d, s}$] The capacity of the bin $s$ places before bin $j$ just before the hashing of balls to level $d$.
    \item[$W^{(j)}_{d, s}$] The number of balls landing between the bins placed $s$ and $s + 1$ places before bin $j$ at level $i$.
\end{itemize}
There are some important relations between the variables. $Y^{(j)}_d$ can be expressed in terms of $W^{(j)}_{d, s}$ and $Z^{(j)}_{d, s}$
as follows $Y^{(j)}_d = W^{(j)}_{d, 0} + \max\set{0, \max_{l \ge 1} \sum_{s = 1}^{l} (W^{(j)}_{d, s} - Z^{(j)}_{d, s} )}$.
Similarly, we can express $X^{(j)}_d$ in terms of $Y^{(j)}_d$ as follows $X^{(j)}_d = \max\set{0, C - \sum_{i \le d} Y^{(j)}_i}$.

Now due to all the dependencies in the system, it is unwieldy to analyse it directly.
Instead, we will analyse a simpler system which we then show can give us high probability bounds for $\sum_{j\in [m]}[X^{(j)}_d\leq t]$ for each $t\leq C$.
First we define $\mathcal{X}^{(j)}_0 = C$ for every bin $j$.
We then define $\mathcal{X}^{(j)}_d$ for $0 < d < k$ recursively as follows: 
First define independent random variables $\mathcal{Z}^{(j)}_{d, s}$ and $\mathcal{W}^{(j)}_{d, s}$ for every bin $j$ and every integer $s$ by
\begin{align}
    \prb{\mathcal{Z}^{(j)}_{d, s} = t_1} &= \prb{\mathcal{X}^{(j)}_{d - 1} = t_1} \\
    \prb{\mathcal{W}^{(j)}_{d, s} = t_2} &= \left( \frac{n/k}{n/k + m} \right)^{t_2} \frac{m}{n/k + m}
\end{align}
for every integers $0 \le t_1 \le C$ and $0 \le t_2$.
So $\mathcal{W}^{(j)}_{d, s}$ is geometrically distributed with parameter $\frac{n/k}{n/k + m}$.
We then define $\mathcal{Y}^{(j)}_d = \mathcal{W}^{(j)}_{d, 0} + \max\set{0, \max_{l \ge 1} \sum_{s = 1}^{l} (\mathcal{W}^{(j)}_{d, s} - \mathcal{Z}^{(j)}_{d, s} )}$ and finally $\mathcal{X}^{(j)}_d = \max\set{0, C - \sum_{i = 1}^{d} \mathcal{Y}^{(j)}_i }$.

Clearly, the two systems have a lot of similarities.
$\mathcal{X}^{(j)}_d$ and $\mathcal{Y}^{(j)}_d$ are defined analogously to how $X^{(j)}_d$ and $Y^{(j)}_d$ are defined.
The difference between the two system is the difference between variables the $\mathcal{Z}^{(j)}_{d, s}$, $\mathcal{W}^{(j)}_{d, s}$ and the variables $Z^{(j)}_{d, s}$, $W^{(j)}_{d, s}$.
Our goal is to show that two systems are in fact very comparable, yet leverage that the second system is much simpler to analyse due to the independence.
This approach leads to the theorem below which provides concentration of $\sum_{j\in [m]}[X^{(j)}_i\leq t]$ around  $m\prb{\mathcal{X}^{(j)}_i \le t}$.

\begin{theorem}\label{eq:concentration-on-full}
    Let $n$ and $m$ be positive integers and set $\mu = \frac{n}{m}$. Let $0 \le \eps \le 1$ be such that
    $C = (1 + \eps)\mu$ is in integer. If $\mu = m^{o(1)}$
    and $\eps = m^{o(1)}$, then with probability at least $1 - m^{-\gamma}$ we have that
    \begin{align}
        \abs{\frac{\sum_{j \in [m]} [X^{(j)}_i \le t]}{m} - \prb{\mathcal{X}^{(j)}_i \le t} }
            \le m^{-1/2 + o(1)}
    \end{align}
    for all levels $1 \le i \le k$ and all $0 \le t \le C$. The constant in the big-O notation
    depends on $\gamma$.
\end{theorem}

We define $A_d$ to be the event that 
\begin{align}
    \abs{\frac{\sum_{j \in [m]} [X^{(j)}_i \le t]}{m} - \prb{\mathcal{X}^{(j)}_i \le t} }
        \le m^{-1/2 + o(1)}
\end{align}
for all $1 \le i \le d$ and all $0 \le t \le C$.
The goal of \Cref{eq:concentration-on-full} is prove that $\prb{A_k} \ge 1 - m^{-\gamma}$.
An important step of the proof is the following lemma.

\begin{lemma}\label{lem:dependent-to-independent}
    Fix $0 \le t \le C$, $j \in [m]$, and a subset $S \subseteq [m] \setminus \set{j}$ of $l \le O(\log m)$ bins. Then
    \begin{align}
        \abs{
            \prbcond{Y^{(j)}_{d} \ge t}{A_{d - 1} \wedge (X^{(i)}_{d})_{i \in S}}
            - \prb{\mathcal{Y}^{(j)}_d \ge t}
        } \le m^{-1/2 + o(1)}
    \end{align}
\end{lemma}

We also need a couple of auxiliary lemmas.
The first is a simple consequence of \Cref{thm:worstcase1}:
\begin{lemma}\label{lem:longest-run}
    With probability at least $1 - m^{-\gamma}$ we have that the longest run at level $d$ is at most $O((\log m) / \eps)$. 
\end{lemma}

We will also need a bound on the number of balls between consecutive bins.
\begin{lemma}\label{lem:ball-consecutive-bins}
  With probability at least $1 - m^{-\gamma}$ there are no more than $O(\log m (\mu/k + 1) )$ balls between any two consecutive virtual bins on level $d$.
\end{lemma}
\begin{proof}
  This is simple observation since the probability that is there lands $l$ balls between consecutive virtual bins is at most
  \begin{align*}
    \prod_{i = 0}^{l - 1} \frac{n/k - i}{n/k - i + m}
      \le \left(1 - \frac{m}{n/k + m} \right)^l
      \le \exp\left(- \frac{m}{n/k + m} l \right)
  \end{align*}
  It is now clear that if $l = \Theta(\log m (\mu/k + 1) )$ that there are no consecutive
  virtual bins which receives more that $l$ balls with probability $1 - m^{-\gamma}$.
\end{proof}

The final lemma is a technical lemma which we will use to get tail bounds.
The proof is deferred to the end of the section.
\begin{lemma}\label{lem:dependent-concentration}
    Let $B_1, \ldots B_n$ be Bernoulli variables, $0 \le \delta \le 1$ a small real,
    and $r > 0$ be an even integer.
    Assume that for any $i \in [n]$ and any subset $S \subseteq [n] \setminus \set{i}$ of size at most $r - 1$
    we have that $\abs{\prbcond{B_i = 1}{(B_j)_{j \in S}} - p} \le \delta$, then
    \[
        \ep{\left( \sum_{i \in [n]} (B_i - p) \right)^r}^{1/r}
            \le \delta n + O(\sqrt{r n}) 
        \; ,
    \]
    and the following tail bound holds
    \begin{align}
        \prb{\abs{\sum_{i \in [n]} (B_i - p)} \ge \delta n + r \sqrt{n}}
            \le \exp(-\Omega(r))
        \; .
    \end{align}
\end{lemma}

We will now prove \Cref{lem:dependent-to-independent}.
\begin{proof}[Proof of \Cref{lem:dependent-to-independent}]
    Let $B_d$ be the event that the longest run on the level $d$ is at most $r = O(\log m / \prb{\mathcal{X}^{(j)}_d > 0}) \le O(\log(m)/\eps)$
    and that there are at most $O(\log m (\mu/k + 1) )$ balls between any two consecutive virtual bins on level $d$.
    By \Cref{lem:longest-run} and \Cref{lem:ball-consecutive-bins} we have that $\prbcond{\neg B_d}{A_{d - 1}} \le m^{-\gamma'}$ hence we get that
    \[
        \abs{
            \prbcond{Y^{(j)}_{d} \ge t}{A_{d - 1} \wedge (X^{(i)}_{d})_{i \in S}}
            - \prbcond{Y^{(j)}_{d} \ge t \wedge B_d}{A_{d - 1} \wedge (X^{(i)}_{d})_{i \in S}}
        } \le m^{-\gamma'}
    \]
    The important observation now is that 
    \[
        Y^{(j)}_d = W^{(j)}_{d, 0} + \max\set{0, \max_{1 \le l \le r} \sum_{s = 1}^{l} (W^{(j)}_{d, s} - Z^{(j)}_{d, s} )}
    \]
    when $B_d$ is true.
    So we only reveal $r$ virtual bins and at most $O(r O(\log m (\mu/k + 1) ))$ balls when determining $Y^{(j)}_d$.

    We will introduce a third system which will act as an intermediate between the two systems.
    Let $\overline{X}^{(j)}_{d - 1}$ be independent random variables where each of the variables has the same marginal distribution as $X^{(j)}_{d - 1}$.
    Let $\overline{Z}^{(j)}_{d, s}$ and $\overline{W}^{(j)}_{d, s}$ be independent random variables where each of $\overline{Z}^{(j)}_{d, s}$ has the same marginal distribution as $Z^{(j)}_{d, s}$, and each of $\overline{W}^{(j)}_{d, s}$ is geometrically distributed with parameter $\frac{n/k}{n/k + m}$.
    We then define $\overline{Y}^{(j)}_{d} = \overline{W}^{(j)}_{d, 0} + \max\set{0, \max_{1 \le l \le r} \sum_{s = 1}^{l} (\overline{W}^{(j)}_{d, s} - \overline{Z}^{(j)}_{d, s} )}$.
    The difference between the intermediate system and the original system is that in the intermediate system we are sampling everything with replacement and in the original system everything is sampled without replacement.

    Let $D$ be the event that $X^{(j)}_{d - 1}$ is a distinct bin from the bins $(Z^{(i)}_{d, s})_{i \in S, 1 \le s \le r}$
    and that the bins $(X^{(i)}_{d - 1})_{i \in S}$ are distinct for the bins $(X^{(j)}_{d, s})_{1 \le s \le r}$. It is easy
    to see that $\prbcond{\neg D \wedge B}{A_{d - 1} \wedge (X^{(i)}_{d})_{i \in S}} \le O(\frac{r l}{m - (l + 1)(r + 1)}) = m^{- 1 + o(1)}$,
    hence we get that
    \[
        \abs{
            \prbcond{Y^{(j)}_{d} \ge t \wedge B_{d}}{A_{d - 1} \wedge (X^{(i)}_{d})_{i \in S}}
            - \prbcond{Y^{(j)}_{d} \ge t \wedge B_{d} \wedge D}{A_{d - 1} \wedge (X^{(i)}_{d})_{i \in S}}
        } \le m^{-1 + o(1)}
        \; .
    \]

    It is standard fact that if we sample $\overline{X}^{(j)}_{d - 1}$ and $(\overline{Z}^{(j)}_{d, s})_{1 \le s \le r}$
    independently with replacement and condition on them sampling distinct bins which are also distinct from the bins
    for $(X^{(i)}_{d - 1})_{i \in S}$ and $(Z^{(i)}_{d, s})_{i \in S, 1 \le s \le r}$, then it has the same distribution
    as sampling without replacement. The probability that we make such a sampling error is bounded by
    $\frac{l(r + 1)^2}{m} = m^{-1 + o(1)}$.

    Similarly, if we sample $(\overline{W}^{(j)}_{d, s})_{0 \le s \le r}$ independently with replacement
    conditioned on all balls being distinct and distinct from the ball sampled for $(W^{(j)}_{d, s})_{i \in S, 0 \le s \le r}$,
    then it has the same distribution as sampling without replacement. With probability $1 - m^{-\gamma}$ we have that
    $\overline{W}^{(j)}_{d, s} \le O(\log m (\mu/k + 1) )$ for all $0 \le s \le r$, hence the probability of making a sampling error with balls is bounded
    by 
    \[
        m^{-\gamma} + O\left( \frac{\log(m)^2 (\mu/k + 1)^2 (r + 1)^2}{n_e} \right)
            = m^{-\gamma} + O\left( \frac{k \log(m)^2 (\mu/k + 1)^2 (r + 1)^2}{n} \right)
            = m^{-1 + o(1)}
        \; .
    \]

    From this two facts we see that
    \[
        \abs{\prbcond{Y^{(j)}_{d} \ge t \wedge B_{d} \wedge D}{A_{d - 1} \wedge (X^{(i)}_{d})_{i \in S}}
            - \prbcond{\overline{Y}^{(j)}_{d} \ge t \wedge B_{d} \wedge D}{A_{d - 1}}
        } \le m^{-1 + o(1)}
        \; .
    \]
    Since $B_{d} \wedge D$ happens with probability at least $1 - m^{-\gamma}$ we get that,
    \[
        \abs{\prbcond{\overline{Y}^{(j)}_{d} \ge t \wedge B_{d} \wedge D}{A_{d - 1}}
            - \prbcond{\overline{Y}^{(j)}_{d} \ge t}{A_{d - 1}}
        } \le m^{-1 + o(1)}
        \; .
    \]
    We then define $\overline{\mathcal{Y}}^{(j)}_d = \mathcal{W}^{(j)}_{d, 0} + \max\set{0, \max_{1 \le l \le r} \sum_{s = 1}^{l} (\mathcal{W}^{(j)}_{d, s} - \mathcal{Z}^{(j)}_{d, s} )}$.
    The difference between $\overline{\mathcal{Y}}^{(j)}_d$ and $\mathcal{Y}^{(j)}_d$ is that $\overline{\mathcal{Y}}^{(j)}_d$ looks at at most $r$ bins in the tail while $\mathcal{Y}^{(j)}_d$ looks at all bins in the tail.
    If $\overline{Z}^{(j)}_{d, s} = \mathcal{Z}^{(j)}_{d, s}$ for all $1 \le s \le r$ then $\prbcond{\overline{Y}^{(j)}_{d} \ge t}{A_{d - 1}} = \prb{\overline{\mathcal{Y}}^{(j)}_{d} \ge t}$.
    This observation imply that
    \begin{align*}
        &\abs{\prbcond{\overline{Y}^{(j)}_{d} \ge t}{A_{d - 1}} - \prb{\overline{\mathcal{Y}}^{(j)}_{d} \ge t}}
            \\&\quad\quad\quad\le \sum_{\tau_1, \ldots, \tau_r} \Big| \prbcond{(\overline{Z}^{(j)}_{d, 1}, \ldots, \overline{Z}^{(j)}_{d, r}) = (\tau_1, \ldots, \tau_r)}{A_{d - 1}} 
                - \prb{\mathcal{Z}^{(j)}_{d, 1}, \ldots, \mathcal{Z}^{(j)}_{d, r}) = (\tau_1, \ldots, \tau_r)} \Big|
            \\&\quad\quad\quad\le \sum_{\tau_1, \ldots, \tau_r} \Big| \prod_{1 \le s \le r} \prbcond{\overline{Z}^{(j)}_{d, s} = \tau_s}{A_{d - 1}} 
                - \prod_{1 \le s \le r} \prb{\mathcal{Z}^{(j)}_{d, s} = \tau_s} \Big|
            \\&\quad\quad\quad\le 2 \sum_{0 \le \tau \le C} \Big| \prbcond{\overline{Z}^{(j)}_{d, 1} = \tau}{A_{d - 1}} - \prb{\mathcal{Z}^{(j)}_{d, 1} = \tau} \Big|
            \\&\quad\quad\quad\le r \sum_{0 \le \tau \le C} m^{-1/2 + o(1)}
            \\&\quad\quad\quad\le m^{-1/2 + o(1)}
    \end{align*}

    Now the same arguments as in the proof of \Cref{lem:longest-run} show that $\mathcal{Y}^{(j)}_{d} = \overline{\mathcal{Y}}^{(j)}_{d}$
    with probability at least $1 - m^{-\gamma}$. Combining all these bounds proves the claim.
\end{proof}

Now having proved \Cref{lem:dependent-to-independent} we are ready to prove \Cref{eq:concentration-on-full}.

\begin{proof}[Proof of \Cref{eq:concentration-on-full}]
    We note that
    \[
        \prb{A_k}
            = \prb{\bigwedge_{1 \le i \le k} A_i}
            = \prod_{1 \le i \le k} \prbcond{A_i}{A_{i - 1}}
    \]
    If we can prove that $\prbcond{A_d}{A_{d - 1}} \ge 1 - O(m^{-\gamma'})$ for all $1 \le d \le k$, where $\gamma'$ is an appropriately chosen constant, then we would get that
    \[
        \prb{A_k} \ge (1 - m^{-\gamma'})^k \ge 1 - m^{-\gamma}
    \]
    The rest of the proof is now to show that $\prbcond{\neg A_d}{A_{d - 1}} \le O(m^{-\gamma'})$.
    
    Let $j \in [m]$ and $S \subseteq [m] \setminus \set{s}$ with $l = \abs{S} \le O(\log m)$.
    We will prove that,
    \begin{align}\label{eq:dep-to-indep}
        \abs{\prbcond{X^{(j)}_d \le t}{A_{d - 1} \wedge (X^{(i)}_d)_{i \in S}} - \prb{\mathcal{X}^{(j)}_d \le t}}
            \le m^{-1/2 + o(1)}
        \; .
    \end{align}
    This will imply the result since if we combine \cref{eq:dep-to-indep} with \Cref{lem:dependent-concentration} we get that,
    \begin{align*}
        &\prbcond{\abs{\frac{\sum_{j \in [m]} [X^{(j)}_d \le t]}{m} - \prb{\mathcal{X}^{(j)}_d \le t} } \ge m^{-1/2 + o(1)}}{A_{d - 1} }
            \le m^{-\gamma''}
        \; . 
    \end{align*}
    For $0 \le t \le C$. Now a union bound over all $0 \le t \le C$ gives us that
    \[
        \prbcond{\neg A_d \wedge B_{d}}{A_{d - 1}}
            \le (C + 1)m^{-\gamma''}
            \le m^{-\gamma'}
        \; .
    \]
    We then get that $\prb{A_k} \ge 1 - m^{-\gamma}$ as we wanted.

    We just need to prove \cref{eq:dep-to-indep}.
    We note that $X^{(j)}_d \le t$ if and only if $X^{(j)}_{d - 1} - Y^{(j)}_{d} \le t$.
    We thus get that,
    \begin{align*}
        \prbcond{X^{(j)}_d \le t}{A_{d - 1} \wedge (X^{(i)}_d)_{i \in S}}
            &= \sum_{s = 0}^{t - 1} \prbcond{X^{(j)}_{d - 1} \le t - s \wedge Y^{(j)}_d = s}{A_{d - 1} \wedge (X^{(i)}_d)_{i \in S}}
                \\&+ \prbcond{Y^{(j)}_d \ge t}{A_{d - 1} \wedge (X^{(i)}_d)_{i \in S}}
    \end{align*}
    If we fix the first $d - 1$ levels then we get that
    \begin{align*}
        \prbcond{X^{(j)}_{d - 1} \le t - s}{(X^{(i)}_d)_{i \in S}}
            = \frac{\sum_{i \in [m] \setminus S} [X^{(j)}_{d - 1} \le t - s]}{m - \abs{S}}
    \end{align*}
    We condition on $A_{d - 1}$ so we know that,
    \begin{align*}
        \abs{ \frac{\sum_{i \in [m]} [X^{(i)}_{d - 1} \le t - s]}{m} - \prb{\mathcal{X}^{(j)}_{d - 1} \le t - s} } \le m^{-1/2 + o(1)}            
    \end{align*}
    This implies that,
    \begin{align*}
        \abs{ \frac{\sum_{i \in [m] \setminus S} [X^{(i)}_{d - 1} \le t - s]}{m - \abs{S}} - \prb{\mathcal{X}^{(j)}_{d - 1} \le t - s} } 
            &\le m^{-1/2 + o(1)} + \abs{\frac{\sum_{i \in S} \left([X^{(i)}_{d - 1} \le t - s] - \prb{\mathcal{X}^{(j)}_{d - 1} \le t - s} \right)}{m - \abs{S}}} 
            \\&\le m^{-1/2 + o(1)}
    \end{align*}
    Here we have used that $\abs{S} \le O(\log(m))$.
    We thus get that 
    \begin{align*}
        &\abs{\prbcond{X^{(j)}_{d - 1} \le t - s \wedge Y^{(j)}_d = s}{A_{d - 1} \wedge (X^{(i)}_d)_{i \in S}}
            - \prb{\mathcal{X}^{(j)}_{d - 1} \le t - s} \prbcond{Y^{(j)}_d = s}{A_{d - 1} \wedge (X^{(i)}_d)_{i \in S}}}
            \\&\qquad\qquad\qquad\le \prbcond{Y^{(j)}_d = s}{A_{d - 1} \wedge (X^{(i)}_d)_{i \in S}} m^{-1/2 + o(1)}
    \end{align*}
    Using this we get that,
    \begin{align*}
        \abs{\prbcond{X^{(j)}_d \le t}{A_{d - 1} \wedge (X^{(i)}_d)_{i \in S}} - \prbcond{\mathcal{X}^{(j)}_{d - 1} - Y^{(j)}_d \le t}{A_{d - 1} \wedge (X^{(i)}_d)_{i \in S}}} \le m^{-1/2 + o(1)}
    \end{align*}
    We now want to exchange $Y^{(j)}_d$ with $\mathcal{Y}^{(j)}_d$ and the approach is similar to what we just did.
    We note that,
    \begin{align*}
        \prbcond{\mathcal{X}^{(j)}_{d - 1} - Y^{(j)}_d \le t}{A_{d - 1} \wedge (X^{(i)}_d)_{i \in S}}
            &= \sum_{s = 0}^{t - 1} \prb{\mathcal{X}^{(j)}_{d - 1} = s} \prbcond{Y^{(j)}_d \le t - s}{A_{d - 1} \wedge (X^{(i)}_d)_{i \in S}}
                \\&+ \prb{\mathcal{X}^{(j)}_{d - 1} \ge t}
    \end{align*}
    We now use \Cref{lem:dependent-to-independent} to get that,
    \[
        \abs{\prbcond{Y^{(j)}_d \le t - s}{A_{d - 1} \wedge (X^{(i)}_d)_{i \in S}} - \prb{\mathcal{Y}^{(j)}_d \le t - s} }
            \le m^{-1/2 + o(1)} + m^{-\gamma}
    \]
    So 
    \[
        \abs{\prb{\mathcal{X}^{(j)}_{d - 1} = s} \prbcond{Y^{(j)}_d \le t - s}{A_{d - 1} \wedge (X^{(i)}_d)_{i \in S}}
            - \prb{\mathcal{X}^{(j)}_{d - 1} = s \wedge \mathcal{Y}^{(j)}_d \le t - s}}
        \le \prb{\mathcal{X}^{(j)}_{d - 1} = s} m^{-1/2 + o(1)}
    \]
    This implies that,
    \begin{align*}
        \abs{\prbcond{X^{(j)}_d \le t}{A_{d - 1} \wedge (X^{(i)}_d)_{i \in S}} - \prb{\mathcal{X}^{(j)}_d \le t}}
            &\le m^{-1/2 + o(1)}
    \end{align*}
    Where we have use that $t \le C \le m^{o(1)}$.
    This proves \cref{eq:dep-to-indep} and thus finishes the proof.
\end{proof}

Later in the paper we will need to bound the contribution to a bin while fixing the previous levels.
The proof structure is very similar to the proof of 
We define $\mathcal{L}_{d - 1}$ to be the sigma-algebra generated by the first $d - 1$ levels.
\begin{lemma}\label{lem:contribution-last-levels}
    Let $0 \le t \le C$, $j \in [m]$, and $1 \le d \le k$.
    Then
    \begin{align}
        \abs{
            \prbcond{\sum_{i = d}^k Y^{(j)}_{i} \ge t}{\mathcal{L}_{d - 1}}
            - \prb{\sum_{i = d}^k \mathcal{Y}^{(j)}_i \ge t}
        } \le k \indicator{A_{d - 1}^c} + 2k m^{-1/2 + o(1)} 
    \end{align}
\end{lemma}
\begin{proof}
    We will prove that,
    \begin{align}\label{eq:contribution-last-levels-induction}
        \abs{
            \prbcond{\sum_{i = k - r}^k Y^{(j)}_{i} \ge t}{\mathcal{L}_{k - r - 1}}
            - \prb{\sum_{i = k - r}^k \mathcal{Y}^{(j)}_i \ge t}
        } \le (1 + r)\indicator{A_{k - r - 1}^c} + (1 + 2r) m^{-1/2 + o(1)}
    \end{align}
    for $0 \le r \le k - d$ and all $0 \le t \le C$.
    We will prove the result by induction on $r$.

    We first consider $r = 0$.
    We then have that,
    \begin{align*}
        \prbcond{Y^{(j)}_{k} \ge t}{\mathcal{L}_{k - 1}}
            &= \indicator{A_{k - 1}} \prbcond{Y^{(j)}_{k} \ge t}{\mathcal{L}_{r - 1}}
                + \indicator{A_{k - 1}^c} \prbcond{Y^{(j)}_{k} \ge t}{\mathcal{L}_{r - 1}}
    \end{align*}
    We now use \Cref{lem:dependent-to-independent} to get that,
    $\indicator{A_{k - 1}}\prbcond{Y^{(j)}_{k} \ge t}{\mathcal{L}_{k - 1}} = \indicator{A_{k - 1}}\left( \prb{\mathcal{Y}^{(j)}_k \ge t} + \delta \right)$ where $\abs{\delta} \le m^{-1/2 + o(1)}$.
    We then get that,
    \begin{align*}
        \prbcond{Y^{(j)}_{k} \ge t}{\mathcal{L}_{k - 1}}
            &= \indicator{A_{k - 1}}\left( \prb{\mathcal{Y}^{(j)}_k \ge t} + \delta \right)
                + \indicator{A_{k - 1}^c} \prbcond{Y^{(j)}_{k} \ge t}{\mathcal{L}_{r - 1}}
            \\&= \prb{\mathcal{Y}^{(j)}_k \ge t} + \indicator{A_{k - 1}^c} \left(\prbcond{Y^{(j)}_{k} \ge t}{\mathcal{L}_{r - 1}} - \prb{\mathcal{Y}^{(j)}_k \ge t} \right) + \indicator{A_{k - 1}}\delta
    \end{align*}
    We have that
    \begin{align*}
        \abs{\indicator{A_{k - 1}^c} \left(\prbcond{Y^{(j)}_{k} \ge t}{\mathcal{L}_{r - 1}} - \prb{\mathcal{Y}^{(j)}_k \ge t} \right) + \indicator{A_{k - 1}}\delta}
            \le \indicator{A_{k - 1}^c} + \abs{\delta}
            \le \indicator{A_{k - 1}^c} + m^{-1/2 + o(1)}
    \end{align*}
    This proves \cref{eq:contribution-last-levels-induction} for $r = 0$ which will be our induction start.

    Now we consider $r \ge 1$ assume that \cref{eq:contribution-last-levels-induction} is true for values less than $r$.
    We note that,
    \begin{align*}
        \prbcond{\sum_{i = k - r}^k Y^{(j)}_{i} \ge t}{\mathcal{L}_{k - r - 1}}
            = \sum_{s = 0}^{t - 1} \prbcond{\sum_{i = k - r + 1}^k Y^{(j)}_{i} \ge t - s \wedge Y^{(j)}_{k - r} = s}{\mathcal{L}_{k - r - 1}}
        + \prbcond{Y^{(j)}_{k - r} \ge t}{\mathcal{L}_{k - r - 1}}
    \end{align*}
    We now fix $0 \le s \le t - 1$ and use the tower property of conditional expectation to get that,
    \begin{align*}
        \prbcond{\sum_{i = k - r + 1}^k Y^{(j)}_{i} \ge t - s \wedge Y^{(j)}_{k - r} = s}{\mathcal{L}_{k - r - 1}}
            = \epcond{\indicator{Y^{(j)}_{k - r} = s}\prbcond{\sum_{i = k - r + 1}^k Y^{(j)}_{i} \ge t - s}{\mathcal{L}_{k - r}}}{\mathcal{L}_{k - r - 1}}
    \end{align*}
    By then induction hypothesis we have that $\prbcond{\sum_{i = k - r + 1}^k Y^{(j)}_{i} \ge t - s}{\mathcal{L}_{k - r}} = \prb{\sum_{i = k - r}^k \mathcal{Y}^{(j)}_i \ge t - s} + \delta_s$ where $\abs{\delta_s} \le r \indicator{A_{k - r}^c} + (2r - 1) m^{-1/2 + o(1)}$.
    This implies that,
    \begin{align*}
            &\epcond{\indicator{Y^{(j)}_{k - r} = s}\prbcond{\sum_{i = k - r + 1}^k Y^{(j)}_{i} \ge t - s}{\mathcal{L}_{k - r}}}{\mathcal{L}_{k - r - 1}}
                \\&\qquad\qquad\qquad= \epcond{\indicator{Y^{(j)}_{k - r} = s} \left(\prb{\sum_{i = k - r}^k \mathcal{Y}^{(j)}_i \ge t - s} + \delta_s \right)}{\mathcal{L}_{k - r - 1}}
                \\&\qquad\qquad\qquad= \prbcond{\sum_{i = k - r + 1}^k \mathcal{Y}^{(j)}_{i} \ge t - s \wedge Y^{(j)}_{k - r} = s}{\mathcal{L}_{k - r - 1}} + \epcond{\indicator{Y^{(j)}_{k - r} = s} \delta_s}{\mathcal{L}_{k - r - 1}}
    \end{align*}
    Hence, we get that,
    \begin{align*}
        \prbcond{\sum_{i = k - r}^k Y^{(j)}_{i} \ge t}{\mathcal{L}_{k - r - 1}}
            &= \sum_{s = 0}^{t - 1} \prbcond{\sum_{i = k - r + 1}^k Y^{(j)}_{i} \ge t - s \wedge Y^{(j)}_{k - r} = s}{\mathcal{L}_{k - r - 1}}
                + \prbcond{Y^{(j)}_{k - r} \ge t}{\mathcal{L}_{k - r - 1}}
            \\&= \sum_{s = 0}^{t - 1} \left( \prbcond{\sum_{i = k - r + 1}^k \mathcal{Y}^{(j)}_{i} \ge t - s \wedge Y^{(j)}_{k - r} = s}{\mathcal{L}_{k - r - 1}} + \epcond{\indicator{Y^{(j)}_{k - r} = s} \delta_s}{\mathcal{L}_{k - r - 1}} \right)
                \\&\qquad\qquad+ \prbcond{Y^{(j)}_{k - r} \ge t}{\mathcal{L}_{k - r - 1}}
            \\&= \prbcond{\sum_{i = k - r + 1}^k \mathcal{Y}^{(j)}_{i} + Y^{(j)}_{k - r} \ge t}{\mathcal{L}_{k - r - 1}}
                + \sum_{s = 0}^{t - 1} \epcond{\indicator{Y^{(j)}_{k - r} = s} \delta_s}{\mathcal{L}_{k - r - 1}}
    \end{align*}
    Now we want to exchange $Y^{(j)}_{k - r}$ with $\mathcal{Y}^{(j)}_{k - r}$ and the method is similar to before.
    We write,
    \begin{align*}
        \prbcond{\sum_{i = k - r + 1}^k \mathcal{Y}^{(j)}_{i} + Y^{(j)}_{k - r} \ge t}{\mathcal{L}_{k - r - 1}}
            = \sum_{s = 0}^{t - 1} \prb{\sum_{i = k - r + 1}^k \mathcal{Y}^{(j)}_{i} = s} \prbcond{Y^{(j)}_{k - r} \ge t - s}{\mathcal{L}_{k - r - 1}}
        + \prb{\sum_{i = k - r + 1}^k \mathcal{Y}^{(j)}_{i} \ge t}
    \end{align*}
    Now by analogous arguments as in the induction start we get that $\prbcond{Y^{(j)}_{k - r} \ge t - s}{\mathcal{L}_{k - r - 1}} = \prb{\mathcal{Y}^{(j)}_{k - r} \ge t - s} + \delta'_s$ where $\abs{\delta'_s} \le \indicator{A_{k - r - 1}^c} + m^{-1/2 + o(1)}$.
    We thus get that,
    \begin{align*}
        &\sum_{s = 0}^{t - 1} \prb{\sum_{i = k - r + 1}^k \mathcal{Y}^{(j)}_{i} = s} \prbcond{Y^{(j)}_{k - r} \ge t - s}{\mathcal{L}_{k - r - 1}}
                + \prb{\sum_{i = k - r + 1}^k \mathcal{Y}^{(j)}_{i} \ge t}
            \\&\qquad\qquad\qquad= \sum_{s = 0}^{t - 1} \prb{\sum_{i = k - r + 1}^k \mathcal{Y}^{(j)}_{i} = s} \left(\prb{\mathcal{Y}^{(j)}_{k - r} \ge t - s} + \delta'_s \right)
                + \prb{\sum_{i = k - r + 1}^k \mathcal{Y}^{(j)}_{i} \ge t}
                \\&\qquad\qquad\qquad= \prb{\sum_{i = k - r}^k \mathcal{Y}^{(j)}_{i} \ge t}
                    +  \sum_{s = 0}^{t - 1} \prb{\sum_{i = k - r + 1}^k \mathcal{Y}^{(j)}_{i} = s} \delta'_s
    \end{align*}
    Now combining it all we get that,
    \begin{align*}
        \prbcond{\sum_{i = k - r}^k Y^{(j)}_{i} \ge t}{\mathcal{L}_{k - r - 1}}
            = \prb{\sum_{i = k - r}^k \mathcal{Y}^{(j)}_{i} \ge t}
            +  \sum_{s = 0}^{t - 1} \left(\epcond{\indicator{Y^{(j)}_{k - r} = s} \delta_s}{\mathcal{L}_{k - r - 1}} + \prb{\sum_{i = k - r + 1}^k \mathcal{Y}^{(j)}_{i} = s} \delta'_s \right)
    \end{align*}
    Now to finish the proof we just need to bound $\abs{\sum_{s = 0}^{t - 1} \left(\epcond{\indicator{Y^{(j)}_{k - r} = s} \delta_s}{\mathcal{L}_{k - r - 1}} + \prb{\sum_{i = k - r + 1}^k \mathcal{Y}^{(j)}_{i} = s} \delta'_s \right)}$.
    \begin{align*}
        &\abs{\sum_{s = 0}^{t - 1} \left(\epcond{\indicator{Y^{(j)}_{k - r} = s} \delta_s}{\mathcal{L}_{k - r - 1}} + \prb{\sum_{i = k - r + 1}^k \mathcal{Y}^{(j)}_{i} = s} \delta'_s \right)}
            \\&\qquad\qquad\qquad\le \max_{s = 0}^{t - 1} \epcond{\abs{\delta_s}}{\mathcal{L}_{k - r - 1}}
                + \max_{s = 0}^{t - 1} \abs{\delta'_s}
            \le \indicator{A_{k - r - 1}^c} + 2r m^{-1/2 + o(1)} + r \prbcond{A_{k - r}^c}{\mathcal{L}_{k - r - 1}}
    \end{align*}
    We now just need to bound $\prbcond{A_{k - r}^c}{\mathcal{L}_{k - r - 1}}$.
    From the proof of \Cref{eq:concentration-on-full} we have that $\indicator{A_{k - r - 1}}\prbcond{A_{k - r}^c}{\mathcal{L}_{k - r - 1}} \le m^{-\gamma}$.
    So we get that 
    \[
        r \prbcond{A_{k - r}^c}{\mathcal{L}_{k - r - 1}}
            \le r \indicator{A_{k - r - 1}^c} + r m^{-\gamma}
            \le r \indicator{A_{k - r - 1}^c} + m^{-1/2 + o(1)}
    \]
    This implies that,
    \begin{align*}
        &\abs{\sum_{s = 0}^{t - 1} \left(\epcond{\indicator{Y^{(j)}_{k - r} = s} \delta_s}{\mathcal{L}_{k - r - 1}} + \prb{\sum_{i = k - r + 1}^k \mathcal{Y}^{(j)}_{i} = s} \delta'_s \right)}
            \le (r + 1) \indicator{A_{k - r - 1}^c} + (1 + 2r) m^{-1/2 + o(1)}
    \end{align*}
    This finishes the induction step and thus the proof.
\end{proof}

We now turn to the proof \Cref{lem:dependent-concentration}.
\begin{proof}[Proof of \Cref{lem:dependent-concentration}]
    Let $j_1, \ldots, j_r$ be different indices and $a_1, \ldots, a_r$ be non-negative integers
    such that $\sum_{i = 1}^{r} a_i = r$. We then want to estimate $\epcond{(X_{j_r} - p)^{a_r}}{(X_{j_i})_{i < r}}$
    \begin{align*}
        \abs{\epcond{(X_{j_r} - p)^{a_r}}{(X_{j_i})_{i < r}}}
            &= \abs{\epcond{X_{j_r}}{(X_{j_i})_{i < r}}(1 - p)^{a_r} + (1 - \epcond{X_{j_r}}{(X_{j_i})_{i < r}})(-p)^{a_r} }
            \\&\le \delta \abs{(1 - p)^{a_r} - (-p)^{a_r}} + \abs{p(1 - p)^{a_r} + (1 - p)(-p)^{a_r} }
    \end{align*}
    Now let $X'_i$ be independent Bernoulli variables with parameter $p'$ where $p' = p + \delta$ if $p < \frac{1}{2}$
    and $p' = p - \delta$ when $p \ge \frac{1}{2}$. It is now easy to check that
    \[
        \abs{\ep{(X'_i - p)^{a_r}}}
            = \delta \abs{(1 - p)^{a_r} - (-p)^{a_r}} + \abs{p(1 - p)^{a_r} + (1 - p)(-p)^{a_r} }
    \]
    Using this we see that
    \[
        \abs{\ep{\prod_{i = 1}^r (X_i - p)^{a_r}}}
            \le \abs{\ep{\prod_{i = 1}^r (X'_i - p)^{a_r}}}
    \]
    From this we conclude that $\ep{\left( \sum_{i} (X_i - p) \right)^r } \le \ep{\left( \sum_{i} (X'_i - p) \right)^r }$.
    Now by the triangle inequality and Hoeffding's inequality we get that
    \begin{align*}
        \ep{\left( \sum_{i} (X_i - p) \right)^r }^{1/r}
            &\le \ep{\left( \sum_{i} (X'_i - p) \right)^r }^{1/r}
            \\&\le \delta n + \ep{\left( \sum_{i} (X'_i - p') \right)^r }^{1/r}
            \\&\le \delta n + O(\sqrt{r n})
    \end{align*}

    Now using Markov's inequality give us the tail bound.
\end{proof}
\subsection{The probability that a bin is not full}\label{sec:f-bound}

In this section we will bound the probability $\prb{\mathcal{X}^{(j)}_k = 0}$ for any bin $j$.
Since the bound is the same for all bins we will suppress $j$ from the notation.
We note that $\prb{\mathcal{X}_k = 0} = \prb{\sum_{i = 1}^k \mathcal{Y}_i \ge C}$. Now an important
observation is that if we define $1 - f_d =  \prb{\sum_{i = 1}^d \mathcal{Y}_i \ge C}$, then $\mathcal{Y}_d$ is
geometrically distributed with parameter $\alpha_d = \frac{n/k}{n/k + f_d m} = \frac{1}{1 + \tfrac{f_d k}{\mu}}$.
The reason is that when generating $\mathcal{Y}_d$, we sample with replacement so when sampling a bin, the  probability that it will be filled is $1 - f_d$ independently of the history.
Thus, at any point, the probability of getting another ball is
\begin{align*}
    \frac{n/k}{n/k + m} \sum_{i = 0}^{\infty} \left( \frac{m}{n/k + m} (1 - f_d) \right)^i
        = \frac{\frac{n/k}{n/k + m}}{1 - \frac{m}{n/k + m}(1 - f_d)}
        = \frac{n/k}{n/k + f_d m}
        = \frac{1}{1 + \tfrac{f_d k}{\mu}}
    \; .
\end{align*}
Which is exactly what we get from a geometrically distributed variable.

From simple facts about geometrically distributed variables we get that $\mu_d = \ep{\mathcal{Y}_d} = \frac{\mu}{k f_d}$,
and $\sigma_d^2 = \var{\mathcal{Y}_d} = \frac{\mu}{k f_d} \left(1 + \frac{\mu}{k f_d} \right) \ge \mu_d$.
We note that 
\begin{align}
    \sum_{i = 1}^k \sigma_i^2 = \sum_{i = 1}^k \frac{\mu}{k f_i}\left(1 + \frac{\mu}{k f_i} \right)
        \ge \sum_{i = 1}^k \frac{\mu}{k f_i}    
\end{align}
We define $S_d = \sum_{i = 1}^d \mathcal{Y}_i$ for $1 \le d \le k$.
Our goal is to prove that there exists a constant $L$ such that,
\begin{align}\label{eq:f-expression}
    f_k \ge L\begin{cases}
        \eps C &\text{if $C \le\left(1/(\eps \sqrt{C})\right)$} \\
        \eps \sqrt{C \log\left(\tfrac{1}{\eps \sqrt{C}} \right)} &\text{if $\log\left(1/(\eps \sqrt{C})\right) \le C \le \frac{1}{2 \eps^2}$} \\
        1 &\text{if $\frac{1}{2\eps^2} \le C$}
    \end{cases} .
\end{align}
Combining this with~\Cref{eq:concentration-on-full}, the  second part of~\Cref{thm:jakobs} will follow. 
Now to prove~\Cref{eq:f-expression}, it suffices to consider the case $C\leq \gamma/\eps^2$ for a sufficiently small constant $\gamma$.
Indeed, otherwise, we apply a reduction similar to the one in Case 3 in the proof of~\Cref{thm:main0}. We will make this assumption in what follows. 
We also note that we can assume that $C$ is larger than $\tfrac{1}{4 L}$ because if $C \le \tfrac{1}{4 L}$ then we get that,
\begin{align*}
    \eps/4 \ge L\begin{cases}
        \eps C &\text{if $C \le \log\left(1/(\eps \sqrt{C})\right)$} \\
        \eps \sqrt{C \log\left(\tfrac{1}{\eps \sqrt{C}} \right)} &\text{if $\log\left(1/(\eps \sqrt{C})\right) \le C \le \frac{1}{2 \eps^2}$} \\
        1 &\text{if $\frac{1}{2\eps^2} \le C$}
    \end{cases}
    \; .
\end{align*}
We will argue that $f_k$ is always larger than $\eps/4$.
We know that $\sum_{j \in [m]} \indicator{X_k^{(j)} = 0} \le \frac{n}{C} = \frac{m}{1 + \eps}$ and $1 - f_k \le \frac{\sum_{j \in [m]} \indicator{X_k^{(j)} = 0}}{m} + m^{-1/2 + o(1)}$ with probability at least $1 - m^{-\gamma}$ by \Cref{thm:jakobs}.
Fixing such event give us that, 
\[
    f_k \ge 1 - \frac{1}{1 + \eps} - m^{-1/2 + o(1)} \ge \eps/2 - \eps/4 = \eps/4 \; .
\]
Here we have used that $\eps \le 1$ and that $1/\eps = m^{o(1)}$.
So if $C \le \tfrac{1}{4 L}$ then \cref{eq:f-expression} holds and we in from now assume that $C > \tfrac{1}{4 L}$.

We will prove the result by showing the stronger result that for all $1 \le d \le k$,
\begin{align}\label{eq:f-expression-induction}
    f_d \ge L\begin{cases}
        \eps C &\text{if $C \le \log\left(1/(\eps \sqrt{C})\right)$} \\
        \eps \sqrt{C \log\left(\tfrac{1}{\eps \sqrt{C}} \right)} &\text{if $\log\left(1/(\eps \sqrt{C})\right) \le C \le \frac{1}{2 \eps^2}$} \\
        1 &\text{if $\frac{1}{2\eps^2} \le C$}
    \end{cases} .
\end{align}
We will prove \cref{eq:f-expression-induction} by induction on $d$.

First we note that $f_1 \ge f_2 \ge \ldots \ge f_k \eps/4$.
We then get that $\ep{S_d} = \sum_{i = 1}^{d} \mu_i \le \frac{4d \mu}{\eps k} \le \frac{d \eps C}{2}$, where we have used that $k \ge 8/\eps^2$.
So for $d \le 1/\eps$ we get that $\sum_{i = 1}^{d} \mu_i \le C/2$ and Markov's inequality give us that,
\[
    1 - f_d = \prb{S_d \ge C} \le \frac{1}{2}
        \; .
\]
This shows that \cref{eq:f-expression-induction} holds for $d \le 1/\eps$ and since $\eps \le 1$ then it holds for $d = 1$ which will be our induction start.

Now the previous argument shows that when $\ep{S_d} \le C/2$ then \cref{eq:f-expression-induction} holds, so we can assume that $\sum_{i = 1}^{d} \mu_i > C/2$.
We note that $f_d \ge f_{d - 1}/2$ since,
\begin{align*}
    f_d = \prb{S_d < C}
        \ge \prb{S_{d - 1} < C \wedge \mathcal{Y}_d = 0}
        = f_{d - 1}(1 - \alpha_d)
        = f_{d - 1}\frac{1}{1 + \tfrac{\mu}{f_d k}}
    \; .
\end{align*}
This implies that $f_d \ge f_{d - 1} - \frac{\mu}{k} \ge f_{d - 1}/2$ where we use that $k \ge c/\eps^2$ for some sufficiently large constant $c$ and that $f_{d - 1} \ge L \eps^2 C$.
We now note that if $f_i \ge \tfrac{L}{2}\min\set{\eps \sqrt{C}, 1}$ then $\mu_i \le 1$, since,
\begin{align*}
    \mu_d
        = \frac{\mu}{k f_i}
        \le \frac{2 \eps^2 C}{c L \min\set{\eps \sqrt{C}, 1}}
        \le 1
    \; .
\end{align*}
Here we have used that $k \ge c/\eps^2$, that $c$ is sufficiently large, and that $1/\eps^2 \le C$.
We also note that 
\begin{align*}
    \ep{\abs{\mathcal{Y}_i - \mu_i}^3}
        &= \ep{\indicator{\mathcal{Y}_i \le \mu_i} (\mu_i - \mathcal{Y}_i)^3}
            + \ep{\indicator{\mathcal{Y}_i > \mu_i} (\mathcal{Y}_i - \mu_i)^3}
        \\&\le \mu_i^3
            + \alpha_i^{\floor{\mu_i}} \ep{\mathcal{Y}_i^3}
        \le \mu_i^3 + \ep{\mathcal{Y}_i^3}
    \; .
\end{align*}
In the first inequality we have used that the geometric distribution is memoryless.
Now simple calculations give that $\ep{\mathcal{Y}_i^3} \le 6\mu_i(1 + \mu_i)^2 \le 24 \mu_i$,
so we get that $\ep{\abs{\mathcal{Y}_i - \mu_i}^3} \le \mu_i^3 + 24 \mu_i \le 25 \mu_i$.

Depending on $\ep{S_d}$ we will prove different bounds on $f_d$.
Let $M > 0$ be a large constant.
We will prove that if $\ep{S_d} < C + M \sqrt{C}$ then $f_d \ge L$,
if $C + M\sqrt{C} \le \ep{S_d} < C + \tfrac{1}{4}C$ then $f_d \ge L \eps \sqrt{C \log\left(\tfrac{1}{\eps \sqrt{C}} \right)}$,
and if $C + \tfrac{1}{4}C \le \ep{S_d}$ then $f_d \ge L \eps C$.
This will prove the result since $\eps \sqrt{C \log\left(\tfrac{1}{\eps \sqrt{C}} \right)} \ge \eps C$ if and only if $C \ge \left(1/(\eps \sqrt{C})\right)$,
and since $C\leq \gamma/\eps^2$ for a small constant $\gamma$ then $1 \ge \min\set{\eps \sqrt{C \log\left(\tfrac{1}{\eps \sqrt{C}} \right)}, \eps C}$.

If $\ep{S_d} < C + M \sqrt{C}$ then we will show that $f_d \ge L$.
This will follow by a usage of the Berry-Esseen theorem.
\begin{theorem}[Berry Esseen theorem]
    Let $X_1, \ldots, X_d$ be independent random variables with $\ep{X_i} = 0$, $\ep{X_i^2} = \sigma_i^2 > 0$, and $\ep{\abs{X_i}^3} = \rho_i < \infty$.
    Let $F_d$ be the cumulative distribution function of $\sum_{i = 1}^d X_i$, let $\Phi$ be the cumulative distribution function of the standard normal distribution,
    and let $\sigma^2 = \sum_{i = 1}^d \sigma^2$.
    Then,
    \begin{align*}
        \sup_{x \in \R} \abs{ F_d(x) - \Phi(x / \sigma) } \le K_1 \frac{\sum_{i = 1}^d \rho_i}{\sigma^3}
        \; .
    \end{align*}
    where $K_1$ is a universal constant.
\end{theorem}
Since $\ep{S_d} < C + M \sqrt{C}$ then $C \ge \ep{S_d} - M \sqrt{\ep{S_d}} \le \ep{S_d} - M \sqrt{\sum_{i = 1}^d \sigma_i^2}$ and we get that $f_d = \prb{S_d < C} \ge \prb{S_d < \ep{S_d} - M \sqrt{\sum_{i = 1}^d \sigma_i^2}}$.
Now the Berry-Esseen theorem give us that,
\begin{align*}
    f_d \ge \prb{S_d < \ep{S_d} - \sqrt{\sum_{i = 1}^d \sigma_i^2}}
        \ge \Phi(-M) - K_1 \frac{\sum_{i = 1}^d \ep{\abs{\mathcal{Y}_i - \mu_i}^3}}{\left(\sum_{i = 1}^d \sigma_i^2 \right)^{3/2}}
\end{align*}
We know that $\ep{\abs{\mathcal{Y}_i - \mu_i}^3} \le 25 \mu_i$ and that $\sigma_i^2 \ge \mu_i$ for all $1 \le i \le d$, so we get that,
\begin{align*}
    f_d
        \ge \Phi(-M) - 25 K_1 \frac{\ep{S_d}}{\ep{S_d}^{3/2}}
        \ge \Phi(-M) - 25 K_1 \sqrt{8 L}
        \ge L
    \; .
\end{align*}
Here we have used that $\ep{S_d} \ge C/2 \ge \frac{1}{8 L}$ and that $L$ is sufficiently small.

Now we consider the case where $\ep{S_d} \ge C + M \sqrt{C}$.
We define $\beta_d = \ep{S_d} - C$ and note that $\beta_d \ge M \sqrt{C}$.
We will need the following claim.
\begin{claim}\label{claim:log-concave}
    For all $1 \le d \le k$ and all integers $t \ge 1$ we have that,
    \[
        \frac{\prb{S_d = t + 1}}{\prb{S_d = t}} \le \frac{\prb{S_d = t}}{\prb{S_d = t - 1}} \; .
    \]
\end{claim}
\begin{proof}
    We define the sets $A_t = \setbuilder{(a_1, \ldots, a_{d}) \in \N_0^d}{\sum_{i = 1}^d a_i = t}$ and get that
    \[
        \prb{\sum_{i = 1}^d \mathcal{Y}_i = t}
            = \sum_{(a_1, \ldots, a_{d}) \in A_t} \prod_{i = 1}^d\ \alpha_i^{a_i} (1 - \alpha_i) \; .
    \]
    We note that the result it is equivalent to showing that $\prb{X = t + 1}\prb{X = t - 1} \le \prb{X = t}^2$ which in turn is equivalent to 
    $$
        \sum_{(a, b) \in A_{t + 1} \times A_{t - 1}} \prod_{i = 1}^d \alpha_i^{a_i + b_i} (1 - \alpha_i)^2 \le \sum_{(a, b) \in A_t \times A_t} \prod_{i = 1}^d \alpha_i^{a_i + b_i} (1 - \alpha_i)^2.
    $$
    To see that this latter inequality holds, let $s \in A_{2t}$ and define the map $g_s: \N_0^d \to \N_0$ by $g_s(i)=|\{(a,b)\in A_{i}\times A_{2t-i}\mid a+b=s \}|$.
    We note that $g_s(i)>0$ exactly when $i\in \{0,1,\dots,2t\}$.
    The desired inequality is then equivalent to
    $$
        \sum_{s\in A_{2t}} g_s(t+1) \prod_{i = 1}^d \alpha_i^{s_i} \leq \sum_{s \in A_{2t}}g_s(t)\prod_{i = 1}^d\alpha_i^{s_i}.
    $$
    We will show that $g_s$ is log-concave for each $s\in A_{2t}$.
    As $g_s$ is clearly symmetric around $i=t$, it will in particular follow that $g_s(t+1)\leq g_s(t)$ which then leads to the desired inequality.
    To show that $g_s$ is log-concave, we note that it is a convolution of log-concave functions.
    Indeed, fix $s$ and define for $1 \le j \le d$, the map $h_j:\N_0 \to \N_0$ by $h_j(i)=1$ if $0\leq i\leq s_j$ and $h_j(i)=0$ otherwise.
    Then each $h_j$ is log-concave, and moreover, $g_s$ is the convolution $g_s=h_1*\cdots *h_k$, i.e.,
    $$
        g_s(i)=\sum_{\substack{a\in \Z^k \\  a_1+\cdots+a_d=i}}\prod_{j = 1}^d h_j(a_j).
    $$
    It is a standard fact that the convolution of log-concave functions is again log-concave, and the desired inequality follows.
\end{proof}

Now let $\ell_d \in \N$ be the minimal integer satisfying that $\prb{S_d = C - 1}/\prb{S_d = C - 1 - \ell_d} \ge 2$.
Now combining \Cref{claim:log-concave} with the definition of $\ell_d$ we get that,
\begin{align*}
    \prb{S_d < C}
        &= \sum_{t = 1}^{C} \prb{S_d = C - t}
        \ge \sum_{t = 1}^{\ell_d} \prb{S_d = C - t}
        \ge \frac{\ell_d}{2} \prb{S_d = C - 1} \; , \\
    \prb{S_d < C}
        &= \sum_{t = 1}^{C} \prb{S_d = C - t}
        \le \sum_{r = 0}^{\ceil{C/\ell_d}} \ell_d \prb{S_d = C - 1 - r\ell_d}
        \\&\le \ell_d \prb{S_d = C - 1} \sum_{r = 0}^{\infty} 2^{-r}
        = 2 \ell_d \prb{S_d = C - 1} \; , \\
    \ep{(C - S_d)[S_d < C]}
        &= \sum_{t = 1}^{C} t \prb{S_d = C - t}
        \le \sum_{r = 0}^{\ceil{C/\ell_d}} \left(r \ell_d^2 + \frac{\ell_d(\ell_d + 1)}{2} \right) \prb{S_d = C - 1 - r\ell_d}
        \\&\le \sum_{r = 0}^{\ceil{C/\ell_d}} \left(r \ell_d^2 + \frac{\ell_d(\ell_d + 1)}{2} \right) 2^{-r} \prb{S_d = C - 1}
        \le 4 \ell_d^2 \prb{S_d = C - 1} \; ,\\
    \ep{(C - S_d)[S_d < C]}
        &= \sum_{t = 1}^{C} t \prb{S_d = C - t}
        \ge \sum_{t = 1}^{\ell_d} t \prb{S_d = C - 1 - t}
        \ge \frac{\ell_d^2}{4} \prb{S_d = C - 1} \; .
\end{align*}
From this we get that $\frac{\ep{(C - S_d)[S_d < C]}}{8 \ell_d} \le f_d \le \frac{8\ep{(C - S_d)[S_d < C]}}{\ell_d}$.
Now it is clear that $\ep{(C - S_d)[S_d < C]} \ge \ep{(C - S_k)[S_k < C]}$ and we will argue that $\ep{(C - S_k)[S_k < C]} \ge \tfrac{\eps \mu}{2} \ge \tfrac{\eps C}{4}$.
This will imply that $f_d \ge \tfrac{\eps C}{32 \ell_d}$.
Using~\Cref{eq:concentration-on-full} we get that $\prb{S_d \ge C - t} \ge \frac{\sum_{j \in [m]} \indicator{X^{(j)}_k \le t}}{m} - m^{-1/2 + o(1)}$ for all $1 \le t \le C$ with probability $1 - m^{-\gamma}$,
and we know that $\sum_{t = 1}^{C} \frac{\sum_{j \in [m]} \indicator{X^{(j)}_k \le t}}{m} = \eps C$, so fixing such event give us that,
\begin{align*}
    \ep{(C - S_k)[S_k < C]}
        &= \sum_{t = 1}^{C} \prb{S_d \ge C - t}
        \\&\ge \sum_{t = 1}^{C} \left( \frac{\sum_{j \in [m]} \indicator{X^{(j)}_k \le t}}{m} - m^{-1/2 + o(1)} \right)
        \\&= \eps \mu - C m^{-1/2 + o(1)}
        \\&\ge \tfrac{\eps C}{2} - C m^{-1/2 + o(1)}
        \\&\ge \tfrac{\eps C}{4}
    \; .
\end{align*}
Here we have used that $\eps \le 1$ and that $1/\eps = m^{o(1)}$.

Now we just need to upper bound $\ell_d$.
By \Cref{claim:log-concave} we get that $\ell_d \le \ceil{\frac{\log(2)}{\log\left(\tfrac{\prb{S_d = C}}{\prb{S_d = C - 1}}\right)}}$ so we want to lower bound $\frac{\prb{S_d = C}}{\prb{S_d = C - 1}}$.
To do this we will define exponentially tilted variables $(V_i)_{1 \le i \le d}$.
Let $\lambda\in \R$ satisfying $\E[e^{\lambda S_d}]<\infty$ be a parameter which will be determined later.
We define $V_i$ by $\prb{V_i = t} = \frac{\prb{\mathcal{Y}_i = t} e^{\lambda t}}{\ep{e^{\lambda \mathcal{Y}_i}}}$ for $1 \le i \le d$.
Clearly, this is well-defined since $\sum_{t = 0}^\infty \prb{\mathcal{Y}_i = t} e^{\lambda t} = \ep{e^{\lambda \mathcal{Y}_i}}$.
As pointed out in ~\cite{aamand2021sums}, each $V_i$ is also geometric random variables (with parameter $\alpha_i e^{\lambda}$) and, 
\begin{align}\label{eq:exponential-twist}
    \prb{S_d = C - t} = \frac{\ep{e^{\lambda \sum_{i = 1}^d \mathcal{Y}_i}}}{e^{\lambda C}} e^{\lambda t} \prb{\sum_{i = 1}^d V_i = C - t} .
\end{align}
for all integers $t$.
Moreover, there is a unique $\lambda$ maximizing $\lambda C - \log \ep{e^{\lambda \sum_{i = 1}^{d} \mathcal{Y}_i}}$, and with this choice of $\lambda$, it holds that $ \sum_{i = 1}^d \ep{V_i} = C$.
It is easy to see that $\lambda < 0$ since $\beta > 0$.
We start by noticing that $\sum_{i = 1}^d \var{V_i} \ge \sum_{i = 1}^d \ep{V_i} = C$,
and that $\ep{\abs{V_i - \ep{V_i}}^3} \le 25 \ep{V_i}$ by the same reasoning that gave us that $\ep{\abs{\mathcal{Y}_i - \mu_i}^3} \le 25 \ep{\mathcal{Y}_i}$ since $V_i$ is geometrically distributed with parameter $\alpha_i e^{\lambda} < \alpha_i$.

We will also need the following lemma by Aamand et al.~\cite{aamand2021sums}.
We state a simplified version of their lemma which covers our use case.
\begin{lemma}\label{lem:fourier}
    Let $X_1, \ldots, X_d$ be independent geometric distributed random variables with $\var{X_i} = \sigma_i^2 > 0$ and $\ep{\abs{X_i - \ep{X_i}}^3} = \rho_i < \infty$,
    and let $\sigma^2 = \sum_{i = 1}^d \sigma_i^2$.
    Then for every $t$ where $\mu + t\sigma$ is an integer,
    \begin{align*}
        \abs{\prb{X = \mu + t\sigma} - \frac{1}{\sqrt{2\pi} \sigma} e^{-t^2/2} } \le K_2 \left(\frac{\sum_{i = 1}^d \rho_i}{\sigma^3} \right)^2
        \; .
    \end{align*}
    where $K_2$ is a universal constant.
\end{lemma}

We will also need the following claim.
The proof is bit technical so we defer the proof till the end of the section. 
\begin{claim}\label{claim:rate-of-decrease}
    If $\beta \ge C$ then,
    \begin{align*}
        \frac{\prb{S_d = C}}{\prb{S_d = C - 1}} \ge e^{1/8}
            \; ,
    \end{align*}
    and if $\beta < C$ then,
    \begin{align*}
        \frac{\prb{S_d = C}}{\prb{S_d = C - 1}} \ge e^{\tfrac{1}{8} \beta/C}
        \; ,
    \end{align*}
    and
    \begin{align*}
        \frac{\prb{S_d = C - 1}}{\prb{S_d = C - 1 - \tfrac{1}{2}\ceil{\tfrac{C}{\beta}} }} < 2
        \; .
    \end{align*}
\end{claim}

If $\beta \ge \tfrac{1}{4} C$ then using \Cref{claim:rate-of-decrease} we get that $\ell_d \le \ceil{32\log(2)}$ which implies that $f_d \ge \frac{\eps C}{32\ceil{32\log(2)}} \ge L \eps C$.
So now we just need to focus on the case where $\beta < \tfrac{1}{4} C$.
We use \Cref{claim:rate-of-decrease} to get that $\ell_d \le \ceil{8 \log(2) \tfrac{C}{\beta}} \le 7 \tfrac{C}{\beta}$ which implies that $f_d \ge \frac{\eps \beta}{224}$.

We now just need to lower bound $\beta$.
From \Cref{claim:rate-of-decrease} we know that $\ell_d > \tfrac{1}{4}\ceil{\tfrac{C}{\beta}} \ge \tfrac{C}{4 \beta} > 1$,
so $\prb{S_d = C} \le \prb{S_d = C - 1}/2$ and we get that $\ep{(C - S_d)[S_d < C]} \ge \frac{\ell_d^2}{8} \prb{S_d = C}$.
We will argue that $\ep{(C - S_d)[S_d < C]} \le 2 \eps C$.
Using~\Cref{eq:concentration-on-full} we get that $\prb{S_d \ge C - t} \le \frac{\sum_{j \in [m]} \indicator{X^{(j)}_k \le t}}{m} + m^{-1/2 + o(1)}$ for all $1 \le t \le C$ with probability $1 - m^{-\gamma}$,
and we know that $\sum_{t = 1}^{C} \frac{\sum_{j \in [m]} \indicator{X^{(j)}_k \le t}}{m} = \eps C$, so fixing such event give us that,
\begin{align*}
    \ep{(C - S_k)[S_k < C]}
        &= \sum_{t = 1}^{C} \prb{S_d \ge C - t}
        \\&\le \sum_{t = 1}^{C} \left( \frac{\sum_{j \in [m]} \indicator{X^{(j)}_k \le t}}{m} + m^{-1/2 + o(1)} \right)
        \\&= \eps \mu + C m^{-1/2 + o(1)}
        \\&\le \eps C - C m^{-1/2 + o(1)}
        \\&\le 2 \eps C
    \; .
\end{align*}
Here we have used that $1/\eps = m^{o(1)}$.
Combing it all we have that,
\begin{align*}
    \prb{S_d = C}
        \le 16 \frac{\eps C}{\ell_d^2}
        \le 256 \frac{\eps \beta^2}{C}
    \; .
\end{align*}

We will now prove that
\begin{align}\label{eq:pointprob}
    \prb{X = C} \ge \frac{\exp\left(-\tfrac{\beta^2}{C} \right)}{4 \sqrt{C}}
        \; .
\end{align}
This will lead to the desired result. Indeed, combining with the bounds above, we then obtain that 
$$
\frac{\exp\left(-\tfrac{\beta^2}{C} \right)}{\sqrt{C}} \le \frac{\eps\beta^2}{1024 C} \; ,
$$
or $\Delta e^{\Delta} \ge \frac{1}{1024 \eps\sqrt{C}}$, where we have put $\Delta=\beta^2/C$.
Then $\Delta \ge \tfrac{1}{1024} \log\left( \frac{1}{\eps\sqrt{C}}\right)$, so that $\beta \ge \tfrac{1}{32} \sqrt{C\log\left( \frac{1}{\eps\sqrt{C}} \right)}$,
and finally 
\[
    f_d = \prb{S_d < C} \ge \tfrac{1}{32} \eps \sqrt{C\log\left( \frac{1}{\eps\sqrt{C}} \right)} \ge L \eps \sqrt{C\log\left( \frac{1}{\eps\sqrt{C}} \right)}
        \; ,
\]
as desired.

We thus turn to prove~\cref{eq:pointprob}.
By \cref{eq:exponential-twist} we have that,
\begin{align}\label{eq:exponential-twist-full}
    \prb{S_d = C} = \frac{\ep{e^{\lambda \sum_{i = 1}^{d} \mathcal{Y}_i}}}{e^{\lambda C}} \prb{\sum_{i \in [k]} V_i = C}
        \; .
\end{align}
We start by focusing on bounding $\lambda C - \log \ep{e^{\lambda \sum_{i = 1}^d \mathcal{Y}_i}}$.
First write $\psi_d(p) = \log \ep{e^{p \sum_{i = 1}^d \mathcal{Y}_i}}$ and define the function $g_d(t) = \sup_{p}(p t - \psi_d(p))$ which is the Fenchel-Legendre transform of $\psi_d(p)$.
By our choice of $\lambda$, $g_d(C) = \lambda C - \log \ep{e^{\lambda \sum_{i = 1}^d \mathcal{Y}_i}}$.
It is easy to check that $g_d(C + \beta) = 0$ and $g_d'(C + \beta) = 0$, and a standard result on the Fenchel-Legendre transformations is that $g_d''(t) = \frac{1}{\psi_d''(p_d(t))}$ where $p_d(t)$ is the unique number such that $g_d(t) = p_d(t) t - \psi_d(p_d(t))$.
Now by Taylor's expansion formula we have that
\begin{align}\label{eq:exp-Taylor}
    g_d(C) \le \left(\sup_{C \le t \le C + \beta} g_d''(t) \right) \frac{\beta^2}{2}
        = \left(\frac{1}{\inf_{C \le t \le C + \beta} \psi_d''(p_d(t))} \right) \frac{\beta^2}{2}
\end{align}
We have that $\psi_d'(p)=\sum_{i = 1}^d \frac{\ep{\mathcal{Y}_i e^{pY_i}}}{\ep{e^{p\mathcal{Y}_i}}}$ and 
\begin{align*}
    \psi_d''(p) = \sum_{i = 1}^d \left( \frac{\ep{\mathcal{Y}_i^2 e^{p \mathcal{Y}_i}}}{\ep{e^{p \mathcal{Y}_i}}} - \left(\frac{\ep{\mathcal{Y}_i e^{p \mathcal{Y}_i}}}{\ep{e^{p \mathcal{Y}_i}}} \right)^2 \right)
        \ge \sum_{i = 1}^d \frac{\ep{\mathcal{Y}_i e^{p \mathcal{Y}_i}}}{\ep{e^{p \mathcal{Y}_i}}}
        = \psi_d'(p).
\end{align*}
Now, $p_d(t) \ge \lambda$ when $C \le t \le C + \beta$.
This implies that $\psi_d''(p(t)) \ge \psi_d'(\lambda) = C$ when $C \le t \le C + \beta$.
Combining this with \cref{eq:exponential-twist-full} and \cref{eq:exp-Taylor} we get that
\begin{align}
    \prb{S_d = C}
        \ge e^{- \tfrac{\beta^2}{2 C}} \prb{\sum_{i = 1}^d V_i = C}
        \ge e^{- \tfrac{\beta^2}{C}} \prb{\sum_{i = 1}^d V_i = C}
    \; .
\end{align}
To complete the proof of~\cref{eq:pointprob}, it thus suffices to show that $ \prb{\sum_{i \in [k]} V_i = C}= \tfrac{1}{4\sqrt{C}}$.
We use \Cref{lem:fourier} to get that,
\begin{align*}
    \prb{\sum_{i = 1}^d V_i = C} \ge \frac{1}{\sqrt{2 \pi \sum_{i = 1}^d \var{V_i}}} - K_2 \left( \frac{\sum_{i = 1}^d \ep{\abs{V_i - \ep{V_i}}^3}}{\left(\sum_{i = 1}^d \var{V_i}\right)^{3/2}} \right)^2
\end{align*}
Now we use that $\ep{V_i} \le \var{V_i} \le 2 \ep{V_i}$, $\abs{V_i - \ep{V_i}}^3 \le 25 \ep{V_i}$, and $\sum_{i = 1}^d \ep{V_i} = C$ to get that,
\begin{align*}
    \prb{\sum_{i = 1}^d V_i = C} 
        \ge \frac{1}{\sqrt{4 \pi C}} - 25^2 K_2 \frac{1}{C^2} 
\end{align*}
We know that $C \ge \tfrac{1}{4L}$ so if we choose $L$ sufficiently small we get that,
\begin{align*}
    \prb{\sum_{i = 1}^d V_i = C}
        \ge \frac{1}{4\sqrt{C}}
    \; .
\end{align*}
This leads to the desired bound.

We finish the section by proving \Cref{claim:rate-of-decrease}.
\begin{proof}[Proof of \Cref{claim:rate-of-decrease}]
    We start by using \cref{eq:exponential-twist} to get that,
    \begin{align*}
        \frac{\prb{S_d = C}}{\prb{S_d = C - 1}}
            = e^{-\lambda} \frac{\prb{\sum_{i = 1}^d V_i = C}}{\prb{\sum_{i = 1}^d V_i = C - 1}}
    \end{align*}
    We want to argue that $\frac{\prb{\sum_{i = 1}^d V_i = C}}{\prb{\sum_{i = 1}^d V_i = C - 1}} \ge \max\set{e^{-1/8}, e^{-\tfrac{1}{8} \beta/C}}$.
    First we use \Cref{lem:fourier} to get that,
    \begin{align*}
        \frac{\prb{\sum_{i = 1}^d V_i = C}}{\prb{\sum_{i = 1}^d V_i = C - 1}}
            \ge \frac{\tfrac{1}{\sqrt{2\pi \sum_{i = 1}^d \var{V_i}}} - K_2 \left( \tfrac{\sum_{i = 1}^d \ep{\abs{V_i - \ep{V_i}}^3}}{\left(\sum_{i = 1}^d \var{V_i}\right)^{3/2}} \right)^2}{\tfrac{1}{\sqrt{2\pi \sum_{i = 1}^d \var{V_i}}} e^{-1/(2\sum_{i = 1}^d \var{V_i})} + K_2 \left( \tfrac{\sum_{i = 1}^d \ep{\abs{V_i - \ep{V_i}}^3}}{\left(\sum_{i = 1}^d \var{V_i}\right)^{3/2}} \right)^2}
    \end{align*}
    Now we use that $\ep{V_i} \le \var{V_i} \le 2 \ep{V_i}$, $\abs{V_i - \ep{V_i}}^3 \le 25 \ep{V_i}$, and $\sum_{i = 1}^d \ep{V_i} = C$ to get that,
    \begin{align*}
        \frac{\tfrac{1}{\sqrt{2\pi \sum_{i = 1}^d \var{V_i}}} - K_2 \left( \tfrac{\sum_{i = 1}^d \ep{\abs{V_i - \ep{V_i}}^3}}{\left(\sum_{i = 1}^d \var{V_i}\right)^{3/2}} \right)^2}{\tfrac{1}{\sqrt{2\pi \sum_{i = 1}^d \var{V_i}}} e^{-1/(2\sum_{i = 1}^d \var{V_i})} + K_2 \left( \tfrac{\sum_{i = 1}^d \ep{\abs{V_i - \ep{V_i}}^3}}{\left(\sum_{i = 1}^d \var{V_i}\right)^{3/2}} \right)^2}
            \ge \frac{1 - 25^2 K_2 \sqrt{8 \pi \tfrac{1}{C} } }{e^{-1/(4C)} + 25^2 K_2 \sqrt{8 \pi \tfrac{1}{C} }}
    \end{align*}
    Using that $e^{-1/(4c)} \le 1 - \tfrac{1}{8c}$ we the get that,
    \begin{align*}
        \frac{1 - 25^2 K_2 \sqrt{8 \pi \tfrac{1}{C} } }{e^{-1/(4C)} + 25^2 K_2  \sqrt{8 \pi \tfrac{1}{C} }}
            \ge \frac{1 - 25^2 K_2 \sqrt{8 \pi \tfrac{1}{C} } }{1 - \tfrac{1}{8c} + 25^2 K_2  \sqrt{8 \pi \tfrac{1}{C} }}
            = 1 - \frac{2 \cdot 25^2 K_2 \sqrt{8 \pi \tfrac{1}{C}} - \tfrac{1}{8c}}{1 - \tfrac{1}{8c} + 25^2 K_2 \sqrt{8 \pi \tfrac{1}{C} }}
    \end{align*}
    We now that $C \ge \tfrac{1}{4L}$ so choosing $L$ sufficiently small it holds that
    \begin{align*}
        \frac{2 \cdot 25^2 K_2  \sqrt{8 \pi \tfrac{1}{C}} - \tfrac{1}{8c}}{1 - \tfrac{1}{8c} + 25^2 K_2  \sqrt{8 \pi \tfrac{1}{C} }}
            \le \frac{2 \cdot 25^2 K_2 \sqrt{8 \pi}}{\sqrt{C}} 
    \end{align*}
    This implies that,
    \begin{align*}
        \frac{\prb{\sum_{i = 1}^d V_i = C}}{\prb{\sum_{i = 1}^d V_i = C - 1}}
            \ge 1 - \frac{2 \cdot 25^2 K_2 \sqrt{8 \pi}}{\sqrt{C}}
    \end{align*}
    Clearly,  $1 - \frac{2 \cdot 25^2 K_2 \sqrt{8 \pi}}{\sqrt{C}} \ge e^{-1/8}$ by choosing $L$ small enough.
    We also note that,
    \begin{align*}
        1 - \frac{2 \cdot 25^2 K_2 \sqrt{8 \pi}}{\sqrt{C}}
            \ge e^{-\tfrac{4 \cdot 25^2 K_2 \sqrt{8 \pi}}{\sqrt{C}}}
            \ge e^{-\tfrac{M}{8 \sqrt{C}}}
            \ge e^{-\tfrac{\beta}{8 C}}
        \; .
    \end{align*}
    By choosing $M$ large enough.
    The last inequality follows since $\beta \ge M \sqrt{C}$.

    We now have to bound $\lambda$.
    We define the function
    \[
        h(x) = \sum_{i = 1}^d \alpha_i e^x (1 - \alpha_i e^x)^{-1}
            \; .
    \]
    We note that $h(0) = \sum_{i = 1}^d \ep{\mathcal{Y}_i} = C + \beta$.
    We take the derivative of $h$ twice and get that,
    \begin{align*}
        h'(x) &= \sum_{i = 1}^d \alpha_i e^x (1 - \alpha_i e^x)^{-2} \\
        h''(x) &= \sum_{i = 1}^d a e^x (1 + e^x) (1 - \alpha_i e^x)^{-3}
    \end{align*}
    We note that $h'(x) \ge 0$ and $h''(x) \ge 0$ for all $x$ so $h$ is a monotonically increasing convex function,
    and $h'(0) = \sum_{i = 1}^d \var{\mathcal{Y}_i} \le \sum_{i = 1}^d 2\mu_i = 2(C + \beta)$.
    
    If $\beta \ge C$ then again using that $h$ is convex we get that,
    \begin{align*}
        h(- \tfrac{1}{4}) 
            &\ge h(0) - \tfrac{1}{4}  h'(0)
            \ge C + \beta - 2 \tfrac{1}{4} (C + \beta)
            \ge C
        \; .
    \end{align*}
    Since $h$ is increasing then it implies that $\lambda \le - \tfrac{1}{4} $ and we get that,
    \begin{align*}
        \frac{\prb{S_d = C}}{\prb{S_d = C - 1}}
            = e^{-\lambda} \frac{\prb{\sum_{i = 1}^d V_i = C}}{\prb{\sum_{i = 1}^d V_i = C - 1}}
            \ge \ e^{\tfrac{1}{4} } e^{-\tfrac{1}{8}}
            = e^{-\tfrac{1}{8}}
        \; .
    \end{align*}

    If $\beta < C$ then using that $h$ is convex we get that,
    \begin{align*}
        h(- \tfrac{1}{4} \beta/C) 
            &\ge h(0) - \tfrac{1}{4} \beta/C h'(0)
            \ge C + \beta - 2 \tfrac{1}{4} \beta/C  (C + \beta)
            \\&= C + (1 - 2 \tfrac{1}{4} - 2 \tfrac{1}{4} \beta/C )\beta
            \ge C + (1 - 4 \tfrac{1}{4})\beta
            \ge C
        \; .
    \end{align*}
    Since $h$ is increasing then it implies that $\lambda \le - \tfrac{1}{4} \beta/C$ and we get that,
    \begin{align*}
        \frac{\prb{S_d = C}}{\prb{S_d = C - 1}}
            = e^{-\lambda} \frac{\prb{\sum_{i = 1}^d V_i = C}}{\prb{\sum_{i = 1}^d V_i = C - 1}}
            \ge e^{\tfrac{1}{4} \beta/C} e^{-\tfrac{1}{8} \beta/C}
            = e^{\tfrac{1}{8} \beta/C}
        \; .
    \end{align*}

    We now focus on the upper bound.
    We use \cref{eq:exponential-twist} to get that,
    \begin{align*}
        \frac{\prb{S_d = C - 1}}{\prb{S_d = C - 1 - \tfrac{1}{4}\ceil{\tfrac{C}{\beta}}}}
            = e^{-\lambda \tfrac{1}{4}\ceil{\tfrac{C}{\beta}}} \frac{\prb{\sum_{i = 1}^d V_i = C - 1}}{\prb{\sum_{i = 1}^d V_i = C - 1 - \tfrac{1}{4}\ceil{\tfrac{C}{\beta}}}}
    \end{align*}
    We start by lower bounding $\lambda$.
    We first note that $h'(x) \le h(x)$ for all $x$.
    Using that $h$ is convex we get that,
    \begin{align*}
        C + \beta = h(0)
            \ge h(- \beta/C) + \beta/C h'(-\beta/C)
            \ge \frac{C + \beta}{C} h(-\beta/C)
        \; .
    \end{align*}
    This implies that $h(-\beta/C) \le C$ and $\lambda \ge - \beta/C$.
    Now we will bound $\frac{\prb{\sum_{i = 1}^d V_i = C - 1}}{\prb{\sum_{i = 1}^d V_i = C - 1 - \ceil{\tfrac{C}{\beta}}}}$.
    We will again use \Cref{lem:fourier}.
    \begin{align*}
        \frac{\prb{\sum_{i = 1}^d V_i = C - 1}}{\prb{\sum_{i = 1}^d V_i = C - 1 - \tfrac{1}{4}\ceil{\tfrac{C}{\beta}}}}
            &\le \frac{\tfrac{1}{\sqrt{2\pi \sum_{i = 1}^d \var{V_i}}} e^{-1/(2\sum_{i = 1}^d \var{V_i})} + K_2 \left( \tfrac{\sum_{i = 1}^d \ep{\abs{V_i - \ep{V_i}}^3}}{\left(\sum_{i = 1}^d \var{V_i}\right)^{3/2}} \right)^2}
                {\tfrac{1}{\sqrt{2\pi \sum_{i = 1}^d \var{V_i}}} e^{-(1 + \tfrac{1}{4}\ceil{\tfrac{C}{\beta}})^2/(2\sum_{i = 1}^d \var{V_i})} - K_2 \left( \tfrac{\sum_{i = 1}^d \ep{\abs{V_i - \ep{V_i}}^3}}{\left(\sum_{i = 1}^d \var{V_i}\right)^{3/2}} \right)^2}
            \\&\le \frac{\tfrac{1}{\sqrt{2\pi C}} e^{-1/(2\sum_{i = 1}^d \var{V_i})} + 25^2 K_2 \tfrac{1}{C}}
                {\tfrac{1}{\sqrt{2\pi C}} e^{-(1 + \tfrac{1}{4}\ceil{\tfrac{C}{\beta}})^2/(2\sum_{i = 1}^d \var{V_i})} - 25^2 K_2 \tfrac{1}{C}}
    \end{align*}
    Now we note that since $\beta \ge M \sqrt{C}$ then we get that $(1 + \tfrac{1}{4}\ceil{\tfrac{C}{\beta}})^2/(2\sum_{i = 1}^d \var{V_i}) \le \tfrac{1}{12}$.
    We then get that,
    \begin{align*}
        \frac{\prb{\sum_{i = 1}^d V_i = C - 1}}{\prb{\sum_{i = 1}^d V_i = C - 1 - \tfrac{1}{4}\ceil{\tfrac{C}{\beta}}}}
            &\le \frac{\tfrac{1}{\sqrt{2\pi C}} + 25^2 K_2 \tfrac{1}{C}}{\tfrac{1}{\sqrt{2\pi C}} e^{-1/12} - 25^2 K_2 \tfrac{1}{C}}
            \le e^{1/6}
    \end{align*}
    The last inequality follows by $C \ge \tfrac{1}{4L}$ and choosing $L$ small enough.
    We then get that,
    \begin{align*}
        \frac{\prb{S_d = C - 1}}{\prb{S_d = C - 1 - \tfrac{1}{4}\ceil{\tfrac{C}{\beta}}}}
            &= e^{-\lambda \tfrac{1}{4}\ceil{\tfrac{C}{\beta}}} \frac{\prb{\sum_{i = 1}^d V_i = C - 1}}{\prb{\sum_{i = 1}^d V_i = C - 1 - \tfrac{1}{4}\ceil{\tfrac{C}{\beta}}}}
            \\&\le e^{\tfrac{\beta}{C} \tfrac{1}{4}\ceil{\tfrac{C}{\beta}}} e^{1/6}
            \\&\le e^{1/2 + 1/6}
            \\&< 2
    \end{align*}
\end{proof}

\section{The Number of Bins Visited During an Insertion}\label{sec:insertion}
This section is dedicated to proving the part of~\Cref{thm:main} concerning insertions, which we restate below.
\begin{theorem}\label{thm:main-insertion1}
Let $\balls,\bins\in \N$ and $0<\eps<1$ with $1/\eps=n^{o(1)}$. Let $C=(1+\eps)\balls/\bins$. Suppose we insert $\balls$ balls into $\bins$ bins, each of capacity $C$, using consistent hashing with bounded loads and virtual bins having $k$ levels where $k=c/\eps^2$ for $c$ a sufficiently large universal constant. The expected number of bins visited during an insertion of a ball is $O(1/f)$.
\end{theorem}

In fact, the proof uses only that the total number of non-full bins is $\Theta(f)$ with high probability, not the concrete value of $f$. Therefore the complicated expression for $f$ will never occur in the proof of the theorem. All we will occasionally use is the fact that the number of non-full bins is $\Omega(\eps)$, which follows trivially from a combinatorial argument.

The section is structured as follows: We start by providing some preliminaries for the proof of~\Cref{thm:main-insertion1} in~\Cref{sec:insertion-prelim}. In~\Cref{sec:worstcase}, we use the results from~\Cref{sec:f-properties} to provide a strengthening of ~\Cref{thm:worstcase1}. Finally, we provide the proof of~\Cref{thm:main-insertion1} in~\Cref{sec:insertion-proof}.

\subsection{Preliminaries For the Analysis}\label{sec:insertion-prelim}
We start by making the following definition which will be repeatedly be useful in the analysis to follow.
\begin{definition}\label{def:close to full}
Consider any distribution of $\balls$ balls into $\bins$ bins. 
We say that a bin is \emph{close to full} if it contains more than $(1+\eps/2)\balls/\bins$ balls. Otherwise, we say that it is \emph{far from full}.
\end{definition}
Suppose we distribute $\balls$ balls into $\bins$ bins each of capacity $C=(1+\eps)\balls/\bins$ using consistent hashing with bounded loads and virtual bins. By Theorem~\ref{thm:jakobs}, the number of non-full bins is $\Theta(fm)$ with high probability when $k=O(1/\eps^2)$ is sufficiently large. We claim that it also holds that the number of far from full bins is $\Theta(fm)$ with high probability. To see this, suppose that after distributing the $\balls$ balls into the $\bins$ bins of capacity $C=(1+\eps)\balls/\bins$ each, we reduce the capacity of each bin to $C_0=(1+\eps/2)\balls/\bins$. This requires forwarding balls from the now overflowing bins and this forwarding can only increase the number of bins containing $(1+\eps/2)\balls/\bins$ balls. By Theorem~\ref{thm:jakobs}, and with $\eps_0=\eps/2$, the number of non-full bins after the relocating is $\Theta(f_0m)$, where 
$$
f_0=\begin{cases}
\eps_0 C_0, &C_0\leq \log(1/\eps_0) \\
\eps_0\sqrt{C_0\log\left(\frac{1}{\eps_0\sqrt C_0}\right)}, &\log(1/\eps_0)<C_0\leq  \frac{1}{2\eps_0^2} \\
1, &\frac{1}{2\eps_0^2}\leq C_0.
\end{cases}
$$
But clearly, $f_0=\Theta(f)$, so we conclude that the number of far from full bins before modifying the system is $\Theta(fm)$ with high probability.

Summing up, we have the following corollary to Theorem~\ref{thm:jakobs}.
\begin{corollary}\label{cor:farfromfull}
In the setting of Theorem~\ref{thm:main-insertion1}, the number of far from full bins is $\Theta(fm)$ with high probability, i.e., with probability $1-n^{-\gamma}$ for every $\gamma$=O(1).
\end{corollary}
Finally, recall~\Cref{def:run}: The \emph{run} at a given level $i$ containing some virtual bin $b$, is the maximal interval at level $i$ which contains $b$ and satisfies that all bins lying in $I$ gets full at level $i$. 
\subsection{High Probability Bound on the Number of Bins Visited in an Insertion}\label{sec:worstcase}
This section will be dedicated to prove the following strengthening of~\Cref{thm:worstcase1}. 
\begin{theorem}\label{thm:worstcase}
Let $\balls,\bins\in \N$ and $0<\eps<1$ with $1/\eps=n^{o(1)}$.  Suppose we distribute $\balls$ balls into $\bins$ bins, each of capacity $C=(1+\eps)\balls/\bins$, using consistent hashing with bounded loads and virtual bins and $k= c/\eps^2$ levels for a sufficiently large constant $c$.  Let $b$ be a bin at level $i$ which may be chosen dependently on the hashing of balls and bins to level $1,\dots,i-1$ and $I$ the run at level $I$ containing $b$. Let $X$ denote the number of bins in $I$. For any $t\geq 1/f$, 
$$
\Pr[X\geq t]=\exp(-\Omega(tf))+O(n^{-10})
$$
The same statement holds even if $b$ is given an extra start load of $\lambda tf\lceil C \eps /2\rceil$ 'artificial' balls before the hashing of balls and bins to level $i$, where $\lambda$ is a sufficiently small constant.
\end{theorem}
Note that it in particular follows that the number of bins visited at a given level during an insertion is $O(\log (1/\delta)/f)$ with probability $1-\delta$.
\begin{proof}

Let $R$ denote the number of virtual bins in $I$. By Corollary~\ref{cor:farfromfull}, the number of far from full bins after inserting balls at level $1,\dots,i-1$ is at least $c_0fm$ with high probability, where $c_0>0$ is some universal constant. Furthermore, by a standard Chernoff bound, the number of balls hashing to level $i$ is at most $2n/k$ with high probability. Here we used the assumption that $1/\eps=n^{o(1)}$, so $n\gg k\log n$. Condition on those two events and consider the following modified process at level $i$ where (1) $b$ and every bin which was close to full after inserting the balls at level $1,\dots,i-1$ forwards every ball it receives at level $i$, i.e., has its remaining capacity reduced to zero (2) each far from full bin stores at most $\lceil C\eps/2\rceil $ balls from level $i$ before it starts forwarding balls at level $i$, i.e., has its remaining capacity reduced to $\lceil C\eps/2\rceil $. Let $I'$ denote the run containing $b$ with such modified capacities. Letting $R'$ denote the number of virtual bins lying in $I$ it then holds that $R\leq R'$, so it suffices to provide a high probability upper bound on $R'$. 

Let $s\in \N$ be given and let $A_s$ be the event that $s+1\leq R' \leq 2s+1$. Define $J_1^{-}$ and $J_1^+$ to be respectively the intervals at level $i$ ending and starting at $b$ and having length $s/(4m)$. Similarly, let $J_2^{-}$ and $J_2^+$ be respectively the intervals at level $i$ ending and starting at $b$ and having length $4s/m$. We observe that if $A_s$ occur then either of the following events must hold.
\begin{itemize}
\item[$B_1$:] $J_2^{-}$ or $J_2^{+}$ contains at most $3s$ virtual bins.
\item[$B_2$:] $J_1^{-}$ or $J_1^{+}$ contains at least $s/2$ virtual bins
\item[$B_3$:] $J_1^{-}$ or $J_1^{+}$ contains at most $c_0fs/8$ virtual bins which were far from full from levels $1,\dots,i-1$
\item[$B_4$:] $J_2^{-} \cup J_2^{+}$ contains at least $\lceil C\eps/2\rceil \cdot c_0fs/8$ balls.

\end{itemize}
Indeed, suppose that $A_s$ occur and that neither of $B_1,B_2,B_3$ occur. We show that then $B_4$ must occur. To see this observe that if $B_1$ does not occur, then since $I'$ consists of at most $2s+1$ bins, $I'\subseteq J_2^{-}\cup J_2^{+}$. Since $B_2$ does not occur, $I'$ must further fully contain $J_1^{-}$ or $J_1^{+}$. Since $B_3$ does not occur, $I'$ must then contain at least $c_0fs/8$ virtual bins which were far from full from levels $1,\dots,i-1$. Finally any ball allocated to a bin of $I'$ must also hash to $I'$. Since the at least $c_0fs/8$ far from full bins from level $1,\dots,i-1$ which lie in $I'$ each get full at level $i$ and has a total capacity of $\lceil C\eps/2\rceil \cdot c_0fs/8$, it follows that at least $\lceil C\eps/2\rceil \cdot c_0fs/8$ balls must hash to $I'\subseteq J_2^{-}\cup J_2^{+}$. This is exactly the event $B_4$.

As in the proof of Lemma~\ref{thm:worstcase1}, we can use standard Chernoff bounds to conclude that $\Pr[B_1]=\exp(-\Omega(s))$, $\Pr[B_2]=\exp(-\Omega(s))$ and $\Pr[B_3]= \exp(-\Omega(fs))$. For $B_4$, we observe that the expected number of balls, $\mu$, hashing to $J_2^{-}\cup J_2^{+}$ is upper bounded by $2n/k\cdot 8s/m=O(Cs/k)$. As $f=\Omega(\eps)$, we may assume that $k\geq c'/(\eps f)$ for any constant $c'$. Thus, choosing $c'$ sufficiently large, it follows that $\mu\leq Cs\eps fc_0/ 32$. Using another Chernoff bound, it follows $\Pr[B_4]= \exp(-\Omega(fs))$. In conclusion, if $s\geq 1/f$, it holds that  $\Pr[A_s]=\exp(-\Omega(fs))$ and the desired result follows as in the proof of Lemma~\ref{thm:worstcase1}.

Finally, it is easy to modify the constants in the above argument, so that it carries through even when $b$ is given an extra start load of $\lambda t f \lceil C\eps/2\rceil$ balls for a sufficiently small constant $\lambda$, and this gives the final statement of the Theorem.
\end{proof}

\subsection{The Proof of Theorem~\ref{thm:main-insertion1} }\label{sec:insertion-proof}
In this section, we provide the proof of Theorem~\ref{thm:main-insertion1}. In order to do so, we first require a technical lemma which for a given virtual bin, $b$, bounds the number of balls that are either placed in $b$ or forwarded from $b$ at level $i$. 
The technique used to prove this lemma will be used for the final proof of Theorem~\ref{thm:main-insertion1}, but in a more sophisticated way. As such, the lemma below serves as a nice warm up to the proof of Theorem~\ref{thm:main-insertion1}. We start out by choosing $s^*=O(\log n/f)$ sufficiently large, such that Theorem~\ref{thm:worstcase} yields that for a bin $b$ at level $i$, the length of the run containing $b$ at level $i$ (see Definition~\ref{def:run}) has length at most $s^*$ with probability $1-1/n^{10}$.
\begin{lemma}\label{lemma:directballs}
Let $\lambda=O(1)$ be any constant. Let $b$ be a virtual bin at level $i$ that may depend on the distribution of balls into bins at level $1,\dots,i-1$. Let $n_i$ denote the number of balls hashing to level $i$ and suppose that $n/(2k) \leq n_i\leq 2n/k$ where $k=O(1/\eps^2)$ is sufficiently large (depending on $\lambda$).
Let $Z$ denote the number of the $n_i$ balls hashing to level $i$ that either are placed in $b$ or are forwarded from $b$. Define $\alpha =\lceil C \eps/2 \rceil $. For any $\ell\geq 1/5$ satisfying that $\alpha \ell$ is an integer\footnote{The constant $5$ is arbitrary.}, it holds that 
$$
\Pr[Z\geq \alpha \ell]\leq e^{-\lambda\ell}+1/n^{10}.
$$ 
\end{lemma}
\begin{proof}
We define $A_{\ell}$ to be the event that $Z\geq \alpha \ell$. When upper bounding the probability of $A_\ell$ we may assume that every bin which was close to full at level $i-1$ forwards all balls landing in it at level $i$. We may further assume that any bin which was far from full at level $i-1$ stores exactly $\alpha= \lceil C\eps/2\rceil$ balls and then starts forwarding balls. Let $Z'$ denote the number of balls landing in $b$ or being forwarded from $b$ at level $i$ in this modified process. Then clearly, $Z'\geq Z$ so $\Pr[Z\geq \alpha \ell] \leq \Pr[Z'\geq \alpha \ell]$. 

\begin{figure}[h]
\centering
\includegraphics{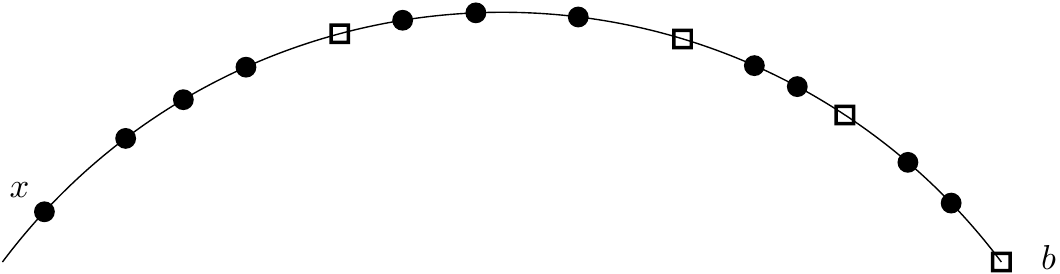}
\caption{$s$ bins that are far from full and $\alpha s+\alpha \ell$ balls. The bins are represented as boxes and the balls as disks.}
\label{figure:ballsinb}
\end{figure}

Next note that if $Z'\geq \alpha \ell$, then there must an integer $s\geq 0$ and an interval of the $i$'th level ending in $b$ which contains exactly $s$ virtual bins which are far from full and exactly $\alpha s +\alpha \ell$ balls. See Figure~\ref{figure:ballsinb}. Indeed, of the $\ell$ balls landing or being forwarded from $b$ consider the one hashing furthest behind $b$ at level $i$, call it $x$. Let $s$ be the the number of far from full  bins hashing between $x$ and $b$ at level $i$. Aside from the $\alpha \ell$ balls landing in $b$ or being forwarded from $b$, there must hash enough balls between $x$ and $s$ to put $\alpha$ balls in each of the far from full bins between $x$ and $b$, and thus the interval between $x$ and $b$ contains exactly $s$ far from full bins and $\alpha s + \alpha \ell$ balls. We denote the event that there exists such an interval by $A_{\ell,s}$ noting that we may then upper bound $\Pr[A_\ell]\leq \sum_{s=0}^{s^*} \Pr[A_{\ell,s}]+1/n^{10}$. Here we used that the run containing $b$  has length at most $s^*$ with probability at least $1-1/\balls^2$. We proceed to upper bound $\Pr[A_{\ell,s}]$ for each $0 \leq s \leq s^* $.

\begin{figure}[h]
\centering
\includegraphics{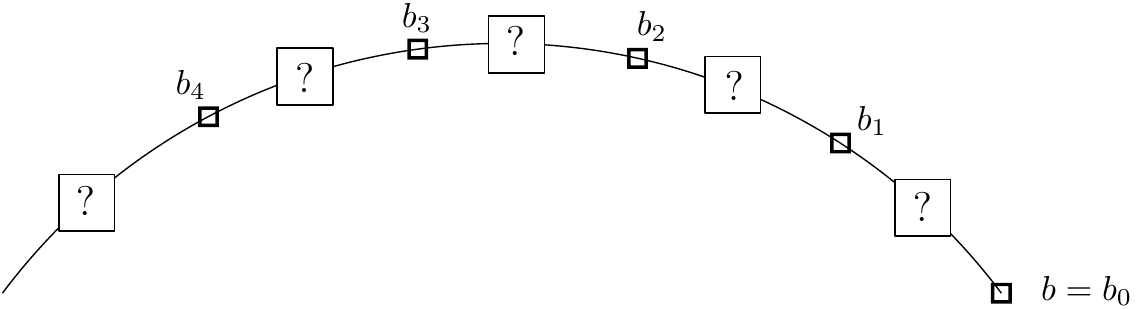}
\caption{We generate the number of balls hashing directly between $b_i$ and $b_{i+1}$ sequentially. In each step the number of such balls is dominated by a geometric distribution with parameter $q$.}
\label{figure:ballsinb2}
\end{figure}
So fix $s\geq 0$. 
We generate the sequence of the $s+1$ far from full bins $b=b_0, b_1,\dots,b_{s+1}$ leading up to $b$ and the balls hashing between them in a backwards order. Starting at $b_0$ we go backwards along the cyclic order. At some point we reach a bin, $b_1$ and we let $X_0$ be the number of balls met along the way in the between $b_0$ and $b_1$. We continue this was, going backwards until we have met $s+1$ bins $b_1,\dots,b_{s+1}$ and for each $1\leq i \leq s$ we let $X_i$ be the number of balls met in the cyclic order between $b_i$ and $b_{i+1}$. See Figure~\ref{figure:ballsinb2} for an illustration of the process. Let $f$ denote the fraction of bins which were far from full from level $1,\dots,i-1$. As we saw after Definition~\ref{def:close to full}, $f\geq \eps /3$. Now when going backwards from $b_i$ until we get to $b_{i+1}$, the probability of meeting a ball in each step is upper bounded by $\frac{n_i}{n_i+ mf-s}\leq \frac{n_i}{n_i+ mf-s^*}:=q$ regardless of the values of $X_0,\dots,X_{i-1}$.  Letting $X_0',\dots,X_s'$ be independent geometric variables with parameter $1-q$,  $X=\sum_{i=0}^s X_i$, and $X'=\sum_{i=0}^s X_i'$ it follows tht for any $t>0$, $\Pr[X\geq t]\leq \Pr[X'\geq t]$.

If $A_{\ell,s}$ holds, then $X\geq s \alpha+\ell \alpha$, so we may upper bound
$$
\Pr[A_{\ell,s}]\leq \Pr[X'\geq s\alpha+\ell \alpha].
$$
The expected value of $X_i'$ is 
$$
\mathbb{E}[X_i']=\frac{q}{1-q}=\frac{\balls_i}{\bins f-s^*}=O\left( \frac{n}{kmf} \right)=O\left( \frac{\alpha}{kf \eps} \right)\leq \frac{\alpha}{\lambda_0}.
$$
Here $\lambda_0=O(1)$ is a sufficiently large constant which we will choose later. Here we again used the assumption that $1/\eps=m^{o(1)}$ and moreover that $k=O(1/\eps^2)$  is sufficiently large. It follows that $\mathbb{E}[X']=\frac{(s+1)\alpha}{\lambda_0}$. Note in particular that we can ensure that $\E[X']\leq \frac{s\alpha+\ell \alpha}{2}$, so that 
$$
\Pr[X'\geq s\alpha+\ell \alpha]\leq \Pr\left[X'\geq \E[X']+ \frac{s\alpha+\ell \alpha}{2}\right].
$$
We apply~\Cref{thm:geotail} to bound this quantity. If we are in the case, where we are to use the second bound of~\cref{eq:geotail}, we obtain that 
$$
 \Pr\left[X'\geq \E[X']+ \frac{s\alpha+\ell \alpha}{2}\right]\leq \left(1 - \frac{1}{1 + 2 \alpha/\lambda_0} \right)^{\frac{s\alpha+\ell \alpha}{4}},
$$
It is easy to check that $\left(1 - \tfrac{1}{1 + 2 \alpha/\lambda_0} \right)^{\alpha}$ can be made smaller than any sufficiently small constant, just by choosing $\lambda_0$ sufficiently large. Thus it follows that 
\begin{align}\label{eq:rightbound}
\Pr[X'\geq s\alpha+\ell \alpha]\leq e^{-\lambda_1(s+\ell)},
\end{align} 
where we can make $\lambda_1=O(1)$ sufficiently large. However, we may have to use the first bound of~\cref{eq:geotail} and we investigate now which bound we obtain in this case. Relating back to~\Cref{thm:geotail}, we define $\mu_0=\E[X_i']\leq \alpha/\lambda_0$, $A=\left(1 + \tfrac{1}{2 \mu_{0}}\right) \log \left(1 + \tfrac{1}{2\mu_0} \right)$ and $t=\frac{1}{4\sigma^2}\frac{s\alpha+\ell \alpha}{2}$. We further define $\sigma^2=\Var[X']=(s+1)\Var[X_i']=(s+1)\mu_0(1+\mu_0)$. If $\mu_0\geq 1$, then $A\leq 1/\mu_0$ and $\sigma^2\leq (s+1)2\mu_0^2$, so that
$$
A\sigma^2\leq 2(s+1)\mu_0\leq \frac{2(s+1)\alpha}{\lambda_0}< \frac{s\alpha+\ell \alpha}{8}=t\sigma^2,
$$
by choosing $\lambda_0$ large enough. Thus, in this case we obtain the bound in~\cref{eq:rightbound}. If on the other hand $\mu_0<1$, then 
$$
t\geq \frac{s\alpha+\ell\alpha}{16(s+1)\mu_0}\geq \frac{\lambda_0(s+\ell)}{16(s+1)}\geq \lambda_2
$$
for a sufficiently large constant $\lambda_2$. Then also $W_0(t)$ can be made larger than any given constant, so we obtain that the bound of~\cref{eq:rightbound} holds in general.

We now sum over $s$ to obtain that 
$$
\Pr[A_\ell]\leq 1/n^{10}+ \sum_{s=0}^{s^*} \Pr[A_{\ell,s}] \leq  1/n^{10}+\sum_{s=0}^{s^*}e^{-\lambda_1(s+\ell)}\leq 1/n^{10}+e^{-\lambda \ell},
$$
where again $\lambda$ can be made sufficiently large. This completes the proof.
\end{proof}
With this lemma in hand we are ready to proceed with the proof of Theorem~\ref{thm:main-insertion1}. To guide the reader, we will start by providing a high level idea of how to obtain the result as follows. First of all, it will be helpful to recall in details how an insertion of a ball is handled using consistent hashing with bounded loads and virtual bins. When inserting a ball, $x$, we uniformly hash $x$ to a random point at a random level. Suppose that the hash value of $x$, $h(x)$, lies in the $i$'th level $i$ for some $i$. Starting at $h(x)$ we walk along  level $i$ until we arrive at a virtual bin. If the virtual bin is filled to its capacity with balls hashing to level $1,\dots, i$, we forward a ball from that bin at level $i$ (it could be $x$ but it could also be another ball that hashed to level $i$ of lower priority than $x$). We repeat the step, continuing to walk along level $i$ until we meet a new virtual bin. The first time we meet a virtual bin, $b$, which was not filled to its capacity with balls hashing to level $1,\dots,i$, we insert the forwarded ball and find the smallest level $j>i$ such that the virtual bin of $b$ at level $j$ received a ball at level $j$. If no such level exists, the insertion is completed. Otherwise $b$ has an overflow of one ball at level $j$, and we continue the insertion walking along level $j$ starting at $b$. Theorem~\ref{thm:main-insertion1} claims that the expected number of bins visited during this entire process is upper bounded by $O(1/f)$.

The idea of in our proof of Theorem $~\ref{thm:main-insertion1}$ is to split the bins visited during the insertion of $x$ into~\emph{epochs}. An epoch starts by visiting $\lceil 1/f \rceil$ virtual bins of the insertion (unless of course the insertion is completed before that many bins has been seen). The last of these $\lceil 1/f \rceil$ virtual bins lies at some level $i$ and we finish the epoch by completing the forwarding of balls needed at level $i$. At this point, we are either done with the insertion or we need to forward a ball from some virtual bin at some level $j>i$. The next epochs are similar; having finished epoch $a-1$, in epoch $a$, we visit $\lceil 1/f \rceil$ virtual bins. At this point, we will be at some level $\ell$ if we are not already done with the insertion. We then finish the part of the insertion which takes place at level $\ell$. Importantly, at the beginning of each epoch, we have just arrived at a virtual bin at a completely fresh level. 

The proof shows that during the first $\lceil 1/f \rceil$ steps of an epoch, the probability of finishing the insertion in each step is $1-\Omega(f)$. The intuition for this, is that when we reach a bin $b$ at some level, $i$, the probability that $b$ is far from full from other levels than $i$ can be showed to be $\Omega(f)$. Since the number of levels $k=O(1/\eps^2)$ is large, the contribution from level $i$ to $b$ only fills $b$ with probability $1-\Omega(1)$. 
Thus, the probability of not finishing the insertion during the first $\lceil 1/f \rceil$ steps of an epoch is $(1-\Omega(f))^{\lceil 1/f \rceil}=e^{-\Omega(1)}=1-\Omega(1)$. Now conditioning on not finishing the insertion during the first $\lceil 1/f \rceil$ steps of an epoch, we can still show that the expected number of bins visited during the rest of the epoch is $O(1/f)$. Letting $\mathcal{E}$ denote the event of finishing the insertion during the first $\lceil 1/f \rceil$ of an epoch and $T$, the total number of bins visited during the insertion, we have on a very high level that
\begin{align}\label{eq:recursion}
\mathbb{E}[T]\leq \Pr[\mathcal{E}]\lceil 1/f \rceil+\Pr[\mathcal{E}^c](O(1/f)+\mathbb{E}[T])=O(1/f)+\Pr[\mathcal{E}^c]\mathbb{E}[T]=O(1/f)+p\mathbb{E}[T],
\end{align} 
where $p=1-\Omega(1)$. Solving this equation, we find that $\mathbb{E}[T]=O(1/f)$. Here it should be noted that the recursive formula~\eqref{eq:recursion} is a bit too simplified. In our analysis, the $\mathbb{E}[T]$ on the left hand side and on the right hand side of~\eqref{eq:recursion} will not exactly be the same. The point is that after finishing epoch $a$, and being ready to start epoch $a+1$ at a new level $j$, we will know a bit more about the hashing of balls to level $1,\dots,j-1$ than we did before the beginning of epoch $a$. However, using Lemma~\ref{thm:worstcase}, we know that it is only a relatively small fraction of the system that we have any information about, and so we can argue that the expectation does not change much.

With this intuition in mind, our next goal is to obtain Theorem~\ref{thm:main-insertion1}.
\begin{proof}[Proof of Theorem~\ref{thm:main-insertion1}]
As described above, we partition the insertion into \emph{epochs} where an epoch consists of the following two steps.
\begin{enumerate}
\item We go through $\lceil 1/f \rceil$ bins of the insertion ending in a bin at some level $\ell$. 
\item We continue the insertion at level $\ell$ until we arrive at some bin $b$ which does not get full at level $\ell$.
\end{enumerate}
After step $2.$~we will have to continue the insertion on some level $j>i$ (if $b$ gets full at that level). Note that the insertion will complete during an epoch if along the way, we meet a bin which does not get full on either of levels  $1,\dots, k$. We will prove the following more technical claim which implies Theorem~\ref{thm:main-insertion1}.
\begin{claim}\label{claim:main}
Let $D>0$ be any constant and $0\leq t\leq D \log n$. Condition on the event that the insertion has been through $t$ epochs so far. Let $\mathcal{E}$ denote the event that we finish the insertion at one of the first $\lceil 1/f \rceil$ bins met during step 1. of  epoch $t+1$. Further define $R$ to be the random variable which counts the number of bins visited during step 2. of epoch $t+1$ (if the insertion completes before we get to step 2.~we put $R=0$). Then 
\begin{align}\label{eq:tool1}
\Pr[\mathcal{E}]\geq c,
\end{align}
for some universal constant $c>0$ (which does not depend on $D$), and 
\begin{align}\label{eq:tool2}
\mathbb{E}[R \mid \mathcal{E}^c]= O(1/f).
\end{align}
\end{claim}
Before proving the claim, we argue how the desired result follows. First of all, choosing $D= 2/c$, it follows from~\eqref{eq:tool1} that the probability of not finishing the insertion during the first $D \log n$ epochs is upper bounded by
$$
(1-c)^{D \log n}\leq \exp(-2\log n)\leq n^{-2}.
$$
Conditioned on this extremely low probability event, the expected time for the insertion is crudely and trivially upper bounded by $mk$, but $mk n^{-2} \ll 1$, so this has no influence on the expected number of bins visited during the insertion, as we will now formalize.  For $1\leq i \leq D\log n$, we let $X_i$ denote the expected number of bins visited during the insertion starting from epoch $i$. If the insertion finishes before epoch $i$, we let $X_{i}=0$. Let further $\mathcal{E}_i$ denote the probability of finishing the insertion during step 1. of epoch $i$. Finally, let $R_i$ denote the number of bins visited during step 2. of epoch $i$. Then, for any $0\leq i \leq D\log n$, it holds that 
$$
\mathbb{E}[X_i]\leq \Pr[\mathcal{E}_{i}] \cdot \lceil 1/f \rceil+\Pr[\mathcal{E}_{i}^c](\mathbb{E}[R_i \mid \mathcal{E}_i^c]+\mathbb{E}[X_{i+1}]).
$$
By the claim, $\Pr[\mathcal{E}_{i}^c]\leq 1-c$ and $\mathbb{E}[R_i \mid \mathcal{E}_i^c]=O(1/f)$, so we obtain that
$$
\mathbb{E}[X_i]\leq O(1/f)+(1-c)\mathbb{E}[X_{i+1}].
$$
Solving this recursion, we obtain that 
$$
\mathbb{E}[X_0]=O(1/f)+(1-c)^i \mathbb{E}[X_{i+1}],
$$
so putting $i =D\log n$, we obtain that $\mathbb{E}[X_0]=O(1/f)+n^{-2} \cdot \mathbb{E}[X_{C\log n+1}]=O(1/f)$. But $\mathbb{E}[X_0]$ is exactly the expected number of bins visited during an insertion.  It thus suffices to prove the claim which is the main technical challenge of the proof.
\begin{proof}[Proof of Claim~\ref{claim:main}]
We split the proof into the proofs of equations~\eqref{eq:tool1} and~\eqref{eq:tool2}.
\subsection*{Proof of Equation~\eqref{eq:tool1}}

It suffices to show that for each of the $\lceil 1/f \rceil$  bins visited during step 1. of the epoch, the probability of ending the insertion at that bin is $\Omega(f)$.
More formally, we let $\mathcal{A}_i$ denote the event that the $i$'th of these bins, $1\leq i \leq \lceil 1/f \rceil$ is still full, i.e., that we do not end the insertion at the $i$'th bin, and show that $\prb{\mathcal{A}_i} \le (1 - \Omega(f))^i + i m^{-1/2 + o(1)}$.
The probability of not completing the insertion during step 1.~of the epoch is then upper bounded by $(1-\Omega(f))^{\lceil 1/f \rceil} + \ceil{1/f} m^{-1/2 + o(1)} \leq (1-\Omega(f))^{\lceil 1/f \rceil} + o(1) \le e^{-\Omega(1)} := c$ which is the desired result.
Here we used that $1/f \le O(1/\eps) = m^{o(1)}$.

We will condition on $\mathcal{A}_{i - 1}$ so start by making the conditioning more precise by describing exactly how the bins met before the $i$'th bin of the epoch at the given level received enough ball to make them full. We then bound the probability of $\mathcal{A}_i$ conditioned on this history. So fix $i$ with $1\leq i\leq \lceil 1/f \rceil$. The conditioning on $\mathcal{A}_{i-1}$ means that we have already seen $i-1$ full bins during the epoch. Suppose that the $i$'th bin, call it $b$, is at some level $\ell$. We then in particular know that the number of bins we have already visited at level $\ell$ is at most $i-1\leq1/f $. Let $a\geq 0$ denote the number of bins already visited on level $\ell$. Going backwards from $b:=b_a$, we denote these bins $b_{a-1},\dots,b_0$.  Thus $b_0$ was the first bin ever visited at level $\ell$. Note that possibly $b_a=b_0$. The conditioning $\mathcal{A}_{i-1}$ especially implies that after level $\ell$, all bins $b_0,\dots,b_{a-1}$ got filled. We now describe how these bins got filled at level $\ell$ as follows (see also Figure~\ref{figure:generatingorder2} for an illustration of the process). Starting with $j=0$, if the remaining capacity of $b_0$ after levels $1,\dots,\ell-1$ is $C_0$, we go backwards until at some point we have met a set of bins of total remaining capacity $C^*$ and exactly $C^*+C_0$ balls for some $C^*$. After this sequence, we insert a question mark \textbf{?}. This sequence of bins and balls describes how $b_0$ received its $C_0$ balls, and the \textbf{?} indicates a yet unknown history. We next go backwards from $b_1$ which has remaining capacity $C_1$, say. If we arrive at $b_0$ before having seen $C_1$ balls get we simply skip past the history of how $b_0$ got fills and continue the process after the \textbf{?}. If we obtain the description of how $b_1$ got filled at level $\ell$ before reaching $b_0$, there might still be more balls hashing between $b_0$ and $b_1$ (but no bins). In this case we insert a question mark, \textbf{?}, after the sequence of balls leading up to $b_1$. More generally, for $j=1,\dots,a-1$, we go backwards from $b_j$ generating a sequence of balls. Whenever we reach a bin, we go back to the nearest \textbf{?} and start generating balls at that point until we find a new bin or are done with describing the filling of $b_j$ --- In the later case we insert a new \textbf{?}. The \textbf{?} before bin $b_0$ has a special status. If we ever reach it, and we still require $C_j$ balls to be filled, we go backwards until we have found a set of bins of total remaining capacity $C^*$ and exactly $C^*+C_0$ balls for some $C^*$. It should be remarked that there is nothing probabilistic going on here. We have simply explained a way to find the positions of a set of balls and bins which certify how bins $b_0,\dots,b_{a-1}$ got filled at level $\ell$. See Figure~\ref{figure:generatingorder2} for an example of how this description of how bins $b_0,\dots,b_{a-1}$ got filled at level $\ell$ can look.

\begin{figure}[h]
\centering
\includegraphics{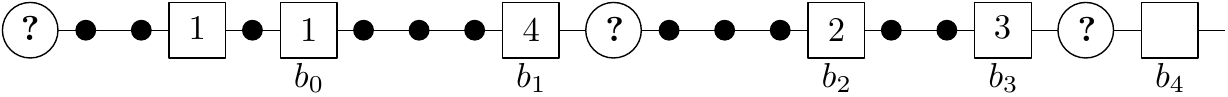}
\caption{The filling of bins $b_0,\dots,b_3$ at level $\ell$. The bins are represented as boxes an the numbers within them describes their remaining capacity at level $\ell$. The balls are represented as disks and the question marks \textbf{?} in circles.}
\label{figure:generatingorder2}
\end{figure}
 
We let $\mathcal{O}$ denote the event that bin $b$ receives more than $\lceil C\eps/2 \rceil$ bins from level $\ell$.
We also let $\mathcal{N}$ denote the event that $b$ receives at least $\frac{n}{m}(1+\lceil C\eps/2\rceil)$ balls from the levels different than $\ell$. 
We then get that,
\begin{align*}
    \prb{\mathcal{A}_i}
        \le \prb{\mathcal{A}_{i - 1} \wedge \left( \mathcal{O} \vee \mathcal{N} \right)}
        \le \prb{\mathcal{A}_{i - 1}}\prbcond{\mathcal{O}}{\mathcal{A}_{i - 1}}
            + \prb{\mathcal{A}_{i - 1} \wedge \mathcal{O}^c \wedge \mathcal{N}}
    \; .
\end{align*}
We will next show that $\prbcond{\mathcal{O}}{\mathcal{A}_{i - 1}} = p$, where $p=1-\Omega(1)$,
and 
\[
    \prb{\mathcal{A}_{i - 1} \wedge \mathcal{O}^c \wedge \mathcal{N}}
        \le \prb{\mathcal{A}_{i - 1}}\prbcond{\mathcal{O}^c}{\mathcal{A}_{i - 1}} (1 - c_0 f) + m^{-1/2 + o(1)}
\]
where $c_0=\Omega(1)$ is a universal constant.
This will then imply that,
\begin{align*}
    \prb{\mathcal{A}_i}
        \le \prb{\mathcal{A}_{i - 1}}(1 - (1 - p)c_0 f) + m^{-1/2 + o(1)}
        \le (1 - (1 - p)c_0 f)^i + i m^{-1/2 + o(1)}
    \; .
\end{align*}
We again split the proof into two parts.

\paragraph{Bounding $\prbcond{\mathcal{O}}{\mathcal{A}_{i - 1}}$:}
In the following we will omit the conditioning of $\mathcal{A}_{i - 1}$ from the notation to avoid clutter.
We have described how bins $b_0,\dots,b_{a-1}$ got filled at level $\ell$. This included a tail of balls behind each bin as well as some positions marked with \textbf{?}. Let $s+1$ be the number of such \textbf{?}-marks including the mark behind bin $b_0$. (See Figure~\ref{figure:generatingorder2}). Then $s\leq a$. Let $X_0$ denote the number of balls being forwarded to $b_a$ from the backmost \textbf{?} before $b_0$ and let $X_1,\dots,X_{s}$, denote the number of balls forwarded to $b_a$ from the remaining positions marked with a \textbf{?}. The number of balls, $n_\ell$ , hashing to level $\ell$ lies between $n/(2k)$ and $2n/k$ with probability $1-O(n^{-10})$ by a standard Chernoff bound. Moreover, the total number of bins lying in the history described so far is $s^*=O(\frac{\log n}{f})$ with probability $1-O(n^{-10})$, by Lemma~\ref{thm:worstcase} including those bins landing before $b_0$ in the description. Now conditioning on this history, for each $1\leq j \leq s$
$$
\E[X_j]\leq \frac{n_k}{m-s^*}=O(C/k).
$$
It follows that 
$$
\E\left[ \sum_{j=1}^{s}X_j\right]=O(sC/k)=O(C/(fk))=O(C/(\eps k)).
$$
If, we choose $k=O(1/\eps^2)$ sufficiently large, it in particular follows that $\E\left[ \sum_{j=1}^{s}X_j\right]\leq \lceil C\eps/2\rceil/20$. Thus, by Markov's inequality,
\begin{align}\label{eq:boundonmiddle}
\Pr\left[ \sum_{j=1}^{s}X_j\geq \lceil C\eps/2\rceil/2\right]\leq 1/10.
\end{align}
Next, we show that $\Pr[X^*\geq \lceil C\eps/2\rceil/2]\leq 1/10$. From this it will follow that, $\Pr[\mathcal{O}]\leq 1/5$ which is what we need. For bounding this probability, we may use Lemma~\ref{lemma:directballs}. To get into the setting of that lemma, we may simply contract the interval of the cyclic order from the most backwards \textbf{?} to $b_a$ and remove all unresolved $\textbf{?}$ in between except for the most backwards one. That the other places marked with \textbf{?} now cannot receive any balls only increases the probability that $X^*\geq t$ for any $t$. Now we are exactly in the setting of Lemma~\ref{lemma:directballs}, which we apply with $\ell=1/2$ to  conclude that if $k=O(1/\eps^2)$ is sufficiently large, then
$$
\Pr[X^*\geq \lceil C\eps/2\rceil/2]\leq 1/10.
$$
The reader may note that as an alternative to the reduction above (contracting the so far described history of how the bins $b_0,\dots,b_{a-1}$ received their balls), we may simply reprove Lemma~\ref{lemma:directballs} in this a tiny bit more complicated setting. The arguments would remain exactly the same. 

In conclusion, we have now argued that $\prbcond{\mathcal{O}}{\mathcal{A}_{i - 1}} \leq 1/5$.

\paragraph{Bounding $\prb{\mathcal{A}_{i - 1} \wedge \mathcal{O}^c \wedge \mathcal{N}}$:}
We start by defining notation which we used in \Cref{sec:f-concentration}.
Let $Y_{i}$ be the number of balls which land in bin $b$ or which are forwarded by bin $b$ on level $i$.
We define $Y_{< \ell} = \sum_{i < \ell} Y_i$ and $Y_{> \ell} = \sum_{i > \ell} Y_i$.
With this notation we get that 
\[
    \prb{\mathcal{A}_{i - 1} \wedge \mathcal{O}^c \wedge \mathcal{N}}
        = \prb{\mathcal{A}_{i - 1} \wedge \mathcal{O}^c \wedge Y_{< \ell} + Y_{> \ell} \ge \frac{n}{m}(1+\lceil C\eps/2\rceil)}
\]

We let let $\mathcal{L}_{\ell}$ be the sigma-algebra generated by the random choices on the first $\ell$ levels,
and $A_d$ will be the event as defined in \Cref{sec:f-concentration}.

We recall the simpler system from \Cref{sec:f-concentration} which we will compare to.
Let $\mathcal{Y}_{i}$ be the number of balls which land in bin $b$ or which are forwarded by bin $b$ on level $i$ in the simpler system.
We similarly define $\mathcal{Y}_{< \ell} = \sum_{i < \ell} \mathcal{Y}_i$ and $\mathcal{Y}_{> \ell} = \sum_{i > \ell} \mathcal{Y}_i$.

We will prove that,
\begin{align}\label{eq:contrib-other-levels}
    &\prb{\mathcal{A}_{i - 1} \wedge \mathcal{O}^c \wedge Y_{< \ell} + Y_{> \ell} \ge \frac{n}{m}(1+\lceil C\eps/2\rceil)}
        \\&\qquad\qquad\qquad\le \prb{\mathcal{A}_{i - 1} \wedge \mathcal{O}^c} \prb{ \mathcal{Y}_{< \ell} + \mathcal{Y}_{> \ell} \ge \frac{n}{m}(1+\lceil C\eps/2\rceil)} + m^{-1/2 + o(1)}
\end{align}
This will imply the result since 
\begin{align*}
    \prb{ \mathcal{Y}_{< \ell} + \mathcal{Y}_{> \ell} \ge \frac{n}{m}(1+\lceil C\eps/2\rceil)}
        \le \prb{\sum_{i = 1}^k \mathcal{Y}_{i} \ge \frac{n}{m}(1+\lceil C\eps/2\rceil)}
\end{align*}
Now using \Cref{eq:concentration-on-full} we get that $\prb{\sum_{i = 1}^k \mathcal{Y}_{i} \ge \frac{n}{m}(1+\lceil C\eps/2\rceil)} \le \prb{\sum_{i = 1}^k Y_{i} \ge \frac{n}{m}(1+\lceil C\eps/2\rceil)} + m^{-1/2 + o(1)}$,
and the discussion at the start of \Cref{sec:insertion-prelim} give us that $\prb{\sum_{i = 1}^k Y_{i} \ge \frac{n}{m}(1+\lceil C\eps/2\rceil)} \le 1 - c_0 f$.
Thus we just need to prove \cref{eq:contrib-other-levels}.

We start by noticing that,
\begin{align*}
    &\prb{\mathcal{A}_{i - 1} \wedge \mathcal{O}^c \wedge Y_{< } + Y_{> \ell} \ge \frac{n}{m}(1+\lceil C\eps/2\rceil)}
        \\&\qquad\qquad= \sum_{s = 0}^{\frac{n}{m}(1+\lceil C\eps/2\rceil) - 1} \prb{\mathcal{A}_{i - 1} \wedge \mathcal{O}^c \wedge Y_{< \ell} = s \wedge Y_{> \ell} \ge \frac{n}{m}(1+\ceil{ C\eps/2}) - s}
        \\&\qquad\qquad\qquad\qquad+ \prb{\mathcal{A}_{i - 1} \wedge \mathcal{O}^c \wedge Y_{< \ell} \ge \frac{n}{m}(1+\lceil C\eps/2\rceil)}
\end{align*}
We fix $0 \le s \le \frac{n}{m}(1+\lceil C\eps/2\rceil) - 1$ and get that,
\begin{align*}
    &\prb{\mathcal{A}_{i - 1} \wedge \mathcal{O}^c \wedge Y_{< \ell} = s \wedge Y_{> \ell} \ge \frac{n}{m}(1+\ceil{ C\eps/2}) - s}
        \\&\qquad\qquad\qquad= \ep{\indicator{\mathcal{A}_{i - 1} \wedge \mathcal{O}^c \wedge Y_{< \ell} = s} \prbcond{Y_{> \ell} \ge \frac{n}{m}(1+\ceil{ C\eps/2}) - s}{\mathcal{L}_l}}
\end{align*}
Now we use \Cref{lem:contribution-last-levels} and get that $\prbcond{Y_{> \ell} \ge \frac{n}{m}(1+\ceil{ C\eps/2}) - s}{\mathcal{L}_l} \le \prb{\mathcal{Y}_{> \ell} \ge \frac{n}{m}(1+\ceil{ C\eps/2}) - s} + k \indicator{A_{\ell}^c} + (1 + 2k) m^{-1/2 + o(1)}$.
Using this we get that,
\begin{align*}
    &\prb{\mathcal{A}_{i - 1} \wedge \mathcal{O}^c \wedge Y_{< \ell} + Y_{> \ell} \ge \frac{n}{m}(1+\lceil C\eps/2\rceil)}
        \\&\qquad\qquad\le \prb{\mathcal{A}_{i - 1} \wedge \mathcal{O}^c \wedge Y_{< \ell} + \mathcal{Y}_{> \ell} \ge \frac{n}{m}(1+\lceil C\eps/2\rceil)}
            \\&\qquad\qquad\qquad+ \sum_{s = 0}^{\frac{n}{m}(1+\lceil C\eps/2\rceil) - 1} \ep{\indicator{\mathcal{A}_{i - 1} \wedge \mathcal{O}^c \wedge Y_{< \ell} = s} \left( k \indicator{A_{\ell}^c} + (1 + 2k) m^{-1/2 + o(1)} \right) }
\end{align*}
Now we note that,
\begin{align*}
    \sum_{s = 0}^{\frac{n}{m}(1+\lceil C\eps/2\rceil) - 1} \ep{\indicator{\mathcal{A}_{i - 1} \wedge \mathcal{O}^c \wedge Y_{< \ell} = s} \left( k \indicator{A_{\ell}^c} + (1 + 2k) m^{-1/2 + o(1)} \right) }
        &\le k \prb{A_{\ell}^c} + (1 + 2k)m^{-1/2 + o(1)}
        \\&\le k m^{-\gamma} + (1 + 2k)m^{-1/2 + o(1)}
        \\&\le m^{-1/2 + o(1)}
\end{align*}
The second last inequality uses \Cref{eq:concentration-on-full} and last uses that $k = m^{o(1)}$.

We also want to also exchange $Y_{< \ell}$ with $\mathcal{Y}_{< \ell}$ and we will do this in similar fashion.
\begin{align*}
    &\prb{\mathcal{A}_{i - 1} \wedge \mathcal{O}^c \wedge Y_{< \ell} + \mathcal{Y}_{> \ell} \ge \frac{n}{m}(1+\lceil C\eps/2\rceil)}
        \\&\qquad\qquad\qquad= \sum_{s = 0}^{\frac{n}{m}(1+\lceil C\eps/2\rceil) - 1} \prb{\mathcal{Y}_{> \ell} = s} \prb{\mathcal{A}_{i - 1} \wedge \mathcal{O}^c \wedge Y_{< \ell} \ge \frac{n}{m}(1+\ceil{ C\eps/2}) - s}
        \\&\qquad\qquad\qquad\qquad\qquad\qquad+ \prb{\mathcal{A}_{i - 1} \wedge \mathcal{O}^c \wedge \mathcal{Y}_{> \ell} \ge \frac{n}{m}(1+\lceil C\eps/2\rceil)}
\end{align*}
Again we fix $s$ and get that,
\begin{align*}
    \prb{\mathcal{A}_{i - 1} \wedge \mathcal{O}^c \wedge Y_{< \ell} \ge \frac{n}{m}(1+\ceil{ C\eps/2}) - s}
        &= \prb{\mathcal{A}_{i - 1} \wedge \mathcal{O}^c} \prbcond{Y_{< \ell} \ge \frac{n}{m}(1+\ceil{ C\eps/2}) - s}{\mathcal{A}_{i - 1} \wedge \mathcal{O}^c}
        \\&\le \prb{\mathcal{A}_{i - 1} \wedge \mathcal{O}^c} \prbcond{Y_{< \ell} \ge \frac{n}{m}(1+\ceil{ C\eps/2}) - s}{\mathcal{A}_{i - 1} \wedge \mathcal{O}^c \wedge A_{\ell - 1}}
            \\&\qquad\qquad+ \prb{A_{\ell - 1}^c}
\end{align*}
By \Cref{eq:concentration-on-full} we know that $\prb{A_{\ell - 1}^c} \le m^{-\gamma}$.
Now similarly to $Y_{< \ell}$ we define $Y_{< \ell}^{(j)}$ to be the number of balls which lands in $j$ or which are forwarded by bin $j$ on levels before level $\ell$.
We know that $b$ is chosen uniformly from the set $[m] \setminus \set{b_0, \ldots, b_{a - 1}}$ so if we fix the first $\ell - 1$ then the probability that $Y_{< \ell} \ge \frac{n}{m}(1+\ceil{ C\eps/2}) - s$ is equal to
\begin{align*}
    \frac{\sum_{j \in [m] \setminus \set{b_0, \ldots, b_{a - 1}}} \indicator{Y_{< \ell}^{(j)} \ge \frac{n}{m}(1+\ceil{ C\eps/2}) - s} }{m - a}
\end{align*}
Since we condition on $A_{\ell - 1}$ then we have that,
\begin{align*}
    \abs{\frac{\sum_{j \in [m]} \indicator{Y_{< \ell}^{(j)} \ge \frac{n}{m}(1+\ceil{ C\eps/2}) - s} }{m} - \prb{\mathcal{Y}_{< \ell} \ge \frac{n}{m}(1+\ceil{ C\eps/2}) - s}} \le m^{-1/2 + o(1)}
\end{align*}
This implies that,
\begin{align*}
    \abs{\frac{\sum_{j \in [m] \setminus \set{b_0, \ldots, b_{a - 1}}} \indicator{Y_{< \ell}^{(j)} \ge \frac{n}{m}(1+\ceil{ C\eps/2}) - s} }{m - a} - \prb{\mathcal{Y}_{< \ell} \ge \frac{n}{m}(1+\ceil{ C\eps/2}) - s}}
        \le \frac{m}{m - a}  m^{-1/2 + o(1)} + \frac{a}{m - a}
\end{align*}
Now we use \Cref{thm:worstcase1} to get that $a \le O(\log(m)/\eps) = m^{o(1)}$ with probability $1 - m^{-\gamma}$.
Here we use that $1/\eps = m^{o(1)}$.
Combining this we get that,
\begin{align*}
    &\prbcond{Y_{< \ell} \ge \frac{n}{m}(1+\ceil{ C\eps/2}) - s}{\mathcal{A}_{i - 1} \wedge \mathcal{O}^c \wedge A_{\ell - 1}} 
        \\&\qquad\qquad\le \prb{\mathcal{Y}_{< \ell} \ge \frac{n}{m}(1+\ceil{ C\eps/2}) - s} + \frac{m}{m - m^{o(1)}}  m^{-1/2 + o(1)} + \frac{m^{o(1)}}{m - m^{o(1)}} + m^{-\gamma}
        \\&\qquad\qquad\le \prb{\mathcal{Y}_{< \ell} \ge \frac{n}{m}(1+\ceil{ C\eps/2}) - s} +  m^{-1/2 + o(1)}
\end{align*}
We then get that,
\begin{align*}
    \prb{\mathcal{A}_{i - 1} \wedge \mathcal{O}^c \wedge Y_{< \ell} \ge \frac{n}{m}(1+\ceil{ C\eps/2}) - s}
        &= \prb{\mathcal{A}_{i - 1} \wedge \mathcal{O}^c} \prbcond{Y_{< \ell} \ge \frac{n}{m}(1+\ceil{ C\eps/2}) - s}{\mathcal{A}_{i - 1} \wedge \mathcal{O}^c}
        \\&\le \prb{\mathcal{A}_{i - 1} \wedge \mathcal{O}^c} \left( \prb{\mathcal{Y}_{< \ell} \ge \frac{n}{m}(1+\ceil{ C\eps/2}) - s} +  m^{-1/2 + o(1)} \right)
            \\&\qquad\qquad+ m^{-\gamma}
        \\&\le \prb{\mathcal{A}_{i - 1} \wedge \mathcal{O}^c} \prb{\mathcal{Y}_{< \ell} \ge \frac{n}{m}(1+\ceil{ C\eps/2}) - s} +  m^{-1/2 + o(1)}
\end{align*}
Using this we get that,
\begin{align*}
    &\prb{\mathcal{A}_{i - 1} \wedge \mathcal{O}^c \wedge Y_{< \ell} + \mathcal{Y}_{> \ell} \ge \frac{n}{m}(1+\lceil C\eps/2\rceil)}
        \\&\qquad\qquad\le \prb{\mathcal{A}_{i - 1} \wedge \mathcal{O}^c} \prb{\mathcal{Y}_{< \ell} + \mathcal{Y}_{> \ell} \ge \frac{n}{m}(1+\lceil C\eps/2\rceil)}
            \\&\qquad\qquad\qquad+ \sum_{s = 0}^{\frac{n}{m}(1+\lceil C\eps/2\rceil) - 1} 
                \prb{\mathcal{Y}_{> \ell} = s} m^{-1/2 + o(1)}
        \\&\qquad\qquad\le \prb{\mathcal{A}_{i - 1} \wedge \mathcal{O}^c} \prb{\mathcal{Y}_{< \ell} + \mathcal{Y}_{> \ell} \ge \frac{n}{m}(1+\lceil C\eps/2\rceil)} + m^{-1/2 + o(1)}
\end{align*}
This finishes the proof \cref{eq:contrib-other-levels}.

This concludes the proof that equation~\eqref{eq:tool1} of the claim holds. 

\subsection*{Proof of Equation~\eqref{eq:tool2}}
We restate what we have to prove, namely that
$$
\mathbb{E}[R \mid \mathcal{E}^c]= O(1/f),
$$
Where $R$ is the number of bins visited during step 2. of epoch $t+1$ and $\mathcal{E}^c$ is the event that we did not finish the insertion during step 1. of epoch $t+1$. Let $b_0,\dots,b_a$ denote the  bins that we have visited so far at the level where we are currently at, call it $\ell$. All bins $b_0,\dots,b_a$ got filled from levels $1,\dots, \ell$, and as in the proof of equation \eqref{eq:tool1} of the claim, we may again describe the history of how the bins $b_0,\dots,b_a$ got filled to their capacity at level $\ell$. See Figure~\ref{figure:generatingorder3} for an example of such a history.

\begin{figure}[h]
\centering
\includegraphics{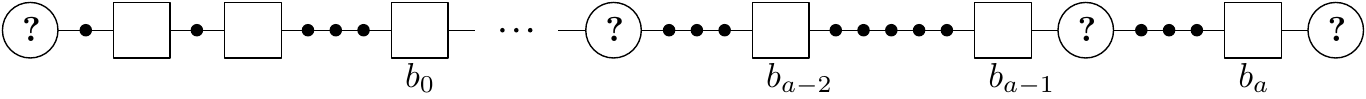}
\caption{An example of how the conditioning on $\mathcal{E}^c$ might look. Except for the \textbf{?} coming before $b_0$, the circled \textbf{?}'s, are parts of the cyclic order which has not yet been fixed, but which are known to consist solely of balls. The circled \textbf{?} appearing before $b_0$, which is the yet unknown history of how many balls $b_0$ are to further forward, has a special role. Indeed, this history does not have to consist solely of balls but can consist of a run of balls and bins such that all of the bins in the run gets filled at this level.}
\label{figure:generatingorder3}
\end{figure}
Let $s\geq1/f$. We wish to argue that the conditional probability 
\begin{align}\label{eq:cond-run-length}
\Pr[R\geq s \mid \mathcal{E}^c]=O(\exp(-\Omega(sf)))+O(n^{-10}).
\end{align}
Ignoring the unimportant $O(n^{-10})$ term, it will follow that 
$$
\mathbb{E}[R \mid \mathcal{E}^c]\leq 1/f+\sum_{i=0}^\infty \Pr[2^i /f  \leq R \leq 2^{i+1}/f]2^{i+1}/f=1/f+O\left( \frac{1}{f}\sum_{i=0}^\infty \exp(-\Omega(2^i))2^i\right)=O(1/f).
$$
and including the $O(n^{-10})$ term in the computation could only increase the bound with an additive $n^{-8}$, say, as we can here use the trivial bound on the length of a run of $mk$. Thus, this yields the desired result. For the bound on $\Pr[s\leq R\leq 2s \mid \mathcal{E}^c]$, it clearly suffices to assume that $s\geq c/f$ where $c=O(1)$ is a sufficiently large constant. 

We start by noting that with probability $1-O(n^{-10})$, the number of balls hashing to level $\ell$ is at most $2n/k$ which we assume to be the case in what follows. Let $q$ denote the number of places marked with \textbf{?} between $b_0$ and $b_a$ and let $X_1,\dots,X_q$ denote the number of balls landing at these positions. Then $q\leq a\leq 1/f$. Let $\alpha =\lceil C\eps/2 \rceil $. Let $A$ denote the event that $X_1+\dots+X_q \geq \lambda sf\alpha$, where $\lambda=\Omega(1)$ is a sufficiently small constant to be chosen later. We start by providing an upper bound on $\Pr[A]$. For this, we let $X=\sum_{i=1}^q X_i$ and note, like in the proof of Lemma~\ref{lemma:directballs}, that for each $i$, $X_i$ is dominated by a geometric variable with parameter $q$ where $q=\frac{n_\ell}{m+n_\ell}$. Here $n_\ell=2n/k$ is the upper bound on the number of ball hashing to level $\ell$. Furthermore, this claim holds even conditioning on the values of $(X_j)_{j< i}$. Let $s'=s\lambda$. Letting $Y_1,\dots,Y_{1/f}$ be independent such geometric variables and $Y=\sum_{i=1}^{1/f} Y_i$, we can thus upper bound 
$$
\Pr[A]\leq \Pr[Y\geq s'f\alpha].
$$
Note that 
$$
\E[Y_i]\leq \frac{n_{\ell}}{m}\leq \frac{2C}{k}\leq \frac{4\alpha}{\eps k}
$$
for $1\leq i\leq 1/f$, so that $\E[Y]\leq \frac{4\alpha}{\eps f k}\leq \alpha$, were the last inequality follows by assuming that $k=O(1/\eps^2)$ is sufficiently large. We may also assume that $s'f$ is larger than a sufficiently large constant, as described above, so we can upper bound 
$$
\Pr[A]\leq \Pr[Y\geq \E[Y]+s'f\alpha/2].
$$ 
By applying the bound~\cref{eq:geotail} of~\Cref{thm:geotail} similarly to how we did in the proof of~\Cref{lemma:directballs} it follows after some calculations that 
$$
\Pr[A]=\exp(-\Omega(sf)).
$$
Now condition on $A^c$ and let us focus on upper bounding $\Pr[R\geq  s \mid \mathcal{E}^c\cap A^c]$. For this, we apply~\eqref{thm:worstcase}. To get into the setting of that theorem, we contract the part of the history revealed so far between the back-most \textbf{?}-mark before $b_0$ and up til and including $b_a$ into a single \emph{unified} bin. By the conditioning on $A^c$, this unified bin comes with an extra start load of at most $\lambda sf\alpha$ balls, where we can choose $\lambda=\Omega(1)$ to be any sufficiently small constant. Thus, with the conditioning, we are exactly in the setting to apply~\Cref{thm:worstcase}, and we may thus bound
$$
\Pr[R\geq  s \mid \mathcal{E}^c\cap A^c]=\exp(-\Omega(sf)).
$$
It follows that 
$$
\Pr[R\geq  s \mid \mathcal{E}^c]\leq \Pr[A]+\Pr[R\geq  s \mid \mathcal{E}^c\cap A^c]=\exp(-\Omega(sf)),
$$
which is the desired. This completes the proof of~\Cref{claim:main}.
\end{proof}
As explained before the proof of~\Cref{claim:main}, this completes the proof of our theorem.
\end{proof}

\section{Insertions of Bins and Deletions of Balls and Bins}\label{sec:focus-on-insertions}
In this section, we prove the statements of~\Cref{thm:main1,thm:main} concerning the deletions of balls and insertions and deletions of bins. Combined with the results of~\Cref{sec:simple-improvement,sec:insertion,sec:faster-search}, this proves the two theorems in full. 

\paragraph{Deletions of Balls.} By the history independence, a deletion of a ball is symmetric to an insertion. The bins visited when deleting a ball $x$ are the same as the bins visited if $x$ had not been in the system and was inserted. Thus, we can upper bound the expected number of bins visited when deleting a ball by $O(1/\eps)$ for~\Cref{thm:main1} and $O(1/f)$ for~\Cref{thm:main}. This also upper bounds the number of balls moved in a deletions.

\paragraph{Deletions of Bins.} A deletion of a super bin is the same as reinserting the balls lying in that super bin. We claimed that  that the expected cost of deleting a super bin is $O(C/f)$ in~\Cref{thm:main}.  At first, this may seem completely obvious, since the cost of inserting a single ball is $O(1/f)$. However, this cost is for inserting a ball which is selected independently of the random choices of the hash function. Now, we are looking at the balls placed in a given super bin $b$, and those are highly dependent on the hash function. However, we do know that the expected average cost of all balls in the system is $O(1/f)$. Moreover, all bins are symmetric, so the bin $b$ behaves like a random bin amongst those in the system. Thanks to our load balancing, the balls are almost uniformly spread between the bins, so a random ball from a random bin is almost a uniformly random ball, so a random ball from $b$ has expected cost $O(1/f)$. There are at most $C$ balls in $b$ them, so the total expected cost is $O(C/f)$. A similar argument applies in the case of~\Cref{thm:main1}.

\paragraph{Insertions of Bins}
Again, by the history independence an insertion of a bin is symmetric to its deletion. The balls that are moved when inserting a bin are thus the same as if that bin was in the system but was deleted. Thus we can use the result for deletions of bins to conclude the bound of $O(C/f)$ on the number of balls moved when inserting a bin. A similar argument applies in the case of~\Cref{thm:main1}.

\section{Faster Searches Using the Level-Induced Priorities}\label{sec:faster-search}
In this section we make the calculation demonstrating that giving the balls random priorities, we obtain the better bounds on the number of bins visited during an insertion as claimed in~\Cref{thm:main1,thm:main}. This is not a new idea but is in fact an old trick~\cite{AK74:ordered-hash-table,knuth-vol3}. What we need to do is verify that applying it, with the particular formula for $f$ in~\cref{eq:formula-for-f}, we obtain the stated search times. In fact, what we require for the analysis is only the fact that if two balls hash to different levels, the ball hashing to the lower level has the highest priority of the two. Within a given level, the priorities can be arbitrary. This is important for the practical version of our scheme described in~\Cref{sec:implementation2} where the priorities are not uniformly random and independent of the hashing of balls, but where the hashing of the balls in fact determines the priorities, with higher hash values implying lower priorities. We start by arguing about the expected number of bins visited during a search as stated in~\Cref{thm:main1}.

\paragraph{Number of Bins Visited During a Search:~\Cref{thm:main1}.}
We encourage the reader to recall the setting described in the theorem. Define $X$ to be the number of bins visited during the search for some ball $x$. Importantly, if $x$ hashes to level $i$, then all virtual bins visited during the search of $x$ also lie on level $i$. For $i\in [k]$, we let $A_i$ denote the event that $x$ hashes to level $i$, so that $\Pr[A_i]=p_i$. By a standard Chernoff bound, the number of balls hashing to the first $i$ levels is $np_{\leq i}\pm m^{1/2+o(1)}=np_{\leq i}(1\pm m^{-1/2+o(1)})$, with probability $1-O(n^{-10})$, say. Here we used that $1/\eps=m^{o(1)}$. Condition on this event and define $n_{<i}$ to be the number of balls hashing to the first $i$ levels.  Finally letting $\eps_i$ be such that that $(1+\eps_i)n_{\leq i}/m=C$, we obtain from the part of~\Cref{thm:main1} concerning insertions (which was proved in~\Cref{sec:simple-improvement}) that $\E[X_i\mid A_i]=O\left(1/\eps_i\right)$. Moreover, $\Pr[A_i]\leq 2^{-i+1}$ for each $i$. It finally follows from the Chernoff bound above that $\eps_i\geq 1/2^i$, and so 
$$
\E[X]=\sum_{i\in [k]} \E[X \mid A_i]\Pr[A_i]=O(k)=O(\log 1/\eps)
$$
as desired.

\paragraph{Number of Bins Visited During a Search: \Cref{thm:main}.}
We now perform a similar calculation to the one above, in the more complicated setting of~\Cref{thm:main}. 
Let us for simplicity assume that the number of balls hashing to each level is exactly $n/k$. It is trivial to later remove this assumption. We also assume for simplicity that $k\geq 1/\eps^2$ is a power of $2$, $k=2^a$ for some $a$. Let $\ell=\lceil  \log (1/\eps)\rceil$ noting that $\ell \leq a$. We partition $[k]=I_0\cup \cdots I_\ell$, where $I_i=[2^a-2^{a-i-1}]\setminus [2^a-2^{a-i}]$ for $0\leq i\leq \ell-1$ and $I_\ell=[2^a] \setminus [2^a-2^{a-\ell}]$. Let $A_i$ be the event that the given ball to be searched $x$ hashes to some level in $I_i$, so that $\Pr[A_i]=2^{-i+1}$ for $0\leq i\leq \ell-1$ and $\Pr[A_\ell]=2^{-\ell}$. For $0\leq i\leq \ell$ we define $n_{\leq i}$ to be the number of balls hashing to some level in $I_0\cup\cdots\cup I_i$. Finally, let $\eps_i$ be such that $(1+\eps_i)n_{\leq i}/m=C$ and note that $\Pr[A_i]=\Theta(\eps_i)$.

We partition $[\ell+1]$ into three sets, $[\ell+1]=J_1\cup J_2\cup J_3$ where
$$
J_1=\{i\in [\ell+1]: C\leq \log 1/\eps_i\}, \quad I_2=\{i\in [\ell+1]: \log 1/\eps_i<C\leq \frac{1}{2\eps_i^2}\}, \quad \text{and} \quad I_3=\{i\in [\ell+1]: \frac{1}{2\eps_i^2}<C\}.
$$
It then follows from the part of~\Cref{thm:main} dealing with insertions (proved in~\Cref{sec:insertion}) that
$$
\E[X]=O\left(\sum_{i\in I_1} \frac{\Pr[A_i]}{\eps_i C}+\sum_{i\in I_2}\frac{\Pr[A_i]}{\eps_i \sqrt{C\log \left(\tfrac{1}{\eps_i\sqrt{C}}\right)}} +\sum_{i\in I_3} \Pr[A_i]\right)=O\left(1+\frac{|I_1|}{C}+\sum_{i\in I_2}\frac{1}{\sqrt{C\log \left(\tfrac{1}{\eps_i\sqrt{C}}\right)}}\right)
$$
We have the trivial bound $|I_1|\leq \ell+1 =O(\log 1/\eps)$. Moreover, for $i\in I_2$, it holds that
$$
e^{-C}\leq \eps_i\leq \sqrt{\frac{1}{2C}},
$$
and since $ \eps_i=\Theta( 2^{-i})$, it follows that 
$$
\sum_{i\in I_2}\frac{1}{\sqrt{C\log \left(\tfrac{1}{\eps_i\sqrt{C}}\right)}}=O\left(\frac{1}{\sqrt{C}}\sum_{i=1}^{O(C)}\frac{1}{\sqrt{i}}\right)=O(1).
$$ 
In conclusion, 
$$
\E[X]=O\left(1+\frac{\log 1/\eps}{C}\right),
$$
and splitting into the cases, $C\leq \log 1/\eps$ and $C<\log 1/\eps$, we obtain the desired result.

\section{The Practical Implementation.}\label{sec:practical-imp}
In this section we sketch why our results continue to holds when using the practical implementation described in~\Cref{sec:implementation2} even when the hashing is implemented using the practical mixed tabulation scheme from~\cite{DKRT15:k-part}.
Let us call the implementation from~\Cref{sec:implementation2} the \emph{practical implementation}.

We first discuss the practical implementation with fully random hashing. For this, recall the definition of a run (\Cref{def:run}).
Using a similar argumentation to the one used in the proof of~\Cref{thm:worstcase1}, it is easy to show that in this implementation, for any constant $\gamma=O(1)$, the maximal number of bins in a run is $O((\log n)/\eps)$ with probability $1-n^{-\gamma}$. Denote this high probability event $\mathcal{E}$. The number of balls lying in a run consisting of $\ell$ bins is trivially upper bounded by $C(\ell+1)$, so if $\mathcal{E}$ occurs, the maximal number of balls hashing to a fixed run is $O(C(\log n)/\eps)$. It follows that the number of balls that are forwarded past any given point is $O(C(\log n)/\eps)$. In particular for any level $i$, the number of balls that are forwarded from level $i$ to level $i+1$ is $O(C(\log n)/\eps)$ and the total number of such balls over all levels is $O(kC(\log n)/\eps)=m^{o(1)}$. One can now modify our inductive proof of~\Cref{thm:jakobs} to check that its statement remains valid even with the influence of these extra balls. Recall that in~\Cref{thm:jakobs}, $X_{i,j}$ denoted the number of bins with at most $j$ balls after the hashing of balls to levels $0,\dots,i-1$. Intuitively, in the inductive step, these $m^{o(1)}$ extra balls can only affect $m^{o(1)}$ bins which does not affect the high probability bound stating that $|X_{i,j}-\mu_{i,j}|\leq m^{1/2+o(1)}$. To exclude the bad event $\mathcal{E}^c$, we simply use a union bound and that $\mathcal{E}$ happened with very high probability. Once we have a version of~\Cref{thm:jakobs} which holds in the practical implementation, we can repeat the proof of~\Cref{thm:main-insertion1}, again using union bounds for the event that the insertion interacts with the run of size $m^{o(1)}$ entering the given level from below.  

Let us now discuss the implementation with mixed tabulation. A mixed tabulation hash function $h$ is defined using two  of the simple tabulation hash functions from~\cite{patrascu11charhash}, $h_1:\Sigma^c\to \Sigma^d$ and $h_2:\Sigma^{c+d} \to R$. Here $\Sigma $ is some character alphabet with $\Sigma^c=[u]$ and $c,d=O(1)$ are constants. Then for a key $x$, $h(x)=h_2(x,h_1(x))$. An important property of mixed tabulation, proved in~\cite{DKRT15:k-part}, is the following: Suppose $X$ is a set of keys, $p_1,\dots,p_b$ are output bit positions and $v_1,\dots,v_b$ are desired bit values. Let $Y$ be the set of keys $x\in  X$ for which the $p_i$'th output bit $h(x)_{p_i}=v_i$ for all $i$. If $\E[|Y|]\leq |\Sigma|/(1+\Omega(1))$, then the remaining output bits of the hash values in $Y$ are completely independent with probability $1-O(|\Sigma|^{1-\floor{d/2}})$. Another important property is that mixed tabulation obeys the same concentration bounds as simple tabulation on the number of balls landing in an interval~\cite{patrascu11charhash}.  

For the implementation with mixed tabulation, we use $k$ independent mixed tabulation functions, $h_1,\dots,h_k$, to distribute the virtual bins, and a single mixed tabulation function $h^*$ for the balls (independent of $h_1,\dots,h_k$). We moreover assume that $|\Sigma|=u^{1/c}=n^{\Omega(1)}$ which can be achieved using a standard universe reduction. To obtain our results using mixed tabulation, the idea is essentially the same as above. Again, we first need to prove an analogue of~\Cref{thm:jakobs}, and we would do this using induction on the level, bounding $|X_{i,j}-\mu_{i,j}|$ with high probability for each level $i$. To do this, we partition level $i$ into dyadic intervals where we expect at most $|\Sigma|/2$ balls or bins to hash. Then we can use the concentration bound from~\cite{patrascu11charhash} (which also holds for mixed tabulation) to obtain concentration on the number of bins of a given capacity from the previous levels hashing to each interval. Moreover, we can use the result of~\cite{DKRT15:k-part} to conclude that restricted to such an interval the hashing of balls and bins is fully random. Again, we can prove a version of~\Cref{thm:worstcase1} with mixed tabulation (by using that mixed tabulation provides concentration bounds) and conclude that the total number of balls that are forwarded from one interval to another is $O(C(\log n)/\eps)=m^{o(1)}=|\Sigma|^{o(1)}$. Essentially, the good distribution of the $X_{i-1,j}$ ensures that we also obtain a good distribution of the number of bins with each capacity in each of the intervals of level $i$ (using that the influence of the $|\Sigma|^{o(1)}$ balls passing between intervals can only affect $|\Sigma|^{o(1)}$ bins), and this gives a good distribution of the $X_{i,j}$. For this, it is important to be aware that there are now more intervals, essentially $n/|\Sigma|$, but since $|\Sigma|=n^{\Omega (1)}$, we still obtain that the total number of balls that are forwarded from one interval to another is $n^{1-\Omega(1)}$. The high probability bound we obtain on $|X_{i,j}-\mu_{i,j}|$  then instead takes the form $|X_{i,j}-\mu_{i,j}|=n^{1-\Omega(1)}$, but this still suffices for our purposes. Finally, we may prove a mixed tabulation version of~\Cref{thm:main-insertion1}, again using the fully random hashing within each interval and using union bounds to bound away the probability that we interact with the $|\Sigma|^{o(1)}$ balls that are forwarded between intervals. As such, showing that our results hold using mixed tabulation uses essentially the same ideas as is needed to show that the implementation in~\Cref{sec:implementation2} does, but with a finer partitioning into intervals. 

\section{Modifying the Analysis for Dynamically Changing Capacities}\label{sec:dyn-cap-analysis}
In this last short section, we describe how our analysis can still be carried through even with the dynamically changing capacities described in~\Cref{sec:dyn-load-cap}.  In the preceding sections, we assumed that the capacities of the bins were all equal to some integer $C$. However, in the setting of~\Cref{sec:dyn-load-cap} we are interested in the case where the \emph{total} capacity is $Cm$, with $m_1$ bins of capacity $\floor{C}$ and $m_2$ bins of capacity $\ceil{C}$. Thus $C$ is no longer assumed to be integral. This corresponds to all bins having the same capacity $\ceil{C}$, but where we include an extra $-1$'th level, where $m_1$ bins each receive a single artificial ball.

To analyse this new setting one can first observe that the proofs in~\Cref{sec:insertion} carry through without significant changes. Thus it is mainly in regards to the bounds on the fraction of non-full bins in~\Cref{sec:f-properties} that there is something to discuss. Recall, that we showed in~\Cref{sec:f-concentration} that the contribution of balls from the levels to a given random bin essentially behaves like a sum of geometric variables. With the terminology introduced in~\cite{aamand2021sums}, geometric variables are strongly monotone, and we could then apply the bound of that paper to estimate the point probabilities of this sum. Now Bernoulli variables are also strongly monotone, and so the bound in~\cite{aamand2021sums} can also be applied when some of the variables in the sum are Bernoulli. Now with the $-1$'th level described above, the number of balls landing in a random bin at the new lowest level is Bernoulli. Then the contribution to a random bin is essentially a sum of geometric variables and a single Bernoulli variable, and since the bound in~\cite{aamand2021sums} holds for such a sum, we can still use it for estimating the point probabilities of the number of balls in a bin. The remaining parts of the proof carries through almost unchanged.

\section*{Acknowledgement}
The authors wish to thank Noga Alon and Nick Wormald for helpful discussions. With Noga Alon, we studied sums of integer variables \cite{aamand2021sums}, including bounds needed for the analysis of this paper. In unpublished work on a the random graph $d$-process, Nick Wormald and Andrzej Ruci\'{n}ski also used the idea of analyzing balls in capacitated bins by throwing an appropriately larger number of balls into uncapacitated bins. They did not present an estimate on the number of non-full bins, as needed for this paper.

Research supported by grant 16582, Basic Algorithms Research Copenhagen (BARC), from the VILLUM Foundation.

{
\bibliographystyle{alpha} 
\bibliography{general}
}

\end{document}

%% file: Preamble.tex
\usepackage{amssymb, amsmath, amsthm}

\usepackage{float}
\usepackage[english]{babel}
\usepackage[utf8]{inputenc}
\usepackage{textcomp}
\usepackage[hidelinks]{hyperref}
\usepackage[T1]{fontenc}
\usepackage{comment}
\usepackage{todonotes}
\usepackage[capitalise]{cleveref}
\usepackage{fullpage}
\usepackage{mathtools}
\usepackage{caption}
\usepackage{multirow}
\usepackage[ruled,vlined,linesnumbered]{algorithm2e}
\usepackage{wrapfig}
\usetikzlibrary{patterns}
\usepackage{csquotes}
\usepackage{lineno}
\usepackage{tikz, fp, subfigure, epsfig, filecontents, float}

\newcommand{\eps}{\varepsilon}

\newcommand{\R}{\mathbb{R}}
\newcommand{\N}{\mathbb{N}}
\newcommand{\Z}{\mathbb{Z}}

\newcommand{\E}{\mathbb{E}}
\DeclareMathOperator*{\Var}{Var}

\newcommand{\balls}{n}
\newcommand{\bins}{m}

\newcommand{\abs}[1]{\left| {#1} \right|}

\newcommand{\ep}[2][]{\E_{{#1}}\left[ {#2} \right]}
\newcommand{\epcond}[3][]{\E_{{#1}}\left[ {#2} \, \middle| \, {#3} \right]}
\newcommand{\prb}[2][]{\Pr_{#1}\left[ #2 \right]}
\newcommand{\prbcond}[3][]{\Pr_{{#1}}\left[ {#2} \, \middle| \, {#3} \right]}
\newcommand{\ceil}[1]{\lceil {#1} \rceil}
\newcommand{\floor}[1]{\lfloor {#1} \rfloor}
\newcommand{\var}[1]{\Var\left[ {#1} \right]}
\newcommand{\indicator}[1]{\left[ {#1} \right]}

\newcommand\drop[1]{}

\newcommand{\set}[1]{\left\{ #1 \right\}}
\newcommand{\setbuilder}[2]{\set{ #1 \; \middle\vert \; #2 }}

\newtheorem{theorem}{Theorem}

\newtheorem{lemma}[theorem]{Lemma}

\newtheorem{corollary}[theorem]{Corollary}

\newtheorem{claim}{Claim}

\theoremstyle{definition}
\newtheorem{definition}{Definition}
\theoremstyle{remark}
\newtheorem*{remark}{Remark}